\documentclass[ecta,nameyear,final]{econsocart}
\RequirePackage[colorlinks,citecolor=blue,linkcolor=blue,urlcolor=blue,pagebackref]{hyperref}

\startlocaldefs

\usepackage{ragged2e}
\usepackage{graphicx}
\usepackage{float}
\usepackage{amsthm}
\usepackage{amsmath} 
\usepackage{amsfonts}
\usepackage{mathtools}
\usepackage{dsfont} 
\usepackage{amssymb}
\usepackage{bbm}
\usepackage{bm}
\mathtoolsset{showonlyrefs}
\usepackage{subcaption}
\usepackage{letltxmacro}
\usepackage{scalerel}
\usepackage{tabularx,booktabs}
\usepackage{multirow}
\usepackage{stackrel}
\usepackage{adjustbox}
\usepackage{tikz,lipsum,lmodern}
\usepackage[most]{tcolorbox}

\theoremstyle{plain}

\newtheorem{theorem}{Theorem}

\newtheorem{lemma}{Lemma}

\newtheorem{cor}{Corollary}
\newtheorem{assumption}{Assumption}
\theoremstyle{remark}

\newtheorem{remark}{Remark}

\theoremstyle{plain}

\DeclareRobustCommand{\Rb}{\mathbb{R}}
\DeclareRobustCommand{\bD}{\boldsymbol{D}}

\DeclareRobustCommand{\bA}{\boldsymbol{A}}
\DeclareRobustCommand{\bH}{\boldsymbol{H}}

\DeclareRobustCommand{\bU}{\boldsymbol{U}}
\DeclareRobustCommand{\bV}{\boldsymbol{V}}
\DeclareRobustCommand{\bS}{\boldsymbol{S}}

\DeclareRobustCommand{\bB}{\boldsymbol{B}}

\newcommand{\bI}{\boldsymbol{I}}

\newcommand{\bY}{\boldsymbol{Y}}

\newcommand{\ba}{\boldsymbol{a}}

\DeclareRobustCommand{\bSigma}{\boldsymbol{\Sigma}}

\DeclareRobustCommand{\bhV}{\widehat{\boldsymbol{V}}}

\DeclareRobustCommand{\Rb}{\mathbb R}

\DeclareRobustCommand{\Pb}{\mathbb P}

\DeclareMathOperator{\Cov}{Cov}
\DeclareMathOperator{\Var}{Var}

\DeclareMathOperator{\tr}{tr}

\DeclareMathOperator*{\argmin}{\arg\!\min}

\newcommand{\hM}{\widehat{M}}

\newcommand{\hY}{\widehat{Y}}
\newcommand{\halpha}{\widehat{\alpha}}
\newcommand{\talpha}{\tilde{\alpha}}

\newcommand{\distas}[1]{\mathbin{\overset{#1}{\kern\z@\sim}}}%
\newsavebox{\mybox}\newsavebox{\mysim}
\newcommand{\distras}[1]{%
  \savebox{\mybox}{\hbox{\kern3pt$\scriptstyle#1$\kern3pt}}%
  \savebox{\mysim}{\hbox{$\sim$}}%
  \mathbin{\overset{#1}{\kern\z@\resizebox{\wd\mybox}{\ht\mysim}{$\sim$}}}%
}

\newcommand{\htheta}{\widehat{\theta}}
\newcommand{\tbeta}{\tilde{\beta}}

\newcommand{\hbeta}{\widehat{\beta}}

\newcommand{\Ex}{\mathbb{E}}

\newcommand{\hv}{\widehat{v}} 

\newcommand{\sdid}{\text{sdid}}

\newcommand{\SDID}{\textsf{SDID}}
\newcommand{\DID}{\textsf{DID}}
\newcommand{\ESDID}{\emph{\textsf{SDID}}}

\newcommand{\asc}{\text{asc}}

\newcommand{\EASC}{\emph{\textsf{ASC}}}
\newcommand{\ASC}{\textsf{ASC}}

\newcommand{\vt}{\text{vt}} 
\newcommand{\hz}{\text{hz}} 
\newcommand{\evt}{\emph{vt}} 
\newcommand{\ehz}{\emph{hz}}
\newcommand{\mix}{\text{mix}} 
\newcommand{\emix}{\emph{mix}}

\newcommand{\ols}{\text{ols}}

\newcommand{\bone}{\boldsymbol{1}} 
\newcommand{\bzero}{\boldsymbol{0}} 
\newcommand{\hnu}{\widehat{\nu}} 

\newcommand{\EHZ}{\emph{\textsf{HZ}}}
\newcommand{\EVT}{\emph{\textsf{VT}}}

\newcommand{\HZ}{\textsf{HZ}}
\newcommand{\VT}{\textsf{VT}}

\newcommand{\hsigma}{\widehat{\sigma}} 
\newcommand{\hvarepsilon}{\widehat{\varepsilon}} 
\newcommand{\bgamma}{\boldsymbol{\gamma}} 
\newcommand{\bhgamma}{\widehat{\bgamma}} 
\newcommand{\hgamma}{\widehat{\gamma}} 
\newcommand{\by}{\boldsymbol{y}} 
\newcommand{\bb}{\boldsymbol{b}}

\newcommand{\bhSigma}{\widehat{\bSigma}} 

\newcommand{\balpha}{\boldsymbol{\alpha}} 
\newcommand{\bbeta}{\boldsymbol{\beta}} 
 
\newcommand{\bu}{\boldsymbol{u}} 
\newcommand{\bv}{\boldsymbol{v}} 
\newcommand{\bz}{\boldsymbol{z}} 
\newcommand{\bx}{\boldsymbol{x}} 
 
\newcommand{\bDelta}{\boldsymbol{\Delta}} 
\newcommand{\bGamma}{\boldsymbol{\Gamma}} 
\newcommand{\bvarepsilon}{\boldsymbol{\varepsilon}} 
\newcommand{\bhvarepsilon}{\widehat{\bvarepsilon}} 
\newcommand{\btheta}{\boldsymbol{\theta}} 
\newcommand{\bhalpha}{\widehat{\balpha}} 
\newcommand{\bhbeta}{\widehat{\bbeta}} 
\newcommand{\bhtheta}{\widehat{\btheta}} 
 
\newcommand{\btalpha}{\tilde{\balpha}} 
\newcommand{\btbeta}{\tilde{\bbeta}} 

\newcommand{\jack}{\text{jack}}
\newcommand{\ejack}{\emph{jack}}
\newcommand{\hrk}{\text{HRK}} 
\newcommand{\ehrk}{\emph{HRK}} 
\newcommand{\homo}{\text{homo}} 
\newcommand{\ehomo}{\emph{homo}}  
\newcommand{\hzeta}{\widehat{\zeta}} 
\newcommand{\bhzeta}{\widehat{\boldsymbol{\zeta}}}

\newcommand{\bOmega}{\boldsymbol{\Omega}}

\endlocaldefs

\begin{document}

\begin{frontmatter}

\title{
Same Root Different Leaves: Time Series and Cross-Sectional Methods in Panel Data
}
\runtitle{
Time Series and Cross-Sectional Methods in Panel Data
}

\begin{aug}
%
%
%
\author[id=au1,addressref={add1}]{\fnms{Dennis}~\snm{Shen}\ead[label=e1]{dshen24@berkeley.edu}}
\author[id=au2,addressref={add2}]{\fnms{Peng}~\snm{Ding}\ead[label=e2]{pengdingpku@berkeley.edu}}
\author[id=au3,addressref={add3}]{\fnms{Jasjeet}~\snm{Sekhon}\ead[label=e3]{jasjeet.sekhon@yale.edu}}
\author[id=au4,addressref={add4}]{\fnms{Bin}~\snm{Yu}\ead[label=e4]{binyu@berkeley.edu}}
\address[id=add1]{%
\orgdiv{Simons Institute for the Theory of Computing},
\orgname{University of California, Berkeley}}

\address[id=add2]{%
\orgdiv{Department of Statistics},
\orgname{University of California, Berkeley}}

\address[id=add3]{%
\orgdiv{Departments of Statistics \& Data Science and Political Science},
\orgname{Yale University}}

\address[id=add4]{%
\orgdiv{Departments of Statistics and EECS},
\orgname{University of California, Berkeley}}
\end{aug}

\support{
We sincerely thank Alberto Abadie, Avi Feller, Guido Imbens, and Devavrat Shah for their thoughtful comments and insightful feedback. 
We gratefully acknowledge support from NSF grants 1945136, 1953191, 2022448, 2023505 on Collaborative Research: Foundations of Data Science Institute (FODSI), and ONR grant N00014-17-1-2176.
The data and code to reproduce the results in this article are available at \href{https://github.com/deshen24/panel-data-regressions}{https://github.com/deshen24/panel-data-regressions}.}
%

\begin{abstract}
A central goal in social science is to evaluate the causal effect of a policy. 
One dominant approach is through panel data analysis in which the behaviors of multiple units are observed over time. 
The information across time and space motivates two general approaches: 
(i) horizontal regression (i.e., unconfoundedness), which exploits time series patterns, 
and (ii) vertical regression (e.g., synthetic controls), which exploits cross-sectional patterns. 
Conventional wisdom states that the two approaches are fundamentally different. 
We establish this position to be partly false for estimation but generally true for inference.
In particular, we prove that both approaches yield identical point estimates under several standard settings.
For the same point estimate, however, each approach quantifies uncertainty with respect to a distinct estimand. 
The confidence interval developed for one estimand may have incorrect coverage for another. 
This emphasizes that the source of randomness that researchers assume has direct implications for the accuracy of inference. 
\end{abstract} 


\end{frontmatter}

\section{Introduction} \label{sec:intro} 
In a seminal paper, \cite{abadie1} set out to investigate the economic impact of terrorism in Basque Country. 
Prior to the outset of terrorist activity in the early 1970's, Basque Country was considered to be one of the wealthiest regions in Spain. 
After thirty years of turmoil, however, its economic activity dropped substantially relative to its neighboring regions. 
Although intuition affirms that Basque Country's economic downturn can be attributed, at least partially, to its political and civil unrest, it is difficult to quantitatively isolate the economic costs of conflict. 
In response to this challenge, \cite{abadie1} introduced the synthetic controls framework. 
At its core, synthetic controls constructs a synthetic Basque Country from a weighted composition of control regions that are largely unaffected by the instability to estimate Basque Country's economic evolution in the absence of terrorism. 
This novel concept has inspired an entire subliterature within econometrics 
 that is ``arguably the most important innovation in the policy evaluation literature in the last 15 years'' \citep{athey_imbens}. 

Researchers have historically tackled problems of this flavor using repeated observations of units across time, i.e., panel data, where a subset of units are exposed to a treatment during some time periods while the other units are unaffected.  
In the study above, the per capita gross domestic product (GDP) of $17$ Spanish regions are measured from 1955--1998. 
Basque Country is the sole treated unit and the remaining regions are the control units; the pre- and post-treatment periods are defined as the time horizons before and after the first wave of terrorist activity, respectively. 

Synthetic controls has become a cornerstone for panel studies in recent years and across numerous fields. 
Beforehand, the unconfoundedness approach \citep{rubin_rosenbaum, imbens_wooldridge} served as a common workhorse. 
Whereas synthetic controls posits a relation between treated and control units that is stable across time, unconfoundedness posits a relation between treated and pretreatment periods that is stable across units. 
Accordingly, synthetic controls exploits cross-sectional correlation patterns while unconfoundedness exploits time series correlation patterns.
Considering the panel data format, unconfoundedness and synthetic controls based methods are commonly referred to as horizontal (\HZ) and vertical (\VT) regressions, respectively. 
Given their conceptual and computational distinctions, the two approaches are considered to be fundamentally different \citep{mc_panel}. 
%
%

Yet, contrary to conventional wisdom, it turns out that \HZ~and \VT~regressions can yield identical point estimates. 
As Figure~\ref{fig:teaser_est.1} shows, when the regression models are learned via ordinary least squares (OLS) or principal component regression (PCR), then the two approaches produce the same economic evolution for Basque Country in the absence of terrorism. 
Figure~\ref{fig:teaser_est.2}, by contrast, shows that when the regression models are learned via lasso or lie within the simplex---as proposed by \cite{abadie1} for \VT~regression---then the two approaches output contrasting economic trajectories. 
%
\begin{figure} [!t]
	\centering 
	\begin{subfigure}[b]{0.32\textwidth}
		\centering 
		\includegraphics[width=\linewidth]
		{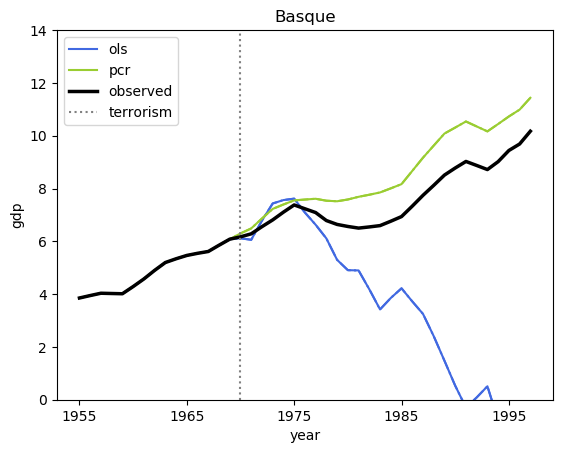}
		\caption{Symmetric regressions.} 
		\label{fig:teaser_est.1} 
	\end{subfigure} 
	\qquad \quad
	\begin{subfigure}[b]{0.32\textwidth}
		\centering 
		\includegraphics[width=\linewidth]
		{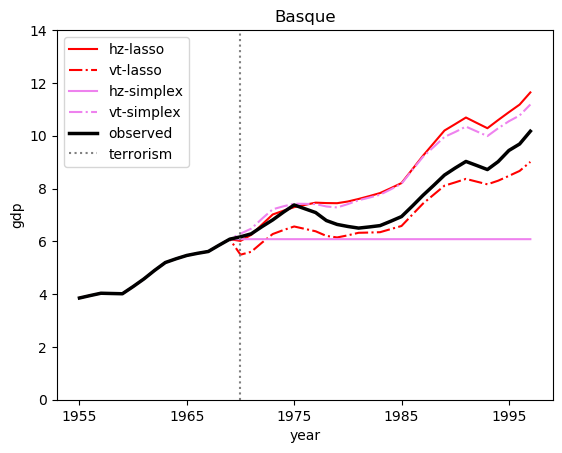}
		\caption{Asymmetric regressions.} 
		\label{fig:teaser_est.2} 
	\end{subfigure} 
	\caption{
	\ref{fig:teaser_est.1}: Estimates of OLS with minimum $\ell_2$-norm and PCR. 
	\ref{fig:teaser_est.2}: Estimates of lasso and simplex regression. 
	\HZ~and \VT~estimates correspond to colored solid and dashed-dotted lines, respectively. 
	The outset of terrorism is the vertical line 
	and Basque Country's observed GDP is in solid black. 
	}
	\label{fig:teaser_est} 
\end{figure}  
Curiously, Figure~\ref{fig:teaser_inf} indicates that even when the two regressions arrive at the same point estimate, the confidence intervals can be markedly different under different sources of randomness. 
\begin{figure} [!t]
	\centering 
	\begin{subfigure}[b]{0.32\textwidth}
		\includegraphics[width=\linewidth]
		{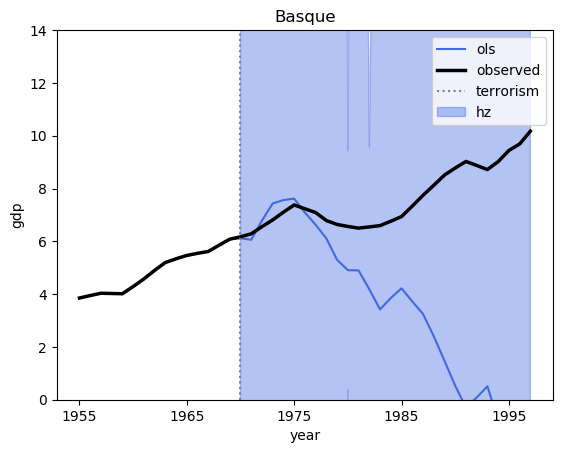}
		\caption{OLS under \HZ~model.} 
		\label{fig:basque_hz_jack}
	\end{subfigure} 
	\qquad \quad
	\begin{subfigure}[b]{0.32\textwidth}
		\includegraphics[width=\linewidth]
		{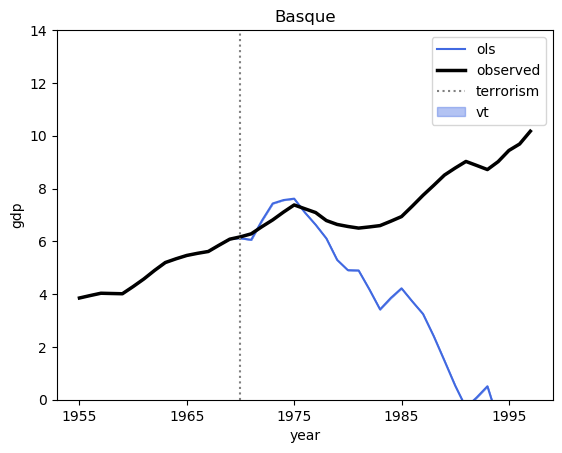}
		\caption{OLS under \VT~model.} 
		\label{fig:basque_vt_jack} 
	\end{subfigure} 
	\\
	\begin{subfigure}[b]{0.32\textwidth}
		\centering 
		\includegraphics[width=\linewidth]
		{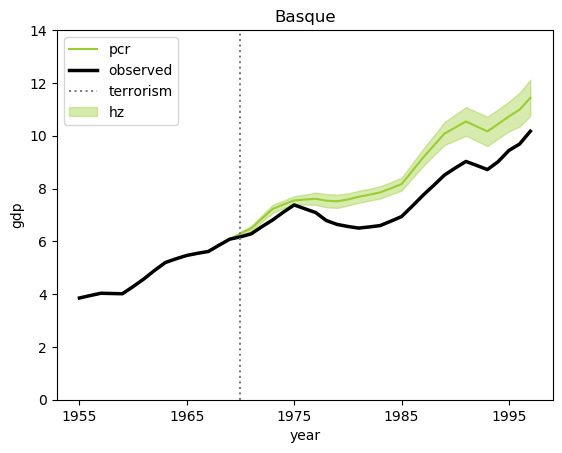}
		\caption{PCR under \HZ~model.} 
	\end{subfigure} 
	\qquad \quad
	\begin{subfigure}[b]{0.32\textwidth}
		\centering 
		\includegraphics[width=\linewidth]
		{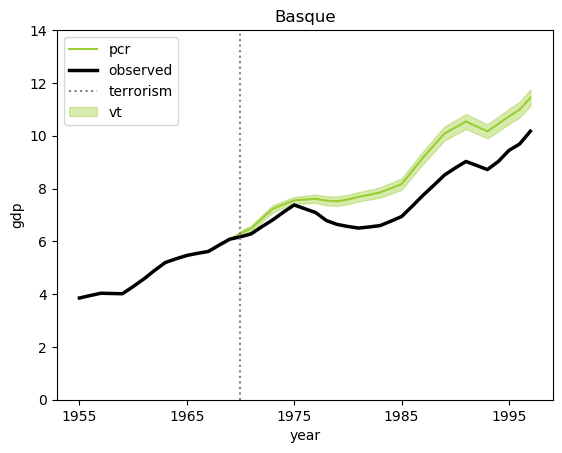}
		\caption{PCR under \VT~model.} 
	\end{subfigure} 
	\caption{
	Confidence intervals for OLS (top) and PCR (bottom) constructed from \HZ-based~(left) and \VT-based~(right) generative models, which are precisely defined in Section~\ref{sec:assumptions.mb}. 
	}
	\label{fig:teaser_inf} 
\end{figure}  
The juxtaposition of these figures beg two questions: 
\medskip 
\begin{tcolorbox}[colback=gray!10!white,colframe=black!75!black]
	\begin{center}
	{
		{\bf Q1:} ``When are \HZ~and \VT~point estimates identical?'' 
		\\ 
		{\bf Q2:} ``When the point estimates are identical, how does the source of randomness impact inference?''
	} 
	\end{center} 
\end{tcolorbox}

\noindent 
{\bf Contribution.} 
This article tackles Q1--Q2 from first principles. 
In this endeavor, we begin by classifying several widely studied regression formulations into 
(i) a symmetric class that yields identical point estimates and (ii) an asymmetric class that yields contrasting point estimates.  
Within the symmetric class, we study properties of the estimator with randomness stemming from 
(i) time series patterns,
(ii) cross-sectional patterns,
and (iii) both patterns. 
We conduct our analysis from a (i) model-based perspective, which attributes randomness to the potential outcomes, and a (ii) design-based perspective, which attributes randomness to the treatment assignment mechanism. 
In both frameworks, we find that the source of randomness has large implications for the estimand and inference. 
Under the model-based framework, we construct confidence intervals for each source of randomness. 
Through data-inspired simulations and empirical applications, we demonstrate that the confidence interval developed for one estimand often has incorrect coverage for another estimand. 
Taken together, our results emphasize that the source of randomness that researchers assume has direct implications for the accuracy of the inference that can be conducted.

\smallskip 
\noindent
{\bf Organization.} 
Section~\ref{sec:setup} overviews the panel data framework. 
Sections~\ref{sec:estimation}--\ref{sec:inference} provide one set of answers for Q1--Q2. 
Section~\ref{sec:illustrations} illustrates concepts developed in this article. 
Section~\ref{sec:conclusion} summarizes our findings. 
%
Details of simulations and empirical applications, select discussions, and mathematical proofs are relegated to the appendix.  

\smallskip 
\noindent
{\bf Notation.} 
Let $\bI$ be the identity matrix.
Let $\bone$ and $\bzero$ be the vectors of ones and zeros, respectively. 
The curled inequality denotes $\succeq$ the generalized inequality, i.e., componentwise inequality between vectors and matrix inequality between symmetric matrices.
Let $\circ$ denote the componentwise product. 
For vectors $\ba$ and $\bb$, let $\langle a, b \rangle = a' b$ denote the inner product. 
%
%
Let $\tr(\bA)$ denote the trace of $\bA$. 
We define $0/0 = 0$ when applicable.

\section{The Panel Data Framework} 
\label{sec:setup} 
We anchor on the Basque study to introduce the panel data framework. 
Panel data contains observations of $N$ units over $T$ time periods.  
The Basque study, for instance, consists of per capita GDP across $N=17$ Spanish regions over $T=43$ years. 
In each time period $t$, each unit $i$ is characterized by two potential outcomes, $Y_{it}(0)$ and $Y_{it}(1)$, which correspond to its outcome in the absence and presence of a binary treatment, respectively. 
The potential outcomes framework posits that each region possesses two possible levels of economic activity each year, one that is immune to terrorism and another that is affected by terrorism. 
In reality, however, 
we can only observe one economic state, $Y_{it}(0)$ or $Y_{it}(1)$---this is the fundamental challenge of causal inference. 

Let $Y_{it}$ be the observed outcome. 
Often, we observe all $N$ units without treatment (control) for $T_0$ time periods, i.e., $Y_{it} = Y_{it}(0)$ for all $i \le N$ and $t \le T_0$. 
For the remaining $T_1 = T-T_0$ time periods,  
$N_1$ units receive treatment while the remaining $N_0 = N-N_1$ units remain under control, i.e., if we arbitrarily label the first $N_0$ units as the control group, then 
$Y_{it} = Y_{it}(1)$ for all $i > N_0$ and $t > T_0$, and $Y_{it} = Y_{it}(0)$ for all $i \le N_0$ and $t > T_0$. 
In our study, Basque Country is the single treated unit, thus $N_1 = 1$ and $N_0=16$. 
The first wave of terrorist activity partitions the time horizon into pre- and post-treatment periods of lengths $T_0=15$ and $T_1=28$ years, respectively. 

For ease of exposition, this article considers a single treated unit and single treated period indexed by the $N$th unit and $T$th time period, respectively. 
However, our results hold for any $(i,t)$ pair where $i > N_0$ is a treated unit and $t > T_0$ is a treated period. 
We organize our observed control data into an $N \times T$ matrix, $\bY = [Y_{it}]$, as shown in Figure~\ref{fig:matrix}.   
In our example,  
$\by_N = [Y_{Nt}: t \le T_0] \in \Rb^{T_0}$ represents Basque Country's economic evolution prior to the outset of terrorism; 
$\bY_0 = [Y_{it}: i \le N_0, t \le T_0] \in \Rb^{N_0 \times T_0}$ represents the control regions' economic evolution prior to the outset of terrorism; 
and $\by_T = [Y_{iT}: i \le N_0] \in \Rb^{N_0}$ represents the control regions' economic evolution after the outset of terrorism. 
Our object of interest is Basque Country's counterfactual GDP in the absence of terrorism, $Y_{NT}(0)$. 

\begin{figure} [!t]
	\centering 
	\includegraphics[width=0.65\linewidth]
	{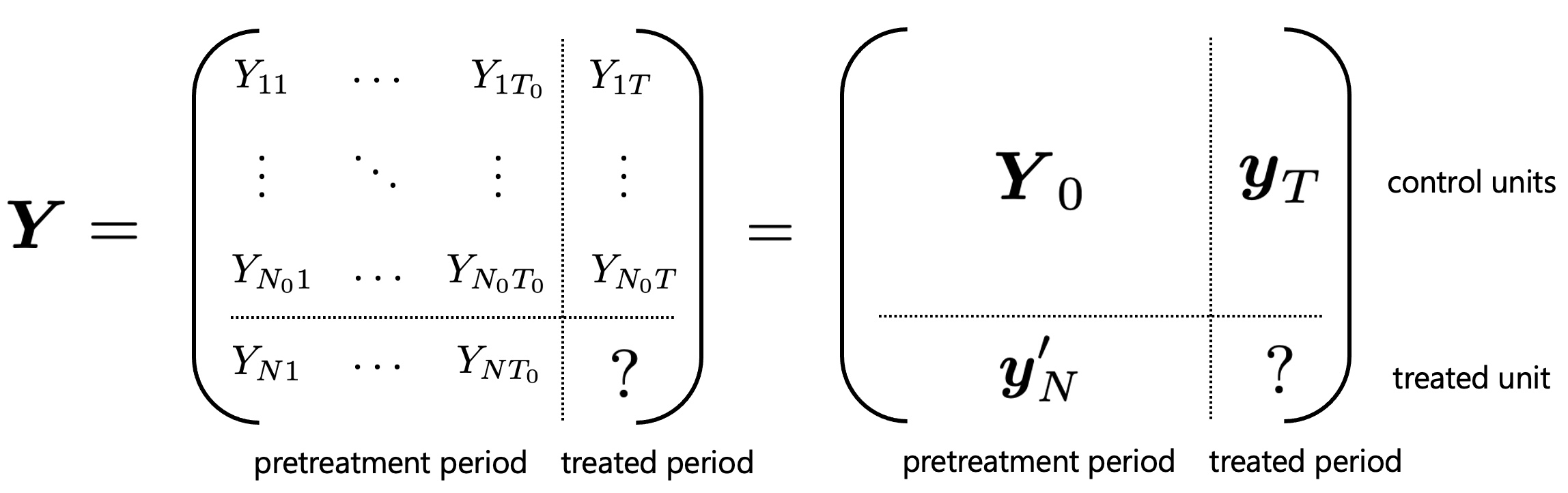}
	\caption{
	Panel data format with rows and columns indexed by units and time, respectively.  
	}
	\label{fig:matrix} 
\end{figure} 

\subsection{Time Series Versus Cross-Sectional Based Regressions} 
The information across time and space motivates two natural ways to impute 
the missing $(N,T)$th entry. 
These perspectives are explored in two large and mostly separate bodies of work \citep{mc_panel}. 

\subsubsection{Horizontal Regression and Unconfoundedness} 
The unconfoundedness literature operates on the concept that ``history is a guide to the future''.  
As such, unconfoundedness methods express outcomes in the treated period as a weighted composition of outcomes in the pretreatment periods.  
This is carried out by regressing the control units' treated period outcomes $\by_T$ on its lagged outcomes $\bY_0$ and applying the learned regression coefficients to the treated unit's lagged outcomes $\by_N$ to predict the missing $(N,T)$th outcome. 
%
%
Following \cite{mc_panel}, we refer to such methods as horizontal (\HZ) regression. 

\subsubsection{Vertical Regression and Synthetic Controls} 
The synthetic controls literature is built on the concept that ``similar units behave similarly''. 
Therefore, synthetic controls methods express the treated unit's outcomes as a weighted composition of control units' outcomes. 
This is carried out by regressing the treated unit's lagged outcomes $\by_N$ on the control units' lagged outcomes $\bY'_0$ and applying the learned regression coefficients to the control units' treated period outcomes $\by_T$ to predict the missing $(N,T)$th outcome. 
%
%
Following \cite{mc_panel}, we refer to such methods as vertical (\VT) regression. 

\subsubsection{Conventional Wisdom} 
The asymmetry between \HZ~and \VT~regressions has created the conception that they are fundamentally different approaches \citep{mc_panel}. 
In fact, the unregularized forms of \HZ~and \VT~regressions are cautioned against when $T > N$ and $N > T$, respectively \citep{abadie4, sc_enet, li_bell, mc_panel}. 
With regularization, however, \cite{mc_panel} argues the two approaches can be applied to the same setting. 
In turn, this allows the two approaches to be systematically compared through methods such as cross-validation. 

In parallel, the growth rates of the two literatures have also exhibited asymmetry. 
While the development of the unconfoundedness literature has seemingly plateaued, the synthetic controls literature continues to rapidly expand.
Across many domains, synthetic controls based methods are arguably the de facto approach for panel studies.

\section{Point Estimation} \label{sec:estimation} 
\smallskip 
\begin{tcolorbox}[colback=gray!10!white,colframe=black!75!black]
	\begin{center}
	{
		{\bf Q1:} ``When are \HZ~and \VT~point estimates identical?'' 
	} 
	\end{center} 
\end{tcolorbox}

We tackle Q1 by studying the finite-sample estimation properties of \HZ~and \VT~regressions. 
We denote the singular value decomposition of $\bY_0$ as 
$\bY_0 = \sum_{\ell =1}^{R} s_\ell \bu_\ell \bv'_\ell = \bU \bS \bV'$,
where $\bu_\ell \in \Rb^{N_0}$ and $\bv_\ell \in \Rb^{T_0}$ are the left and right singular vectors, respectively, $s_\ell \in \Rb$ are the ordered singular values, and $R = \text{rank}(\bY_0) \le \min\{N_0, T_0\}$. 
$\bU \in \Rb^{N_0 \times R}$ and $\bV \in \Rb^{T_0 \times R}$ denote the matrices formed by the left and right singular vectors, respectively, and $\bS \in \Rb^{R \times R}$ is the diagonal matrix of  singular values. 
The Moore-Penrose pseudoinverse of $\bY_0$ is 
$\bY_0^\dagger = \sum_{\ell=1}^R (1/s_\ell) \bv_\ell \bu'_\ell = \bV \bS^{-1} \bU'$. 
Critically, we do not place any assumptions on the relative magnitudes of $N$ and $T$. 
 
\subsection{Classifying Notable Regression Formulations} \label{sec:est.notable} 
We present several of the most widely studied regression formulations in the \HZ~and \VT~literatures. 
This list is far from exhaustive given the vastness of these literatures. 
%

\subsubsection{Description of Estimation Strategies}

\noindent 
{\bf Penalized regression.} 
A large class of penalized regressions are expressed as follows: 
\begin{itemize}
	\item [(a)] \HZ~regression:  for $\lambda_1, \lambda_2 \ge 0$, 
	\begin{align} 
		&\bhalpha = \argmin_{\balpha} ~\| \by_T - \bY_0 \balpha \|_2^2 + \lambda_1 \| \balpha \|_1 + \lambda_2 \|\balpha \|_2^2 \label{eq:hz.rr.1} 
		\\ &\hY_{NT}^\hz(0) = \langle \by_N, \bhalpha \rangle. \label{eq:hz.rr.2} 
	\end{align} 
	
	\item [(b)] \VT~regression: for $\lambda_1, \lambda_2 \ge 0$, 
	\begin{align} 
		&\bhbeta = \argmin_{\bbeta} ~\| \by_N - \bY'_0 \bbeta \|_2^2 + \lambda_1 \| \bbeta \|_1 + \lambda_2 \| \bbeta \|_2^2 \label{eq:vt.rr.1} 
		\\ &\hY_{NT}^\vt(0) = \langle \by_T, \bhbeta \rangle. \label{eq:vt.rr.2} 
	\end{align} 
\end{itemize} 
We overview common choices for $(\lambda_1, \lambda_2)$ and describe the corresponding strategy. 

\smallskip 
\noindent 
{\em I: Ordinary least squares (OLS).}
Arguably, the mother of all regressions is OLS, where $\lambda_1 = \lambda_2 = 0$. 
OLS is an unconstrained problem with possibly infinitely many solutions. 
OLS has been analyzed in numerous works on panel studies, including \cite{hcw}, \cite{li_bell}, and \cite{li2020}. 

\smallskip 
\noindent 
{\em II: Principal component regression (PCR).} 
To formalize PCR, let  
\begin{align} \label{eq:rank_k} 
	\bY_0^{(k)} = \sum_{\ell=1}^k s_\ell \bu_\ell \bv'_\ell
\end{align}
denote the rank $k < R$ approximation of $\bY_0$ that retains the top $k$ principal components. 
\HZ~and \VT~PCR corresponds to replacing $\bY_0$ with $\bY_0^{(k)}$ within \eqref{eq:hz.rr.1} and \eqref{eq:vt.rr.1}, respectively, with $\lambda_1 = \lambda_2 = 0$.  
In words, PCR first finds a $k$ dimensional representation of the covariate matrix via principal component analysis; 
then, PCR performs OLS with the compressed $k$ dimensional covariates.  
Within the synthetic controls literature, \cite{rsc, mrsc} and \cite{si} utilize PCR. 

\smallskip 
\noindent 
{\em III: Ridge regression.}  
Consider ridge regression, where $\lambda_1=0$ and $\lambda_2 > 0$. 
When $\bY_0$ is rank deficient, the gram matrix, i.e., $\bY'_0 \bY_0$ for \HZ~regression and $\bY_0 \bY'_0$ for \VT~regression, is ill-conditioned. 
This often discourages the usage of OLS. 
In these settings, ridge provides a remedy by adding a {\em ridge} on the diagonal of the gram matrix, which increases all eigenvalues by $\lambda_2$, thus removing the singularity problem. 
\cite{sc_aug} explores the properties of a doubly robust estimator that utilizes \HZ~ridge regression. 

\smallskip 
\noindent 
{\em IV: Lasso regression.} 
Consider lasso regression, where $\lambda_1 > 0$ and $\lambda_2 = 0$. 
Lasso has become a popular tool for estimating sparse linear coefficients in high-dimensional regimes.
Because the lasso criterion not strictly convex, there are possibly infinitely many solutions.
Thus, for our analysis of lasso only, 
we make the mild assumption that the entries of $\bY_0$ are drawn from a continuous distribution. 
As established in \cite{lasso_ryan}, this guarantees the lasso solution to be unique.
Several notable works in the synthetic controls literature, e.g., \cite{li_bell}, \cite{arco}, and \cite{sc_lasso1}, analyze the lasso. 

\smallskip 
\noindent 
{\em V: Elastic net regression.}  
Mixing both $\ell_1$ and $\ell_2$-penalties, i.e., $\lambda_1, \lambda_2 > 0$, is known as elastic net. 
At a high level, elastic net selects variables similar to the lasso, but deals with correlated variables more gracefully as with ridge.
When $\lambda_2 > 0$, the criterion is strictly convex so the solution is unique. 
\cite{sc_enet} propose an elastic net synthetic controls variant. 

\smallskip 
\noindent 
{\bf Constrained regression.} 

\noindent 
{\em VI: Simplex regression.} 
The next formulation constrains the regression weights to lie within the simplex, i.e., the weights are nonnegative and sum to one: 
\begin{itemize}
	\item [(a)] \HZ~regression: for $\lambda \ge 0$, 
	\begin{align} 
		&\bhalpha = \argmin_{\balpha} ~\| \by_T - \bY_0 \balpha \|_2^2 + \lambda \| \balpha \|_2^2 \, \,\text{~subject to~} \, \balpha' \bone = 1, \balpha \succeq \bzero \label{eq:hz.cvx.1} 
		\\ &\hY_{NT}^\hz(0) = \langle \by_N, \bhalpha \rangle. \label{eq:hz.cvx.2} 
	\end{align} 
	
	\item [(b)] \VT~regression: for $\lambda \ge 0$, 
	\begin{align} 
		&\bhbeta = \argmin_{\bbeta} ~\| \by_N - \bY_0' \bbeta \|_2^2 + \lambda \| \bbeta \|_2^2 \, \, \text{~subject to~} \, \bbeta' \bone = 1, \bbeta \succeq \bzero \label{eq:vt.cvx.1} 
		\\ &\hY_{NT}^\vt(0) = \langle \by_T, \bhbeta \rangle. \label{eq:vt.cvx.2} 
	\end{align} 
\end{itemize} 
We consider a vanishing $\ell_2$ penalty since $\lambda = 0$ (standard formulation) can induce multiple minima \citep{abadie_hour}.
When $\lambda > 0$, the criterion becomes strictly convex and the solution is unique. 
Simplex regression is the original formulation set forth in the pioneering works of \cite{abadie1, abadie2, abadie4}, and its properties continue to be actively studied today.  
Attractive aspects of simplex regression include interpretability, sparsity, and transparency \citep{abadie_survey}. 

\subsubsection{Classification Results} 
To answer Q1, we classify the regression formulations into    
(i) a symmetric class, where \HZ~and \VT~point estimates agree, and 
(ii) an asymmetric class, where \HZ~and \VT~point estimates disagree. 
We use the shorthand $\HZ = \VT$ if the two approaches produce identical point estimates and $\HZ \neq \VT$ otherwise. 

\smallskip 
\noindent 
{\bf I: Symmetric class.} 
We first state the symmetric formulations.  
\begin{theorem} \label{thm:equivalent} 
$\EHZ = \EVT$ for 
(i)  OLS with $(\bhalpha, \bhbeta)$ as the minimum $\ell_2$-norm solutions:  
\begin{align} \label{eq:ols} 
    \hY_{NT}^{\ehz}(0) &= \hY_{NT}^{\evt}(0) 
    = \langle \by_N, \bY_0^\dagger \by_T \rangle 
    = \sum_{\ell =1}^R (1/s_\ell) \langle \by_N, \bv_\ell \rangle \langle \bu_\ell, \by_T \rangle;
\end{align}
(ii) PCR with the same choice of $k < R$: 
\begin{align} \label{eq:pcr} 
    \hY_{NT}^{\ehz}(0) &= \hY_{NT}^{\evt}(0) 
    = \langle \by_N, (\bY_0^{(k)})^\dagger \by_T \rangle 
    = \sum_{\ell =1}^k (1/s_\ell) \langle \by_N, \bv_\ell \rangle \langle \bu_\ell, \by_T \rangle; 
\end{align}
(iii) ridge regression with the same choice of $\lambda_2>0$: 
\begin{align} \label{eq:ridge} 
    \hY_{NT}^{\ehz}(0) &= \hY_{NT}^{\evt}(0) 
    = \langle \by_N,  (\bY_0' \bY_0 + \lambda_2 \bI)^{-1} \bY_0' \by_T \rangle 
    = \sum_{\ell =1}^R \frac{s_\ell}{s_\ell^2 + \lambda_2} \langle \by_N, \bv_\ell \rangle \langle \bu_\ell, \by_T \rangle.
\end{align}
\end{theorem} 
%
%
Theorem~\ref{thm:equivalent} might seem familiar at first glance. 
As observed in \cite{abadie4} and \citet[][Lemma 2]{sc_aug}, the point estimates associated with \HZ~OLS and \HZ~ridge can be written as linear combinations of the elements in $\by_T$, which take the same {\em linear forms} as the corresponding \VT~point estimates. 
However, their results stop short of establishing numerical equivalence as in Theorem~\ref{thm:equivalent}. 
From this perspective, Theorem~\ref{thm:equivalent} is perhaps surprising as it takes the next step forward in contradicting the notion that the two regressions are fundamentally different. 
In particular, Theorem~\ref{thm:equivalent} proves the same point estimate is derived from both approaches whenever the regression model belongs to the symmetric class, namely (i) OLS with minimum $\ell_2$-norm, (ii) PCR for the same choice of $k$, and (iii) ridge for the same choice of $\lambda_2$. 
In this view, \HZ~and \VT~regressions are not perspectives at duel---they are dual perspectives. 

We emphasize that Theorem~\ref{thm:equivalent} holds for any data configuration. 
As such, it clarifies that \HZ~and \VT~OLS are {\em not} invalid when $T > N$ and $N > T$, respectively, as previously believed. 
In fact, the OLS estimate can even be written as $\langle \bhalpha, \bY'_0 \bhbeta \rangle$, which incorporates both regression models. 
The origin of the prior misconception may have come from the fact that infinitely many solutions exist when $\bY_0$ is rank deficient. 
Among these solutions, however, is the unique minimum $\ell_2$-norm model, which is arguably sufficient for inference \citep{shao_deng}. 
This is also the solution when the problem is optimized via gradient descent, a ubiquitous optimizer in practice. 
Phenomena of this form are known as ``implicit regularization'', where the optimization algorithm is biased towards a particular solution even though the bias is not explicit in the objective function \citep{ireg2, ireg4}. 

Through its connection to the $\ell_2$-penalty, the minimum $\ell_2$-norm also offers a high-level intuition for the root of symmetry. 
More specifically, observe that the ridge model converges to the OLS model with minimum $\ell_2$-norm as $\lambda_2 \rightarrow 0$. 
Since the PCR model is precisely the OLS minimum $\ell_2$-norm model that is restricted to the space spanned by the top $k$ principal components, we conjecture that the geometry of the $\ell_2$-ball is a likely source for \HZ~and \VT~estimation symmetry. 

\smallskip 
\noindent 
{\bf II: Asymmetric class.} 
Next, we state the class of formulations that fracture symmetry.
\begin{theorem} \label{thm:not_equivalent} 
$\EHZ \neq \EVT$ for  
(i) lasso, 
(ii) elastic net, 
and (iii) simplex regression. 
\end{theorem} 
To examine the implications of Theorem~\ref{thm:not_equivalent}, we observe that the common thread between the objective functions in the asymmetric class is a penalty or constraint that promotes sparse models. 
Such regularizers are noticeably absent in the symmetric formulations. 
This highlights an interesting trade-off---while sparsity is widely considered to be a salient feature, the geometries of the $\ell_1$-ball and simplex that encourage sparsity are also the likely sources of \HZ~and \VT~estimation asymmetry. 

\subsection{Doubly Robust Estimators} \label{sec:doubly_robust}  
In recent years, there has been a surge of interest in doubly robust estimators. 
Within panel data, we discuss two prominent works that are rising in popularity. 

\subsubsection{Synthetic Difference-in-Differences} 
An important approach that continues to dominate empirical work in panel data is the difference-in-differences (\DID) estimator \citep{did1}. 
At a high level, \DID~posits an additive outcome model with unit- and time-specific fixed effects, known more colloquially as the ``parallel trends'' assumption. 
The recent work of \cite{sdid} anchors on the \DID~principle and brings in concepts from the unconfoundedness and synthetic controls literatures to derive a doubly robust estimator called synthetic difference-in-differences (\SDID). 
In our setting, the \SDID~prediction for the missing $(N,T)$th potential outcome can be written as 
\begin{align}  
	\hY^\sdid_{NT}(0) &= \sum_{i \le N_0} \hbeta_i Y_{iT} + \sum_{t \le T_0} \halpha_t Y_{Nt}  -\sum_{i \le N_0} \sum_{t \le T_0} \hbeta_i \halpha_t Y_{it}  
	\\ &=  \langle \by_T, \bhbeta \rangle +  \langle \by_N, \bhalpha \rangle - \langle \bhbeta, \bY_0 \bhalpha \rangle, \label{eq:sdid}
\end{align} 
where $\bhalpha$ and $\bhbeta$ represent general \HZ~and \VT~models, respectively.  
The authors note that $\bhalpha = (1/T_0) \bone$ and $\bhbeta = (1/N_0) \bone$ recovers \DID. 
Moving beyond simple \DID~to performing a weighted two-way bias removal, the authors propose to learn $\bhalpha$ via simplex regression and $\bhbeta$ via simplex regression with an $\ell_2$-penalty. %

\subsubsection{Augmented Synthetic Controls} 
Another notable work is that of \cite{sc_aug}. 
The authors introduce the augmented synthetic control (\ASC) estimator, which uses an outcome model to correct the bias induced by the classical synthetic controls estimator.\footnote{See \cite{abadie_hour} for a bias correction of synthetic controls through matching.} 
Concretely, the \ASC~estimator predicts the missing $(N,T)$th potential outcome as 
\begin{align}
	\hY^\asc_{NT}(0) &= 
	\hM_{NT}(0) + \sum_{i \le N_0} \hbeta_i (Y_{iT} - \hM_{iT}(0)), \label{eq:ascm}
\end{align} 
where $\hM_{iT}(0)$ is the estimator for the $(i,T)$th entry. 
The authors instantiate 
\begin{align}
	\hM_{iT}(0) &= \sum_{t \le T_0} \halpha_t Y_{it}. 
	\label{eq:hz_outcome} 
\end{align}
Plugging the \HZ~outcome model in \eqref{eq:hz_outcome} into \eqref{eq:ascm} then gives   
\begin{align} \label{eq:asc_hz} 
	\hY^\asc_{NT}(0) &= \langle \by_T, \bhbeta \rangle +  \langle \by_N, \bhalpha \rangle - \langle \bhbeta, \bY_0 \bhalpha \rangle. 
\end{align} 
We consider this particular variant of \ASC~since it takes the same form as \SDID, as seen in \eqref{eq:sdid}. 
In contrast to \cite{sdid}, \cite{sc_aug} learns $\bhalpha$ via ridge regression and $\bhbeta$ via simplex regression. 

\subsubsection{When Doubly Robust Estimators are No Longer Doubly Robust} 
We leverage Theorem~\ref{thm:equivalent} to study the estimation properties of \SDID~and \ASC~when 
$(\bhalpha, \bhbeta)$ are learned via OLS and PCR. 
\begin{cor} \label{cor:sdid} 
$\ESDID = \EASC = \EHZ = \EVT$ for 
(i) $(\bhalpha, \bhbeta)$ as the OLS minimum $\ell_2$-norm solutions
and 
(ii) $(\bhalpha, \bhbeta)$ as the PCR solutions with the same choice of $k < R$. 
\end{cor} 
Corollary~\ref{cor:sdid} states that a researcher who uses the implicitly regularized (but explicitly unconstrained) variants of \SDID~and \ASC~inevitably arrive at the same point estimate as their colleague who simply uses \HZ~or \VT~OLS. 
The same phenomena occurs for PCR. 
From this, we observe that \SDID~and \ASC~can lose their weighted double-differencing effects under certain formulations. 
However, Corollary~\ref{cor:sdid} is not to discredit either approach. 
The fact remains that both methods allow researchers to naturally and simultaneously encode their knowledge of time and unit specific structures into the estimator. 
\SDID~and \ASC~warrant further studies and careful consideration. 

\subsection{Intercepts} \label{sec:intercepts} 
Intercepts can be included in the \HZ~regression model by modifying the $\ell_2$-errors in \eqref{eq:hz.rr.1} and \eqref{eq:hz.cvx.1}
 as $\| \by_T - \bY_0 \balpha - \alpha_0 \bone \|_2^2$; 
 similarly, they can be included in the \VT~regression model by modifying $\ell_2$-errors in \eqref{eq:vt.rr.1} and \eqref{eq:vt.cvx.1} as 
$\| \by_N - \bY'_0 \bbeta - \beta_0 \bone \|_2^2$. 
To date, there is still a lack of general consensus on the inclusion of intercepts within the synthetic controls literature. 
We attempt to shed light on the role of intercepts.  
\begin{cor} \label{cor:intercepts} 
$\EHZ \neq \EVT$ for (i) OLS, (ii) PCR, and (iii) ridge with intercepts. 
\end{cor} 
We develop an intuition for Proposition~\ref{cor:intercepts} by interpreting intercepts in panel studies. 
A nonzero time intercept, $\alpha_0$, imposes a permanent constant difference between the treated and pretreatment periods; 
a nonzero unit intercept, $\beta_0$, imposes a permanent constant difference between the treated and control units. 
These systematic structures then create an asymmetry between the two regressions. 
Below, we propose a methodology based on centering the data that allows for intercepts yet retains symmetry. 

\subsubsection{Including Intercepts and Retaining Symmetry through Data Centering} 
Let $\bY_0$ be twice centered, i.e., the rows and columns of $\bY_0$ are mean zero. 
This can be satisfied by applying $\bI - (1/N_0) \bone \bone'$ and $\bI - (1/T_0) \bone \bone'$ to the left and right, respectively, of $\bY_0$. 
Consider the following modified formulations.  
\begin{itemize}
	\item [(a)] \HZ~regression:  for $\lambda \ge 0$, 
	\begin{align} 
		&(\halpha_0, \halpha_1, \bhalpha) = \argmin_{(\alpha_0, \alpha_1, \balpha)} ~\| \by_T - \bY_0 \balpha - \alpha_0 \bone \|_2^2 + \| \by_N - \alpha_1 \bone \|_2^2 + \lambda \|\balpha \|_2^2 \label{eq:hz.int.1} 
		\\ &\hY_{NT}^\hz(0) = \langle \by_N, \bhalpha \rangle + \halpha_0 + \halpha_1. \label{eq:hz.int.2} 
	\end{align} 
	
	\item [(b)] \VT~regression: for $\lambda \ge 0$,  
	\begin{align} 
		&(\hbeta_0, \hbeta_1, \bhbeta) = \argmin_{(\beta_0, \beta_1, \bbeta)} ~\| \by_N - \bY'_0 \bbeta - \beta_0 \bone \|_2^2 + \| \by_T - \beta_1 \bone \|_2^2 + \lambda \|\bbeta \|_2^2 \label{eq:vt.int.1} 
		\\ &\hY_{NT}^\vt(0) = \langle \by_T, \bhbeta \rangle + \hbeta_0 + \hbeta_1. \label{eq:vt.int.2} 
	\end{align} 
\end{itemize} 
Similar to before, OLS corresponds to $\lambda = 0$, PCR corresponds to OLS with $\bY_0^{(k)}$ for $k < R$ in place of $\bY_0$, and ridge regression corresponds to any $\lambda > 0$. 
\begin{cor} \label{cor:equivalent.int} 
$\EHZ = \EVT$ for the symmetric estimators in Theorem~\ref{thm:equivalent} under the formulations set in \eqref{eq:hz.int.1} and \eqref{eq:vt.int.1} with $\bY_0$ being twice centered. 
\end{cor} 
We inspect \eqref{eq:hz.int.2} and \eqref{eq:vt.int.2} to understand the implications of  Corollary~\ref{cor:equivalent.int}. 
First, we recall Theorem~\ref{thm:equivalent}, which establishes that the \HZ~and \VT~estimates share the same ``base'' estimate, i.e., $\langle \by_N, \bhalpha \rangle = \langle \by_T, \bhbeta \rangle$. 
Next, we note that 
$\halpha_0 = \hbeta_1 = (1/N_0) \bone' \by_T$
and 
$\hbeta_0 = \halpha_1 = (1/T_0) \bone' \by_N$, which correspond to the time and unit fixed effects, respectively. 
Intuitively, the modified point estimates in \eqref{eq:hz.int.2} and \eqref{eq:vt.int.2} include both fixed effect models to compensate for $\bY_0$ being twice centered. 
Putting everything together, the modified \HZ~and \VT~point estimates are identical. 

\section{Inference} \label{sec:inference} 
\smallskip 
\begin{tcolorbox}[colback=gray!10!white,colframe=black!75!black]
	\begin{center}
	{
		{\bf Q2:} ``When the \HZ~and \VT~point estimates are identical, how does the source of randomness impact inference?'' 
	} 
	\end{center} 
\end{tcolorbox}
\smallskip 
To answer Q2, we study the inferential properties of the counterfactual prediction. 
Formal discussions for classical inference require an explicit postulation of where the randomness arises from.  
This article takes both a (i) model-based approach, which makes assumptions about the distribution of the potential outcomes, and a (ii) design-based approach, which makes assumptions about the assignment mechanism of treatment. 
%
We start by taking a model-based route that is similar in parts to the recent works of \cite{li2020}, \cite{cattaneo}, and \cite{sc_lasso1, sc_t_test} from the synthetic controls literature. 
We then transition to a design-based route that is paved by the ideas introduced in \cite{guido_inf}. 

%

For ease of exposition, we focus on OLS and its minimum $\ell_2$-norm solutions, i.e., $\bhalpha = \bY_0^\dagger \by_T$ and $\bhbeta = (\bY'_0)^\dagger \by_N$. 
As such, we define $\hY_{NT}(0) = \hY^\hz_{NT}(0) = \hY^\vt_{NT}(0)$, which is justified under Theorem~\ref{thm:equivalent}. 
To improve readability, we present several results informally and provide their precise versions in Appendix~\ref{sec:appendix_inference}. 

\subsection{Model-Based Inference} \label{sec:model_based} 
Within the model-based framework, we consider a classical regression model. 
This postulation is not always plausible but it is useful to dissect its implications for the role of randomness in conducting inference. 

\subsubsection{Generative Models} \label{sec:assumptions.mb} 
We now study properties of $\hY_{NT}(0)$ from three different sources of randomness. 

\medskip 
\noindent 
{\bf I: Horizontal model.} 
The \HZ~model focuses on time series correlation patterns. 
%
\begin{assumption} \label{assump:hz.outcome} 
Conditional on $(\by_N, \bY_0)$, we have 
\begin{align}
	&Y_{iT} = \sum_{t \le T_0} \alpha^*_t Y_{it} + \varepsilon_{iT}, \quad i=1, \dots, N_0, \label{eq:hz.yt} 
\end{align} 
where $\balpha^*$ is a vector of unknown coefficients and $\varepsilon_{iT}$ is an idiosyncratic error term. 
\end{assumption} 
Assumption~\ref{assump:hz.outcome} motivates the \HZ~approach, which models time series patterns as the source of randomness. 
Under the \HZ~model, 
the statistical uncertainty of $\hY_{NT}(0)$ is governed by the construction of $\bhalpha$ from $(\by_T, \bY_0)$, i.e., the in-sample uncertainty. 
%

\begin{assumption} \label{assump:hz.errors}
$\{\varepsilon_{iT}\}_{i=1}^{N_0}$ has zero mean and is independent over $i = 1, \dots, N_0$, 
conditional on $(\by_N, \bY_0)$. 
\end{assumption} 
Assumption~\ref{assump:hz.errors} states that the errors have zero mean and thus the regressors, i.e., lagged outcomes, are uncorrelated with the errors; this is known in the literature as strict exogeneity.  
The errors are also conditionally independent across space.

\medskip 
\noindent 
{\bf II: Vertical model.} 
The \VT~model focuses on cross-sectional correlation patterns. 

\begin{assumption} \label{assump:vt.outcome} 
Conditional on $(\by_T, \bY_0)$, we have 
\begin{align}
	&Y_{Nt} = \sum_{i \le N_0} \beta^*_i Y_{it} + \varepsilon_{Nt}, \quad t=1, \dots, T_0, \label{eq:vt.yn} 
\end{align} 
where $\bbeta^*$ is a vector of unknown coefficients and $\varepsilon_{Nt}$ is an idiosyncratic error term. 
\end{assumption} 
Assumption~\ref{assump:vt.outcome} is analogous to Assumption~\ref{assump:hz.outcome} with the source of randomness here emanating from cross-sectional patterns.
Hence, the statistical uncertainty of $\hY_{NT}(0)$ under the \VT~model is governed by the construction of $\bhbeta$ from $(\by_N, \bY_0)$. 
%

\begin{assumption} \label{assump:vt.errors}
$\{\varepsilon_{Nt}\}_{t=1}^{T_0}$ has zero mean and is independent over $t = 1, \dots, T_0$, 
conditional on $(\by_T, \bY_0)$. 
\end{assumption} 
Assumption~\ref{assump:vt.errors} is analogous to Assumption~\ref{assump:hz.errors} with the errors here being conditionally independent across time. 

\medskip 
\noindent 
{\bf III: Mixed model.} 
We introduce a new model that mixes aspects of the \HZ~and \VT~models. 
At a high level, the mixed model accounts for randomness along both dimensions of the data. 
This comes at the price of placing additional constraints on the stochastic properties of the errors. 

\begin{assumption} \label{assump:mixed.outcome} 
Conditional on $\bY_0$, we have 
\begin{align}
	&Y_{iT} = \sum_{t \le T_0} \alpha^*_t Y_{it} + \varepsilon_{iT}, \quad i = 1, \dots, N_0, \label{eq:mixed.yt} 
	\\
	&Y_{Nt} = \sum_{i \le N_0} \beta^*_i Y_{it} + \varepsilon_{Nt}, \quad t = 1, \dots, T_0, \label{eq:mixed.yn} 
\end{align} 
where $(\balpha^*, \varepsilon_{iT})$ and $(\bbeta^*, \varepsilon_{Nt})$ are defined as in Assumptions~\ref{assump:hz.outcome} and \ref{assump:vt.outcome}, respectively. 

\end{assumption} 
\eqref{eq:mixed.yt} and \eqref{eq:mixed.yn} correspond to Assumptions~\ref{assump:hz.outcome} and \ref{assump:vt.outcome}, respectively. 
Collectively, they model time series and cross-sectional patterns as two distinct sources of randomness. 
Thus, the statistical uncertainty of $\hY_{NT}(0)$ under the mixed model is governed by the constructions of both $\bhalpha$ and $\bhbeta$. 

\begin{assumption} \label{assump:mixed.errors}
$\{\varepsilon_{iT}\}_{i=1}^{N_0}$ and $\{\varepsilon_{Nt}\}_{t=1}^{T_0}$ have zero mean and are independent over $i =1, \dots, N_0$ and $t = 1, \dots, T_0$, conditional on $\bY_0$. 
\end{assumption} 
Assumption~\ref{assump:mixed.errors} combines Assumptions~\ref{assump:hz.errors} and \ref{assump:vt.errors}. 
In words, it states that $\bY_0$ contains all measured confounders and the errors are independent across both time and space. 
We leave a proper analysis under dependent errors as future work. 

\subsubsection{Inferential Properties} 
Equipped with our three models, we are ready to provide one set of answers to Q2. 
In what follows, we define the error covariance matrices as 
(i) $\bSigma^\hz_T = \Cov(\bvarepsilon_T | \by_N, \bY_0)$, 
(ii) $\bSigma^\vt_N = \Cov(\bvarepsilon_N | \by_T, \bY_0)$,
(iii) $\bSigma^\mix_T = \Cov(\bvarepsilon_T | \bY_0)$,
and (iv) $\bSigma^\mix_N = \Cov(\bvarepsilon_N | \bY_0)$, 
where $\bvarepsilon_T = [\varepsilon_{iT}: i \le N_0]$ and $\bvarepsilon_N = [\varepsilon_{Nt}: t \le T_0]$. 
\begin{theorem} \label{thm:inference} 
(i) \emph{[\HZ~model]} Under Assumptions~\ref{assump:hz.outcome}--\ref{assump:hz.errors} and suitable moment conditions, we have 
\begin{align}
	& \frac{\hY_{NT}(0) - \mu_0^\ehz}{ \sqrt{v_0^\ehz}} \xrightarrow{d} \mathcal{N}(0,1), \label{eq:hz.mb} 
\end{align} 
where $\mu_0^\ehz = \langle \by_N, \bH^v \balpha^* \rangle$ and $v_0^\ehz = \bhbeta' \bSigma^\ehz_T \bhbeta$. 
(ii) \emph{[\VT~model]} Under Assumptions~\ref{assump:vt.outcome}--\ref{assump:vt.errors} and suitable moment conditions, we have 
\begin{align}
	\frac{\hY_{NT}(0) - \mu_0^\evt}{ \sqrt{v^\evt_0}} \xrightarrow{d} \mathcal{N}(0,1), \label{eq:vt.mb}
\end{align} 
where $\mu_0^\evt = \langle \by_T, \bH^u \bbeta^* \rangle$ and $v^\evt_0 = \bhalpha' \bSigma^\evt_N \bhalpha$.  
(iii) \emph{[Mixed model]} Under Assumptions~\ref{assump:mixed.outcome}--\ref{assump:mixed.errors} and suitable moment conditions, we have 
\begin{align}
	\frac{ \hY_{NT}(0) - \mu^\emix_0}{\sqrt{v^\emix_0}} \xrightarrow{d} \mathcal{N}(0,1), \label{eq:mixed.mb}
\end{align}
where $\mu^\emix_0 = \langle \balpha^*, \bY'_0 \bbeta^* \rangle$ 
and $$v^\emix_0 = 
	 (\bH^u \bbeta^*)' \bSigma^\emix_T (\bH^u \bbeta^*)
	+ (\bH^v \balpha^*)' \bSigma^\emix_N (\bH^v \balpha^*)  
	+ \tr( \bY^\dagger_0 \bSigma^\emix_T (\bY'_0)^\dagger \bSigma^\emix_N).$$ 
\end{theorem} 
Theorem~\ref{thm:inference} highlights that each model measures uncertainty with respect to a different  estimand. 
In particular, Theorem~\ref{thm:inference} states that the asymptotic variance is controlled by time series patterns under the \HZ~model, cross-sectional patterns under the \VT~model, and both correlation patterns under the mixed model. 
This clarifies that the source of randomness has substantive implications for the estimand and inference. 

%

\subsubsection{Confidence Intervals} \label{sec:ci} 
Theorem~\ref{thm:inference} motivates separate \HZ, \VT, and mixed confidence intervals: for $\theta \in (0,1)$, 
\begin{align} 
	\mu_0^\hz &\in \left[ \hY_{NT}(0) ~\pm~ z_{\frac \theta 2} \sqrt{\hv^\hz_0} \right], 
	\\
	\mu_0^\vt &\in \left[ \hY_{NT}(0) ~\pm~ z_{\frac \theta 2} \sqrt{\hv^\vt_0} \right], 
	\\ 	
	\mu_0^\mix &\in \left[ \hY_{NT}(0) ~\pm~ z_{\frac \theta 2} \sqrt{\hv^\mix_0} \right], 
\end{align}
where 
$z_{\frac \theta 2}$ is the upper $\theta/2$ quantile of $\mathcal{N}(0,1)$, 
and 
$(\hv^\hz_0, \hv^\vt_0, \hv^\mix_0)$ are the estimators of $(v^\hz_0, v^\vt_0, v^\mix_0)$. 
We construct 
\begin{align}
	\hv^\hz_0 = \bhbeta' \bhSigma_T \bhbeta,
	\quad 
	\hv^\vt_0 = \bhalpha' \bhSigma_N \bhalpha, 
	\quad
	\hv^\mix_0  
	= \hv_0^\hz 
	+ \hv_0^\vt
	-  \tr(\bY_0^\dagger \bhSigma_T (\bY'_0)^\dagger \bhSigma_N),  \label{eq:final_var_est} 
\end{align}  
where $\bhSigma_T$ and $\bhSigma_N$ are the estimators of $(\bSigma^\hz_T, \bSigma^\mix_T)$ and $(\bSigma^\vt_N, \bSigma^\mix_N)$, respectively. 
We precisely define them under homoskedastic and heteroskedastic errors below. 
To reduce ambiguity, we index $(\hv^\hz_0, \hv^\vt_0, \hv^\mix_0)$ by the covariance estimator. 
We also denote $\bH^u = \bU \bU'$ and $\bH^u_\perp = \bI - \bH^u$. 
With this notation, the \HZ~and \VT~in-sample errors can be written as $\bH^u_\perp \by_T = \by_T - \bY_0 \bhalpha$ 
and $\bH^v_\perp \by_N = \by_N - \bY'_0 \bhbeta$, respectively. 

It is clear from \eqref{eq:final_var_est} that $(\hv^\hz_0, \hv^\vt_0)$ are plug-in estimators for $(v^\hz_0, v^\vt_0)$. 
As such, we discuss $\hv^\mix_0$ with respect to $v^\mix_0$. 
Recall $\bhalpha = \bH^v \bhalpha$ and $\bhbeta = \bH^u \bhbeta$ by construction. 
To justify the negative trace in $\hv^\mix_0$, note that $\hv_0^\hz$ is a quadratic involving $(\by_N, \by_T)$. 
Since both quantities are random, the expectation of $\hv_0^\hz$ induces an additional term that precisely corresponds to the trace term in $v_0^\mix$.
The same property holds for $\hv_0^\vt$. 
Thus, $\hv^\mix_0$ corrects for this bias via the negative trace. 

\smallskip 
\noindent 
{\bf Homoskedastic errors.} 
Consider $\bSigma^\hz_T$ with identical diagonal elements, i.e., $\bSigma^\hz_T = (\sigma^\hz_T)^2 \bI$, where $(\sigma^\hz_T)^2 = \Var(\varepsilon_{iT} | \by_N, \bY_0)$ for $i = 1, \dots, N_0$. 
Let $(\bSigma^\vt_N, \bSigma^\mix_T, \bSigma^\mix_N)$ be defined analogously. 
We use the standard variance estimators 
\begin{align}
	&\bhSigma_T^\homo = \frac{1}{N_0 - R} \| \bH^u_\perp \by_T \|_2^2 \bI, 	\label{eq:hz.var.1} 
	\\
	&\bhSigma_N^\homo = \frac{1}{T_0 - R} \| \bH^v_\perp \by_N \|_2^2 \bI, \label{eq:vt.var.1} 
\end{align}  
where $R = \text{rank}(\bY_0)$, which can be computed as $R = \tr(\bH^u) = \tr(\bH^v)$. 

\begin{lemma} \label{lemma:inference.1} 
Consider homoskedastic errors. 
(i) \emph{[\HZ~model]} Under Assumptions~\ref{assump:hz.outcome}--\ref{assump:hz.errors}, we have 
\begin{align}
	\Ex[\bhSigma_T^\ehomo | \by_N, \bY_0] =  \bSigma^\ehz_T 
	\quad \text{and} \quad 
	\Ex[ \hv^{\ehz, \ehomo}_0 | \by_N, \bY_0] = v^\ehz_0. 
\end{align} 
(ii) \emph{[\VT~model]} Under Assumptions~\ref{assump:vt.outcome}--\ref{assump:vt.errors}, we have 
\begin{align} 
	\Ex[ \bhSigma_N^\ehomo | \by_T, \bY_0] = \bSigma^\evt_N
	\quad \text{and} \quad 
	\Ex[ \hv^{\evt, \ehomo}_0 | \by_T, \bY_0] = v^\evt_0.
\end{align}  
(iii) \emph{[Mixed model]} Under Assumptions~\ref{assump:mixed.outcome}--\ref{assump:mixed.errors}, we have 
\begin{align}
	\Ex[\bhSigma_T^\ehomo | \bY_0] =  \bSigma^\emix_T,
	\quad 
	\Ex[ \bhSigma_N^\ehomo | \bY_0] = \bSigma^\emix_N,
	\quad \text{and} \quad 
	\Ex[ \hv^{\emix, \ehomo}_0 | \bY_0] = v^\emix_0.
\end{align}
\end{lemma} 
Lemma~\ref{lemma:inference.1} is a well known result within the OLS literature, albeit it is typically formalized under the stricter full column rank assumption. 

\smallskip 
\noindent 
{\bf Heteroskedastic errors.}  
We adopt two strategies for the heteroskedastic setting. 

\smallskip 
\noindent 
{\em I: Jackknife.} 
The first estimator is based on the jackknife. 
Traditionally, the jackknife estimates the covariance of the regression coefficients $(\bhalpha, \bhbeta)$. 
By analyzing said estimates, we derive the following: 
\begin{align}
	\bhSigma_T^\jack &= 
	\text{diag}\left( \left[ \bH^u_\perp \circ  \bH^u_\perp \circ \bI \right]^\dagger \left[  \bH^u_\perp \by_T \circ  \bH^u_\perp \by_T \right] \right) \label{eq:hz.var.2} 
	\\ 
	\bhSigma_N^\jack &= 
	\text{diag}\left( \left[ \bH^v_\perp \circ  \bH^v_\perp \circ \bI  \right]^\dagger \left[ \bH^v_\perp \by_N \circ \bH^v_\perp \by_N \right]\right).  \label{eq:vt.var.2} 
\end{align} 

\begin{lemma} \label{lemma:inference.2} 
Consider heteroskedastic errors. 
(i) \emph{[\HZ~model]} Let Assumptions~\ref{assump:hz.outcome}--\ref{assump:hz.errors} hold.  
If $(\bH^u_\perp \circ  \bH^u_\perp \circ \bI)$ is nonsingular, then 
\begin{align}
	\Ex[\bhSigma^\ejack_T | \by_N, \bY_0] \succeq \bSigma^\ehz_T 
	\quad \text{and} \quad 
	\Ex[ \hv^{\ehz, \ejack}_0 | \by_N, \bY_0] \ge v^\ehz_0. 
\end{align} 
(ii) \emph{[\VT~model]} Let Assumptions~\ref{assump:vt.outcome}--\ref{assump:vt.errors} hold. 
If $(\bH^v_\perp \circ  \bH^v_\perp \circ \bI)$ is nonsingular, then
\begin{align} 
	\Ex[ \bhSigma_N^\ejack | \by_T, \bY_0] \succeq \bSigma^\evt_N
	\quad \text{and} \quad 
	\Ex[ \hv^{\evt, \ejack}_0 | \by_T, \bY_0] \ge v^\evt_0.
\end{align}  
(iii) \emph{[Mixed model]} Let Assumptions~\ref{assump:mixed.outcome}--\ref{assump:mixed.errors} hold. 
If $(\bH^u_\perp \circ  \bH^u_\perp \circ \bI)$ and $(\bH^v_\perp \circ  \bH^v_\perp \circ \bI)$ are nonsingular, then
\begin{align}
	\Ex[\bhSigma_T^\ejack | \bY_0] \succeq \bSigma^\emix_T,
	\quad 
	\Ex[ \bhSigma_N^\ejack | \bY_0] \succeq \bSigma^\emix_N,
	\quad \text{and} \quad 
	\Ex[ \hv^{\emix, \ejack}_0 | \bY_0] \ge v^\emix_0.
\end{align}
\end{lemma} 
Lemma~\ref{lemma:inference.2} establishes that the jackknife is conservative, provided $(\bH^u_\perp \circ  \bH^u_\perp \circ \bI)$ and $(\bH^v_\perp \circ  \bH^v_\perp \circ \bI)$ are nonsingular. 
Strictly speaking, the jackknife is well defined if these quantities are singular, as seen through the pseudoinverse in \eqref{eq:hz.var.2} and \eqref{eq:vt.var.2}. 
Lemma~\ref{lemma:inference.2} considers the nonsingular case for simplicity.  
We remark that $\max_\ell H^u_{\ell \ell} < 1$ and $\max_\ell H^v_{\ell \ell} < 1$ are sufficient conditions for invertibility. 
%

\smallskip 
\noindent 
{\em II: HRK-estimator.} 
Next, we consider the covariance estimator proposed by \cite{variance}. 
We index this estimator by the authors, Hartley-Rao-Kiefer: 
\begin{align}
	\bhSigma_T^\hrk &=  \text{diag}\left( \left[ \bH^u_\perp \circ  \bH^u_\perp \right]^{-1} \left[  \bH^u_\perp \by_T \circ  \bH^u_\perp \by_T \right] \right)  \label{eq:hz.var.3} 
	\\ 
	\bhSigma_N^\hrk &= \text{diag}\left( \left[ \bH^v_\perp \circ \bH^v_\perp \right]^{-1} \left[ \bH^v_\perp \by_N \circ \bH^v_\perp \by_N \right] \right). \label{eq:vt.var.3} 
\end{align} 
\begin{lemma} \label{lemma:inference.3} 
Consider heteroskedastic errors. 
(i) \emph{[\HZ~model]} Let Assumptions~\ref{assump:hz.outcome}--\ref{assump:hz.errors} hold.  
If $(\bH^u_\perp \circ  \bH^u_\perp)$ is nonsingular, then 
\begin{align}
	\Ex[\bhSigma^\ehrk_T | \by_N, \bY_0] = \bSigma^\ehz_T 
	\quad \text{and} \quad 
	\Ex[ \hv^{\ehz, \ehrk}_0 | \by_N, \bY_0] = v^\ehz_0. 
\end{align} 
(ii) \emph{[\VT~model]} Let Assumptions~\ref{assump:vt.outcome}--\ref{assump:vt.errors} hold. 
If $(\bH^v_\perp \circ  \bH^v_\perp)$ is nonsingular, then
\begin{align} 
	\Ex[ \bhSigma_N^\ehrk | \by_T, \bY_0] = \bSigma^\evt_N
	\quad \text{and} \quad 
	\Ex[ \hv^{\evt, \ehrk}_0 | \by_T, \bY_0] = v^\evt_0.
\end{align}  
(iii) \emph{[Mixed model]} Let Assumptions~\ref{assump:mixed.outcome}--\ref{assump:mixed.errors} hold. 
If $(\bH^u_\perp \circ  \bH^u_\perp)$ and $(\bH^v_\perp \circ  \bH^v_\perp)$ are nonsingular, then
\begin{align}
	\Ex[\bhSigma_T^\ehrk | \bY_0] = \bSigma^\emix_T,
	\quad 
	\Ex[ \bhSigma_N^\ehrk | \bY_0] = \bSigma^\emix_N,
	\quad \text{and} \quad 
	\Ex[ \hv^{\emix, \ehrk}_0 | \bY_0] = v^\emix_0.
\end{align}
\end{lemma} 
Lemma~\ref{lemma:inference.3} establishes that the HRK estimator is unbiased, provided $(\bH^u_\perp \circ  \bH^u_\perp)$ and $(\bH^v_\perp \circ \bH^v_\perp)$ are invertible. 
For the former quantity, $\max_\ell H^u_{\ell \ell} < 1/2$ is a sufficient condition for invertibility.  
Since $\tr(\bH^u) = R$, this restricts $R < N_0/2$.
A similar conclusion is drawn for \VT~regression. 

\subsection{Design-Based Inference} 
This section studies the counterfactual prediction from a design-based perspective, whereby the potential outcomes are considered fixed and the treatment assignments are considered stochastic. 
%
%

\subsubsection{Assumptions} 
In congruence with the article thus far, we focus on a single treated unit and treated period. 
We will find it useful to separate the assignment mechanism into the selection of each component. 
Accordingly, let $\bA \in \{0,1\}^T$ with $\bone' \bA = 1$ and $\bB \in \{0,1\}^N$ with $\bone' \bB = 1$ be the indicator vectors for the treated time period and treated unit, respectively, 
e.g., $B_N = 1$ and $A_T = 1$ if unit $N$ is treated at time $T$.  
With this notation, we denote the realized outcome as $Y_{it} = B_i A_t Y_{it}(1) + (1 - B_i A_t) Y_{it}(0)$. 
Following \cite{guido_inf}, we consider the following assignment mechanisms: 
\begin{assumption} [Random assignment of time period] \label{assump:time.db} 
\begin{align}
	\Pb(\bA = \ba) = \begin{cases}
		& 1/T, \quad \text{if } a_t \in \{0,1\} ~\forall t, ~\bone' \ba = 1
		\\
		& 0, \quad \text{otherwise}. 
	\end{cases} 
\end{align} 
\end{assumption} 

\begin{assumption} [Random assignment of unit] \label{assump:unit.db} 
\begin{align}
	\Pb(\bB = \bb) = \begin{cases}
		& 1/N, \quad \text{if } b_i \in \{0,1\} ~\forall i, ~\bone' \bb = 1
		\\
		& 0, \quad \text{otherwise}. 
	\end{cases} 
\end{align} 
\end{assumption} 
Assumption~\ref{assump:time.db} considers the treated period to be randomly selected while Assumption~\ref{assump:unit.db} considers the treated unit to be randomly selected. 
As \cite{guido_inf} notes, these assumptions are not always plausible, but they underlie the placebo tests that are commonly used in synthetic controls applications. 

\subsubsection{Estimator} 
To conduct design-based analysis, we consider all possible treatment assignments,  not only the realized assignment. 
Let $Y^*_{it}(0)$ be the OLS fit of $Y_{it}(0)$ based on $[Y_{j\tau} : j \neq i, \tau < t]$, $[Y_{i \tau} : \tau < t]$, and $[Y_{j t}: j \neq i]$. 
We define the design-based estimator as  
\begin{align}
	\hY(0) = \sum_{i \le N} \sum_{t \le T} B_i A_t Y^*_{it}(0). \label{eq:estimator.db} 
\end{align} 
In words, \eqref{eq:estimator.db} predicts the mean counterfactual outcome under control for unit $i$ at time $t$ if unit $i$ is treated at time period $t$, i.e., $B_i = 1$ and $A_t = 1$. 
We reemphasize that the stochasticity of $\hY(0)$ stems from the treatment assignment mechanism since $Y^*_{it}(0)$ is a fixed quantity.  
As such, while the model-and design-based estimators share the same point estimate for the realized assignment, they differ in their formulations and attributions of randomness. 

\subsubsection{Inferential Properties} 
In Table~\ref{table:comparison}, we summarize the estimands associated with the model-based and design-based estimators under three sources of randomness: (i) time, (ii) unit, and (iii) time and unit. 
Within the model-based framework, mechanisms (i)--(iii) correspond to the \HZ~model (Assumptions~\ref{assump:hz.outcome}--\ref{assump:hz.errors}), 
\VT~model (Assumptions~\ref{assump:vt.outcome}--\ref{assump:vt.errors}), 
and mixed model (Assumptions~\ref{assump:mixed.outcome}--\ref{assump:mixed.errors}), respectively. 
Within the design-based framework, mechanisms (i)--(iii) correspond to 
Assumption~\ref{assump:time.db}, Assumption~\ref{assump:unit.db}, and both Assumptions~\ref{assump:time.db} and \ref{assump:unit.db}, respectively. 
Accordingly, the model-based and design-based expectations are taken over different probability measures. 

\begin{table} [!h]
\centering 
\caption{
Model-based and design-based estimands under different sources of randomness. 
We use the shorthand $\btalpha^* = \bH^v \balpha^*$ and $\btbeta^* = \bH^u \bbeta^*$ and consider the realized $(N,T)$th assignment.}
\begin{tabular}{@{}ccc@{}}
\toprule
\midrule 
\multirow{1}{*}{Source of Randomness}
& \multicolumn{1}{c}{Model-Based Estimand} 
& \multicolumn{1}{c}{Design-Based Estimand}    	
\\
\cmidrule(){1-3} 
Time  
& $\Ex[\hY_{NT}(0) | \by_N, \bY_0] = \sum_{t \le T_0} \talpha^*_t Y_{N t}$
& $\Ex[\hY(0) | \bB] = \frac{1}{T} \sum_{t \le T} Y^*_{Nt}(0)$ 
\\
Unit  
& $\Ex[\hY_{NT}(0) | \by_T, \bY_0] = \sum_{i \le N_0} \tbeta^*_i Y_{iT}$
& $\Ex[\hY(0) | \bA] = \frac{1}{N} \sum_{i \le N} Y^*_{iT}(0)$ 
\\
Time and Unit  
& $\Ex[\hY_{NT}(0) | \bY_0] = \sum_{i \le N_0} \sum_{t \le T_0} \alpha^*_t \beta^*_i Y_{i t}$
& $\Ex[\hY(0) ] = \frac{1}{NT} \sum_{i \le N} \sum_{t \le T} Y^*_{it}(0)$ 
\\
\bottomrule
\end{tabular}
\label{table:comparison}
\end{table}

Let us compare the estimands in Table~\ref{table:comparison}. 
In words, $\Ex[\hY_{NT}(0) | \by_N, \bY_0]$ is a weighted combination of outcomes under control for the treated unit across all $T_0$ pretreatment periods;  
$\Ex[\hY(0) | \bB]$ is a simple average of fitted outcomes under control for the treated unit across all $T$ periods, which \cite{guido_inf} calls the ``\HZ'' effect. 
Similarly, $\Ex[\hY_{NT}(0) | \by_T, \bY_0]$ is a weighted combination of outcomes under control for the treated period across all $N_0$ control units;
$\Ex[\hY(0) | \bA]$ is a simple average of fitted outcomes under control for the treated period across all $N$ units, also called the ``\VT'' effect. 
Finally, $\Ex[\hY_{NT}(0) | \bY_0]$ is a weighted combination of outcomes under control across all $N_0$ control units and $T_0$ pretreatment periods; 
$\Ex[\hY(0)]$ is a simple average of fitted outcomes under control across all unit and time period pairs, which we coin the ``mixed'' effect. 

Though our design-based analysis is brief, there two important takeaways: 
(i) the model-based and design-based estimators recover similar estimands for each source of randomness; 
and (ii) different sources of randomness lead to different estimands, which is consistent with our model-based insights. 
The connection between assumptions of randomness and resulting estimand has been previously noticed in related contexts, e.g., \cite{guido_inf0, guido_inf} and \cite{jas_inf}. 


\subsection{Discussion} 
Several remarks on this section's results and extensions are in order. 

\begin{remark} [Correct Specification] 
For expositional convenience, we assume correct specification of the outcome model, as in \cite{li2020}, to explain the panel data intuition and main theoretical results in a simple and transparent fashion. 
Alternatively, we can interpret Assumptions~\ref{assump:hz.outcome}, \ref{assump:vt.outcome}, and \ref{assump:mixed.outcome} as linear prediction models that are detached from structural meanings a la \cite{cattaneo} and \cite{sc_lasso1, sc_t_test}. 
For instance, within the mixed framework, we can redefine $\balpha^* = \argmin \Ex[ \| \by_T - \bY_0 \balpha \|_2^2 | \bY_0]$ as the best linear approximation of $\by_T$ based on $\bY_0$, conditional on $\bY_0$. 
Such models can be justified via factor models and vector autoregressive models \citep{sc_t_test}.  
\end{remark} 

\begin{remark} [Mixed Model] 
Any estimator, whether it falls within the symmetric or asymmetric class, can be studied under the mixed model (Assumptions~\ref{assump:mixed.outcome}--\ref{assump:mixed.errors}). 
There are also numerous ways to encode the mixed perspective beyond our postulation, e.g., random assignment of periods and units (Assumptions~\ref{assump:time.db}--\ref{assump:unit.db}) of \cite{guido_inf}. 
\end{remark} 

\begin{remark} [Mixed Variance Estimators] \label{remark:mixed}
%
%
We highlight that Lemmas~\ref{lemma:inference.1}--\ref{lemma:inference.3} only hold in expectation. 
For any particular realization, $\hv^\mix_0$ may exhibit unexpected properties. 
For instance, if  
$\tr(\bY_0^\dagger \bhSigma_T (\bY'_0)^\dagger \bhSigma_N) > \max\{\hv_0^\hz, \hv_0^\vt \}$, 
then $\hv^\mix_0 < \min\{\hv_0^\hz, \hv_0^\vt \}$; thus, the mixed coverage will be smaller than both \HZ~and \VT~coverages. 
In fact, $\hv^\mix_0$ can be negative if $\tr(\bY_0^\dagger \bhSigma_T (\bY'_0)^\dagger \bhSigma_N) > \hv_0^\hz + \hv_0^\vt$, which may occur if both \HZ~and \VT~in-sample errors are ``too large''. 
For these scenarios, one na\"ive solution is to modify $\hv^\mix_0$ as   
$\hv_0^{\mix} \leftarrow \hv_0^\hz + \hv_0^\vt$, which is conservative by Lemmas~\ref{lemma:inference.1}--\ref{lemma:inference.3}. 
However, this case is arguably better resolved with a different point estimator altogether. 
\end{remark} 


\begin{remark} [Bridging Model-Based and Design-Based Inferences] 
The takeaways from the model-based and design-based analyses are consistent with one another. 
However, a formal and complete connection between the two perspectives on inference remains to be established. 
Towards this, we highlight that \cite{lin_ols} demonstrates that one set of model-based confidence intervals (i.e., the Huber-White sandwich estimator) are justifiable under design-based arguments. 
In this view, a fascinating line of future inquiry is to analyze whether similar arguments hold for the heteroskedastic confidence intervals presented in Section~\ref{sec:ci} or other related model-based constructions. 
%
\end{remark} 

\begin{remark} [On the Role of Randomness] 
This section underscores that the source of randomness plays a pivotal role for conducting inference. 
%
Translated to practice, our results stress that researchers should be scrupulous in reasoning through where the randomness in their data comes from.  
For instance, in the Basque study, some researchers may find it more plausible that the randomness is over time rather than over space, e.g., it is more conceivable that the onset of terrorism could have occurred in a different year but less conceivable that it could have occurred in a different region of Spain but not in the Basque Country. 
We do not take a substantive view on the matter, but we do stress that these decisions are meaningful for proper analysis as they have immediate implications for the resulting estimand and inferential procedure. 
\end{remark} 

\begin{remark} [Model Checking] 
It is possible in an application that substantive knowledge does not make clear on the source of randomness and hence which model to use, i.e.,  \HZ, \VT, or mixed. 
%
%
In such scenarios, the ``in-space'' and ``in-time'' placebo tests  (a la cross-validation) proposed in \cite{abadie2, abadie4} are attractive tools to analyze the prediction properties of the various estimators under consideration. 
More generally, practices established within \cite{pcs} provide an organized framework based on the principles of predictability, computability, and stability (PCS) to conduct rigorous comparison analyses. 
\end{remark} 

\begin{remark} [Extension to PCR] \label{remark:pcr} 
The previous results immediately extend to PCR by replacing $\bY_0$ with $\bY_0^{(k)}$, as defined in \eqref{eq:rank_k}, for any $k < R$. 
Intuitively, PCR-based models operate under the belief that the data is inherently low-dimensional. 
We comment on several benefits of PCR over OLS. 
To begin, the \HZ~and \VT~OLS variance estimators constructed in Section~\ref{sec:ci} can suffer from degeneracy when $N$ and $T$ are of different sizes. 
That is, if $N < T$, then the \HZ~in-sample error is likely zero (otherwise known as overfitting), which causes the \HZ~coverage to collapse on the point estimate; analogous statements hold for the \VT~coverage when $N > T$.  
The PCR-based variance estimators, on the other hand, can avoid degeneracy through the number of chosen principal components $k$ (regularization). 
On a related note, the nonsingularity conditions required for the jackknife and HRK variance estimators can also be by controlled by $k$. 
\end{remark}

\section{Illustrations}  \label{sec:illustrations} 
This section illustrates key concepts developed in this article. 
Our report is based on three canonical synthetic controls studies: (i) terrorism in Basque Country, (ii) California's Proposition 99 \citep{abadie2}, and (iii) the reunification of West Germany \citep{abadie4}. 
In particular, we will conduct a model-based analysis using the confidence intervals developed in Section~\ref{sec:ci}. 
%
We provide an overview of the results and relegate details (e.g., implementation) to Appendix~\ref{sec:sim_emp}. 

\subsection{Background on Case Studies} 
%
\noindent 
{\bf Basque study.} 
See Sections~\ref{sec:intro}--\ref{sec:setup} for details. 

\smallskip \noindent 
{\bf California study.} 
This study examines the effect of California's Proposition 99, an anti-tobacco legislation, on its tobacco consumption. 
The panel data contains per capita cigarette sales of $N = 39$ U.S. states over $T=31$ years. 
There are $T_0=18$ pretreatment observations and $N_0 = 38$ control units. 
Our interest is to estimate California's cigarette sales in the absence of Proposition 99. 

\smallskip \noindent
{\bf West Germany study.}
This study examines the economic impact of the 1990 reunification in West Germany.   
The panel data contains per capita GDP of $N=17$ countries over $T=44$ years. 
There are $T_0 = 30$ pretreatment observations and $N_0=16$ control units. 
Our interest is to estimate West Germany's GDP in the absence of reunification. 

\subsection{Data-Inspired Simulation Studies} \label{sec:simulations.1} 
We look to better understand the trade-offs in conducting inference under different sources of randomness. 
In an attempt to document our analysis in a realistic environment, we calibrate our simulations to our three studies. 

\subsubsection{Data Generating Process} 
We consider the single treated unit and time period setting. 
Specifically, we consider the actual treated unit, e.g., Basque Country, and focus on the first post-treatment period, e.g., one year after the outset of terrorism; hence, $T \leftarrow T_0 + 1$. 
Using the actual data, we generate the underlying regression models as  
\begin{align}
	\balpha^* = \argmin_{\balpha} \| \by^*_T - \bY^*_0 \balpha \|_2^2 
	\quad \text{and} \quad 
	\bbeta^* = \argmin_{\bbeta} \| \by^*_N - (\bY^*_0)' \bbeta \|_2^2, 
\end{align} 
where  
$\by^*_N = [Y_{Nt}: t \le T_0]$,
$\by^*_T = [Y_{iT}: i \le N_0]$, 
and $\bY^*_0 = [Y_{it}: i \le N_0, t \le T_0]$. 

Observationally, we have access to the following quantities. 
Let $\bY_0$ be the rank $r$ approximation of $\bY^*_0$, where $r$ is chosen as the minimum number of singular values needed to capture at least $99.9\%$ of $\bY^*_0$'s spectral energy. 
Next, we sample 
$\by_T \sim \mathcal{N}(\bY_0 \balpha^*, (N_0 - r)^{-1} \| \by^*_T - \bY^*_0 \balpha^* \|_2^2 \bI)$ 
and 
$\by_N \sim \mathcal{N}(\bY'_0 \bbeta^*, (T_0 - r)^{-1} \| \by^*_N - (\bY^*_0)' \bbeta^* \|_2^2 \bI)$. 
We then define three estimands: 
(i) $\mu_0^\hz = \langle \by_N, \bH^v \balpha^* \rangle$, 
(ii) $\mu_0^\vt = \langle \by_T, \bH^u \bbeta^* \rangle$,
and
(iii) $\mu_0^\mix = \langle \balpha^*, \bY'_0 \bbeta^* \rangle$, 
where $(\bH^u, \bH^v)$ are computed from $\bY_0$. 

\subsubsection{Simulation Results} 
For the purposes of stability, we conduct 500 replications of the above DGP for each study. 
In the $\ell$th simulation repeat, we learn the regression coefficients as  
\begin{align}
	\bhalpha^{(\ell)} = \argmin_{\balpha} \| \by^{(\ell)}_T - \bY_0 \balpha \|_2^2 
	\quad \text{and} \quad 
	\bhbeta^{(\ell)} = \argmin_{\bbeta} \| \by^{(\ell)}_N - \bY'_0 \bbeta \|_2^2.  
\end{align} 
The corresponding point estimate is defined as $\hY^{(\ell)}_{NT}(0) = \langle \by^{(\ell)}_N, \bhalpha^{(\ell)} \rangle = \langle \by_T^{(\ell)}, \bhbeta^{(\ell)} \rangle$. 
We then construct separate \HZ, \VT, and mixed homoskedastic confidence intervals based on \eqref{eq:hz.var.1} and \eqref{eq:vt.var.1} around the point estimate. 

In Table~\ref{table:coverage}, we report the coverage probabilities (CP) and average lengths (AL) for each confidence interval with respect to each estimand at the $95\%$ nominal mark. 
With respect to $\mu_0^\hz$, the coverage of the \HZ~confidence interval is closer to the nominal coverage than that of the \VT~and mixed intervals as the latter two can substantially under- or over-cover. 
This storyline is consistent for the \VT~interval with respect to $\mu_0^\vt$ and the mixed interval with respect to $\mu_0^\mix$. 

Collectively, our formal results and simulations demonstrate that 
(i) the choice of estimand directly affects the accuracy of the inference; 
and  
(ii) the variance formulas developed for one estimand may not have the correct coverage for another estimand. 
Accordingly, researchers should carefully consider the source of randomness in their data as it can have a significant influence over their ability to conduct valid inference. 
We comment that these conclusions are in line with those drawn in \cite{jas_inf}, which analyzes the classical difference-in-means estimator with respect to standard estimands for randomized control trials.  

\begin{table} [!t]
\centering 
\caption{Coverage for nominal 95\% confidence intervals across 500 replications. The coverage length is normalized by the magnitude of the corresponding point estimate.}
\begin{tabular}{@{}cccccccccc@{}}
\toprule
\midrule 
\multirow{2}{*}{Case study}
& \multicolumn{3}{c}{$\hv_0^\hz$} 
& \multicolumn{3}{c}{$\hv_0^\vt$}   
& \multicolumn{3}{c}{$\hv_0^\mix$}   
\\ 
\cmidrule(l){2-4} 
\cmidrule(l){5-7} 
\cmidrule(l){8-10}
& $\mu_0^\hz$  & $\mu_0^\vt$   & $\mu_0^\mix$ 
& $\mu_0^\hz$  & $\mu_0^\vt$   & $\mu_0^\mix$ 
& $\mu_0^\hz$  & $\mu_0^\vt$   & $\mu_0^\mix$  	
\\
\cmidrule(){1-10} 
Basque (CP)
& $0.92$    & $0.74$    & $0.63$ 
& $0.99$    & $0.93$    & $0.88$ 
& $1.00$    & $0.97$    & $0.94$ 
\\ 
Basque (AL)
& $0.02$    & $0.02$    & $0.02$ 
& $0.03$    & $0.03$    & $0.03$ 
& $0.04$     & $0.04$    & $0.04$ 
\\
\cmidrule(){1-10} 
California (CP) 
& $0.95$    & $1.00$    & $0.92$ 
& $0.64$    & $0.93$    & $0.60$  
& $0.98$    & $1.00$    & $0.95$
\\ 
California (AL)
& $0.07$    & $0.07$    & $0.07$ 
& $0.03$    & $0.03$    & $0.03$ 
& $0.08$    & $0.08$    & $0.08$ 
\\
\cmidrule(){1-10} 
W. Germany (CP)
& $0.94$    & $1.00$    & $0.93$ 
& $0.49$    & $0.94$    & $0.49$ 
& $0.96$    & $1.00$    & $0.95$ 
\\ 
W. Germany (AL)
& $0.03$    & $0.03$    & $0.03$ 
& $0.01$    & $0.01$    & $0.01$ 
& $0.03$    & $0.03$    & $0.03$ 
\\
\bottomrule
\end{tabular}
\label{table:coverage}
\end{table}

\subsection{Empirical Applications} 
Next, we analyze our three case studies of interest.  
All regression models are built on pretreatment data only, and the point and variance estimation formulas are separately applied for the treated unit at each post-treatment period $t > T_0$. 

\subsubsection{Point Estimation} 
Figure~\ref{fig:estimation} visualizes the counterfactual trajectories generated by the estimators in  Section~\ref{sec:est.notable}. 
Our findings reinforce Theorems~\ref{thm:equivalent} and \ref{thm:not_equivalent}. 
On a separate note, we observe that, within the Basque study, the OLS estimates are wildly different from the other estimates and \HZ~simplex regression reduces to the last observation carried forward (LOCF) estimator. 
In the California and West Germany studies, the estimates are all qualitatively similar with the exception of the \HZ~simplex regression, which again reduces to LOCF. 
In fact, the OLS and ridge estimates appear to overlap, as well as the lasso and elastic net estimates. 

\begin{figure} [!t]
	\centering
	\begin{subfigure}[b]{0.32\textwidth}
		\centering 
		\includegraphics[width=\linewidth]
		{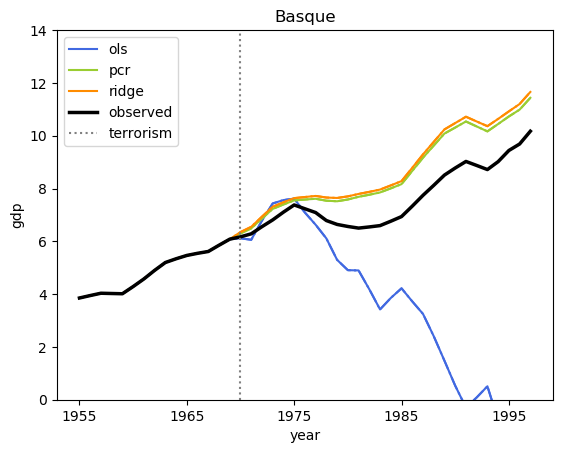}
		\label{fig:sym_basque} 
	\end{subfigure} 
	\begin{subfigure}[b]{0.32\textwidth}
		\centering 
		\includegraphics[width=\linewidth]
		{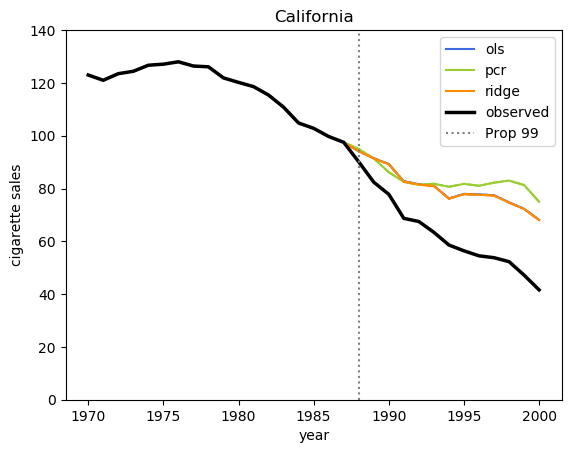}
		\label{fig:sym_prop99} 
	\end{subfigure} 
	\begin{subfigure}[b]{0.32\textwidth}
		\centering 
		\includegraphics[width=\linewidth]
		{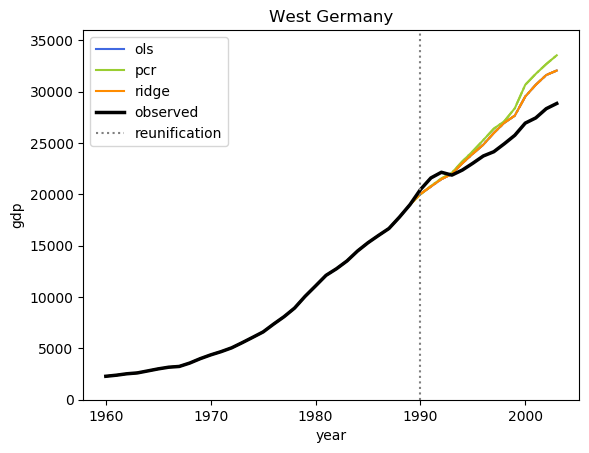}
		\label{fig:sym_germany} 
	\end{subfigure} 
	\\
	\begin{subfigure}[b]{0.32\textwidth}
		\centering 
		\includegraphics[width=\linewidth]
		{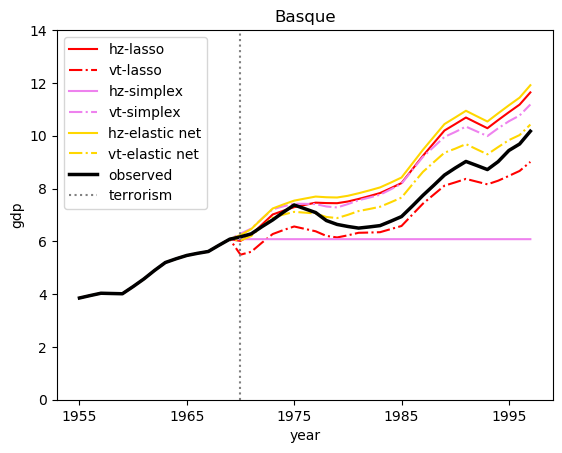}
		\label{fig:asym_basque} 
	\end{subfigure} 
	\begin{subfigure}[b]{0.32\textwidth}
		\centering 
		\includegraphics[width=\linewidth]
		{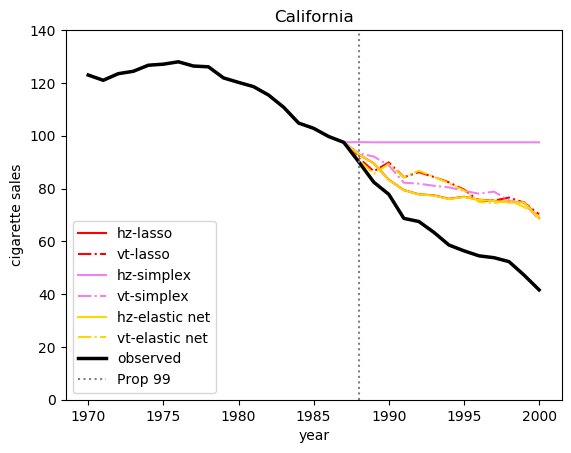}
		\label{fig:asym_prop99} 
	\end{subfigure} 
	\begin{subfigure}[b]{0.32\textwidth}
		\centering 
		\includegraphics[width=\linewidth]
		{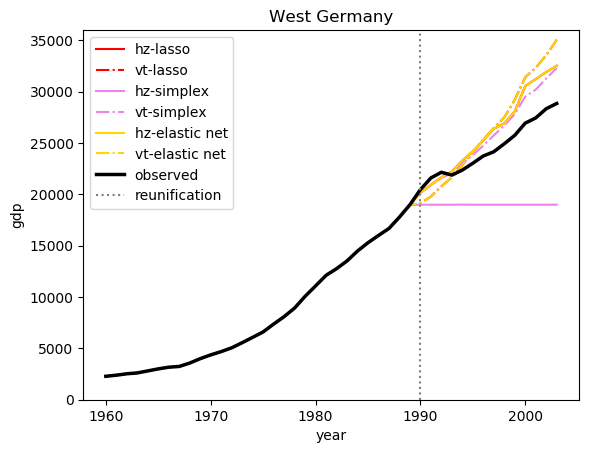}
		\label{fig:asym_germany} 
	\end{subfigure} 
	\caption{
	Top and bottom figures correspond to symmetric and asymmetric estimators, respectively. 
	From left to right, the figures are indexed by the Basque, California, and West Germany studies. 
	Across all figures, the treated year is the dotted vertical line; the observed trajectory is in solid black; and the \HZ~and \VT~counterfactual trajectories are colored solid and dashed-dotted lines, respectively.} 
	\label{fig:estimation} 
\end{figure} 

\subsubsection{Inference} 
To reduce visual redudancy, we only present the figures associated with the jackknife-based intervals for OLS and PCR in Figures~\ref{fig:ols_jack} and \ref{fig:pcr_jack}, respectively. 
Consider the Basque study. 
The top row of plots in both figures demonstrates that $\mu_0^\vt$ is more accurately estimated than both $\mu_0^\hz$ and $\mu_0^\mix$.
Put differently, there is less uncertainty about conducting inference on $\mu_0^\vt$ relative to the other estimands. 
At the same time, these plots indicate that if $\mu_0^\hz$ or $\mu_0^\mix$ are the estimands of interest, then the \VT~confidence interval will undercover in both settings. 
Analogous statements can be made for the remaining subfigures. 
As with our simulations, the large potential differences in coverage reinforce the importance of properly reasoning through the source of randomness in the data.  

\pagebreak 
\begin{figure} [!t]
	\centering 
	\begin{subfigure}[b]{0.32\textwidth}
		\includegraphics[width=\linewidth]
		{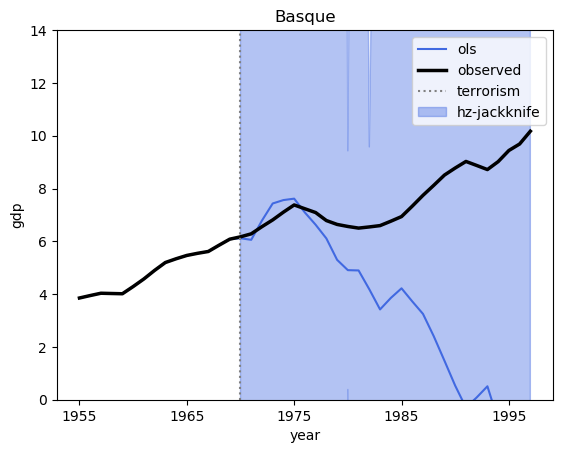}
		\label{fig:basque_hz_jack}
	\end{subfigure} 
	\begin{subfigure}[b]{0.32\textwidth}
		\includegraphics[width=\linewidth]
		{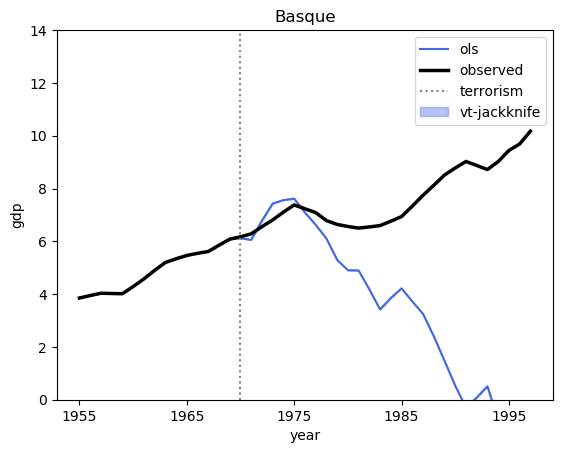}
		\label{fig:basque_vt_jack} 
	\end{subfigure} 
	\begin{subfigure}[b]{0.32\textwidth}
		\includegraphics[width=\linewidth]
		{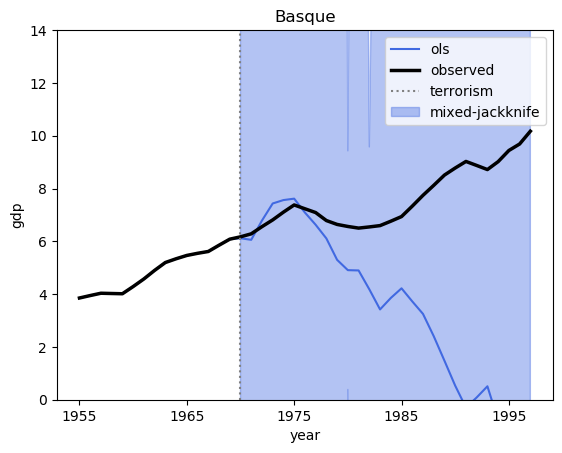}
		\label{fig:basque_mixed_jack}
	\end{subfigure} 	
	\\
	\begin{subfigure}[b]{0.32\textwidth}
		\centering 
		\includegraphics[width=\linewidth]
		{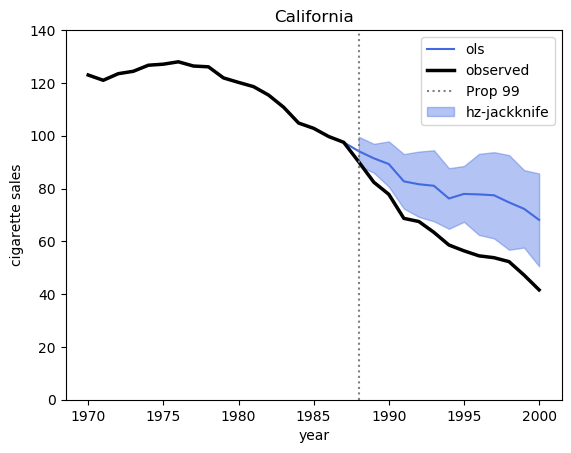}
		\label{fig:prop99_hz_jack}
	\end{subfigure} 
	\begin{subfigure}[b]{0.32\textwidth}
		\centering 
		\includegraphics[width=\linewidth]
		{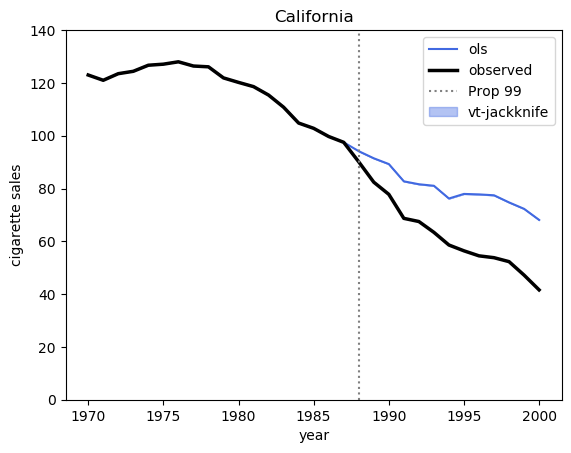}
		\label{fig:prop99_vt_jack}
	\end{subfigure} 
	\begin{subfigure}[b]{0.32\textwidth}
		\centering 
		\includegraphics[width=\linewidth]
		{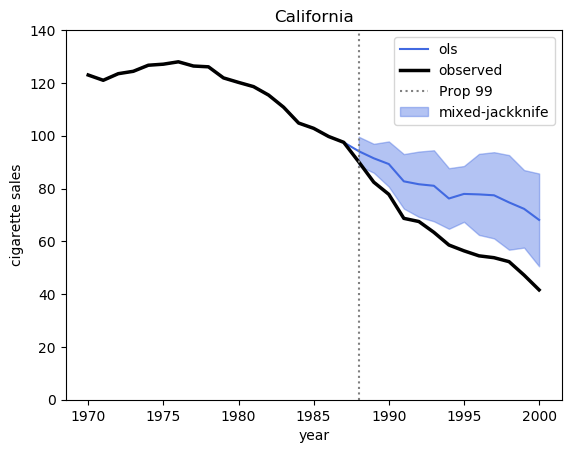}
		\label{fig:prop99_mixed_jack}
	\end{subfigure} 	
	\\
	\begin{subfigure}[b]{0.32\textwidth}
		\includegraphics[width=\linewidth]
		{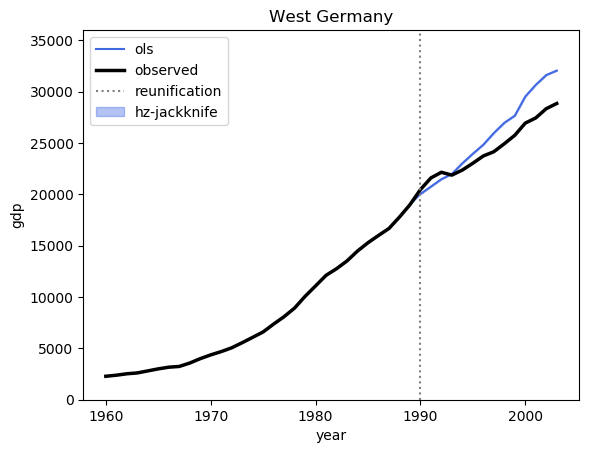}
		\caption{\HZ~model.} 
		\label{fig:germany_hz_jack}
	\end{subfigure} 
	\begin{subfigure}[b]{0.32\textwidth}
		\includegraphics[width=\linewidth]
		{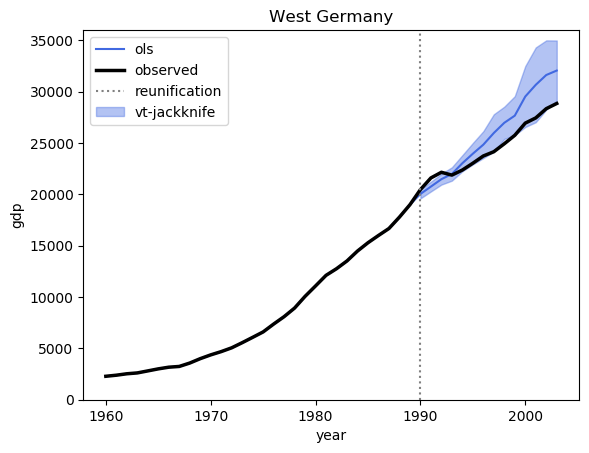}
		\caption{\VT~model.} 
		\label{fig:germany_vt_jack} 
	\end{subfigure} 
	\begin{subfigure}[b]{0.32\textwidth}
		\includegraphics[width=\linewidth]
		{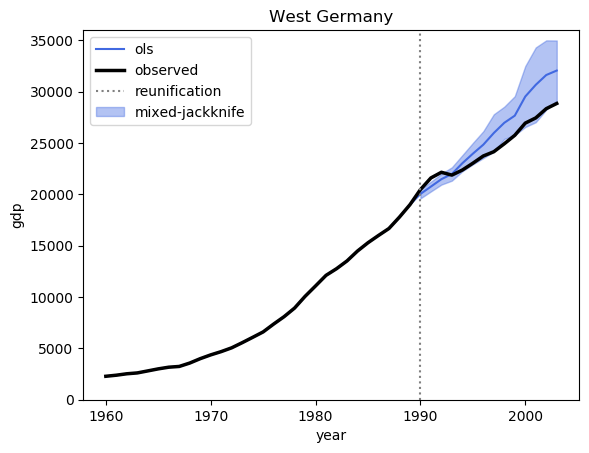}
		\caption{Mixed model.} 
		\label{fig:germany_mixed_jack}
	\end{subfigure} 	
	\caption{
	OLS estimates with jackknife confidence intervals. 
	From top to bottom, the rows are indexed by the Basque, California, and West Germany studies. 
	From left to right, the columns are indexed by the \HZ, \VT, and mixed models. 
	Figure~\ref{fig:ols_jack} visualizes the problem of degeneracy for OLS-based intervals that is discussed in Remark~\ref{remark:pcr}.
	}
	\label{fig:ols_jack} 
\end{figure} 

\begin{figure} [!t]
	\centering 
	\begin{subfigure}[b]{0.32\textwidth}
		\includegraphics[width=\linewidth]
		{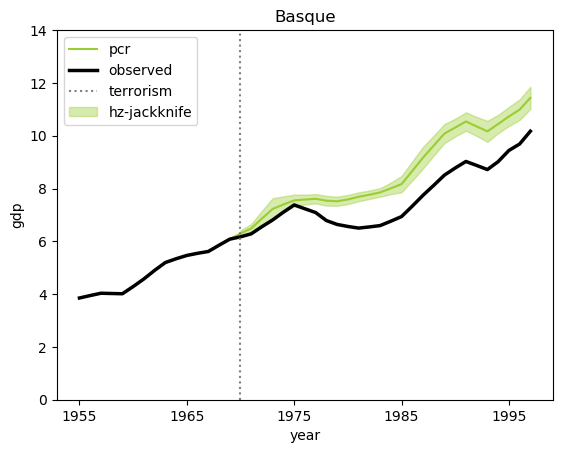}
		\label{fig:pcr_basque_hz_jack}
	\end{subfigure} 
	\begin{subfigure}[b]{0.32\textwidth}
		\includegraphics[width=\linewidth]
		{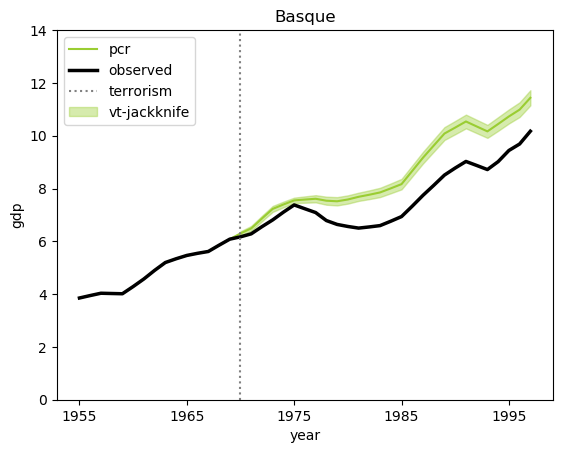}
		\label{fig:pcr_basque_vt_jack} 
	\end{subfigure} 
	\begin{subfigure}[b]{0.32\textwidth}
		\includegraphics[width=\linewidth]
		{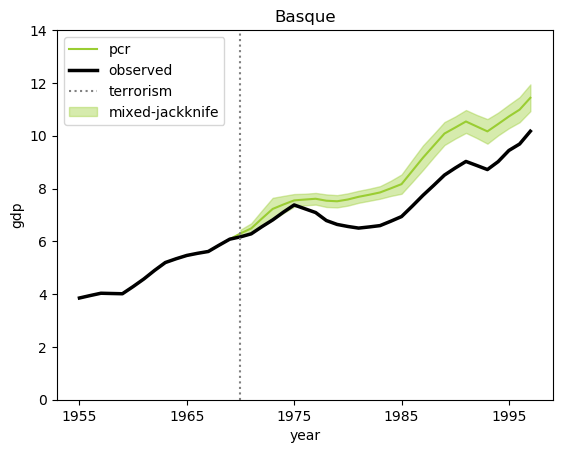}
		\label{fig:pcr_basque_mixed_jack}
	\end{subfigure} 	
	\\
	\begin{subfigure}[b]{0.32\textwidth}
		\centering 
		\includegraphics[width=\linewidth]
		{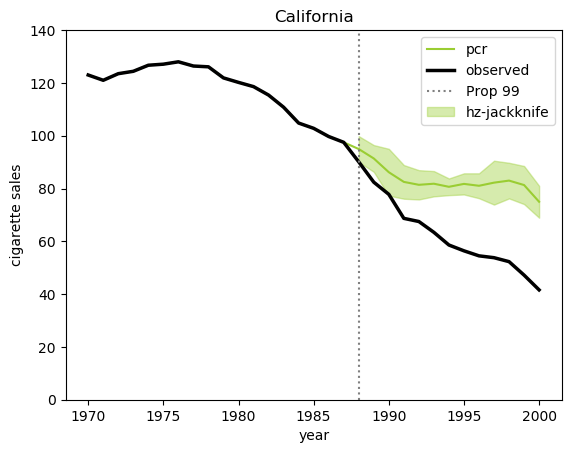}
		\label{fig:pcr_prop99_hz_jack}
	\end{subfigure} 
	\begin{subfigure}[b]{0.32\textwidth}
		\centering 
		\includegraphics[width=\linewidth]
		{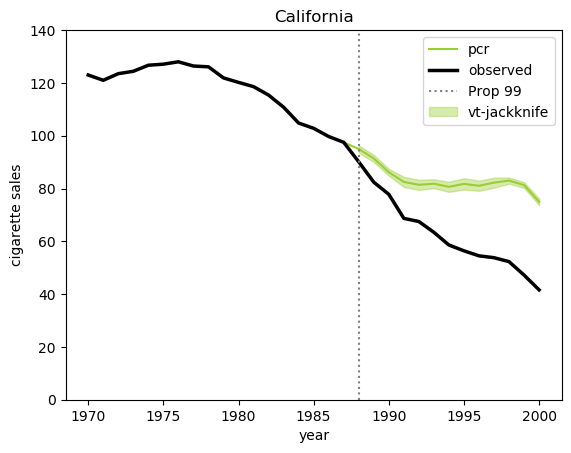}
		\label{fig:pcr_prop99_vt_jack}
	\end{subfigure} 
	\begin{subfigure}[b]{0.32\textwidth}
		\centering 
		\includegraphics[width=\linewidth]
		{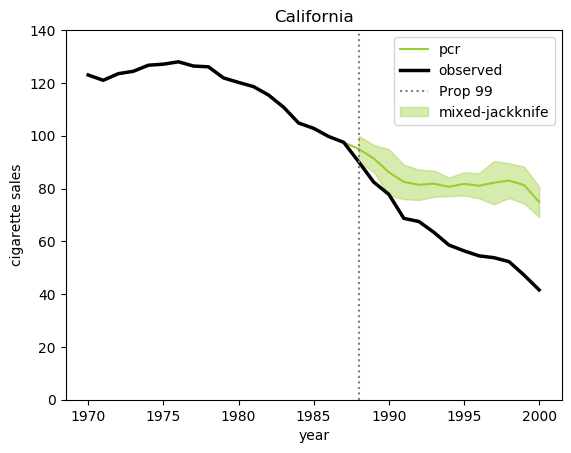}
		\label{fig:pcr_prop99_mixed_jack}
	\end{subfigure} 	
	\\
	\begin{subfigure}[b]{0.32\textwidth}
		\includegraphics[width=\linewidth]
		{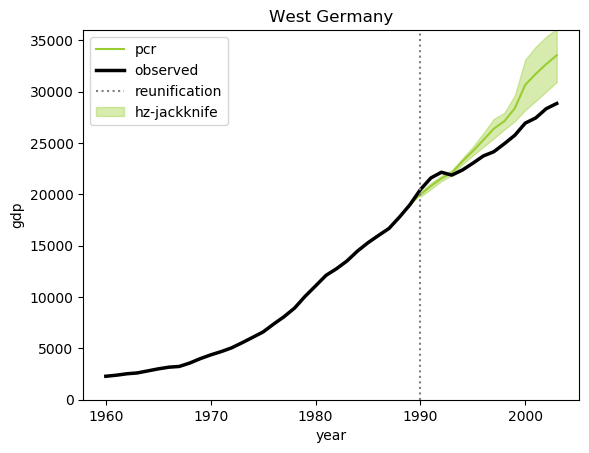}
		\caption{\HZ~model.} 
		\label{fig:pcr_germany_hz_jack}
	\end{subfigure} 
	\begin{subfigure}[b]{0.32\textwidth}
		\includegraphics[width=\linewidth]
		{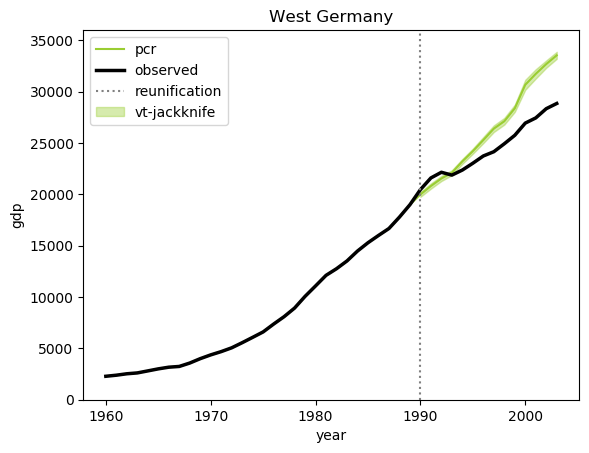}
		\caption{\VT~model.} 
		\label{fig:pcr_germany_vt_jack} 
	\end{subfigure} 
	\begin{subfigure}[b]{0.32\textwidth}
		\includegraphics[width=\linewidth]
		{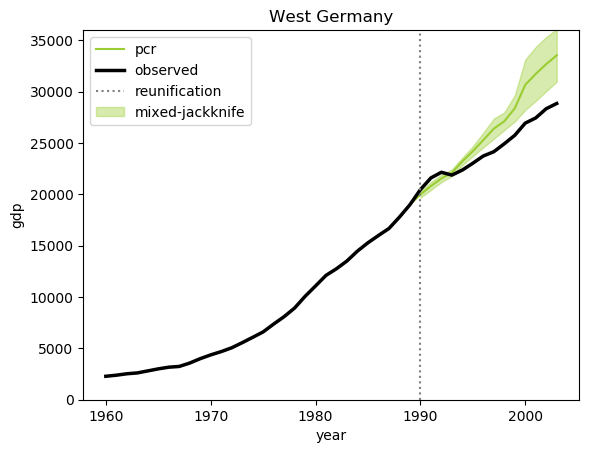}
		\caption{Mixed model.} 
		\label{fig:pcr_germany_mixed_jack}
	\end{subfigure} 	
	\caption{
	PCR estimates with jackknife confidence intervals. 
	From top to bottom, the rows are indexed by the Basque, California, and West Germany studies. 
	From left to right, the columns are indexed by the \HZ, \VT, and mixed models.
	}
	\label{fig:pcr_jack} 
\end{figure} 

\section{Conclusion} \label{sec:conclusion} 
This article rewrites the conventional wisdom on \HZ~and \VT~regressions for panel data analysis. 
Contrary to standard notions, we show that the two regressions yield identical point estimates under several standard settings. 
At the same time, we articulate that the source of randomness directly affects the accuracy of the inference that can be conducted. 
From a practical standpoint, this stresses that researchers should carefully consider where the randomness in their data stems from as this decision will then guide their choice of estimand and inferential procedure. 

%
%




\newpage
\bibliographystyle{ecta-fullname}
\bibliography{bib}


\newpage 
\begin{appendix}
	\section{Inference} \label{sec:appendix_inference}
%

\subsection{Model-Based Inference: Asymptotic Properties} 
First, we state the precise form of Theorem~\ref{thm:inference}. 
Towards this, let $(\sigma^\hz_{iT})^2 = \Var(\varepsilon_{iT} | \by_N, \bY_0)$ for $i=1, \dots, N_0$. 
We define $((\sigma^\vt_{Nt})^2, (\sigma^\mix_{iT})^2, (\sigma^\mix_{Nt})^2)$ for $i=1, \dots, N_0$ and $t = 1, \dots, T_0$ with respect to $(\bSigma^\vt_N, \bSigma^\mix_T, \bSigma^\mix_N)$ analogously. 
\begin{theorem} \label{thm:inference.formal} 
(i) \emph{[\HZ~model]} Let Assumptions~\ref{assump:hz.outcome}--\ref{assump:hz.errors} hold.  
If
\begin{align} \label{eq:lyapunov} 
	\left(\sum_{i \le N_0}  \Ex\left[  |\hbeta_i \varepsilon_{iT} |^3 |~ \by_N, \bY_0 \right] \right)^2 
	&= o \left( \Big( \sum_{i \le N_0} \hbeta^2_i (\sigma^\ehz_{iT})^2 \Big)^3 \right),
\end{align} 
then 
\begin{align}
	& \frac{\hY_{NT}(0) - \mu_0^\ehz}{ \sqrt{v^\ehz_0}} \xrightarrow{d} \mathcal{N}(0,1),
\end{align} 
where $\mu_0^\ehz = \langle \by_N, \bH^v \balpha^* \rangle$ and $v_0^\ehz = \bhbeta' \bSigma^\ehz_T \bhbeta$. 
(ii) \emph{[\VT~model]} Let Assumptions~\ref{assump:vt.outcome}--\ref{assump:vt.errors} hold.  
If
\begin{align} \label{eq:lyapunov.2} 
	\left(\sum_{t \le T_0}  \Ex\left[  | \halpha_t \varepsilon_{Nt} |^3 |~ \by_T, \bY_0 \right] \right)^2 
	&= o \left( \Big( \sum_{t \le T_0} \halpha^2_t (\sigma_{Nt}^\evt)^2 \Big)^3 \right),
\end{align} 
then
\begin{align}
	\frac{\hY_{NT}(0) - \mu_0^\evt}{ \sqrt{v_0^\evt}} \xrightarrow{d} \mathcal{N}(0,1), 
\end{align} 
where $\mu^\evt_0 = \langle \by_T, \bH^u \bbeta^* \rangle$ and $v^\evt_0 = \bhalpha' \bSigma^\evt_N \bhalpha$. 
(iii) \emph{[Mixed model]} Let Assumptions~\ref{assump:mixed.outcome}--\ref{assump:mixed.errors} hold. 
If
\begin{align}
	&\left( \sum_{i \le N_0} \sum_{t \le T_0} \Ex\Big[ | (\bY_0^\dagger)_{it} \left\{ \Ex[Y_{iT} | \bY_0] \varepsilon_{Nt} 
	+ \Ex[Y_{Nt} | \bY_0] \varepsilon_{iT} + \varepsilon_{iT} \varepsilon_{Nt} \right\} |^3 ~| ~\bY_0 \Big] \right)^2 
	\\ &  = 
	o \left( \left( \sum_{i \le N_0} \sum_{t \le T_0} (\bY_0^\dagger)^2_{it} \left\{ \Ex[Y_{iT} | \bY_0]^2 (\sigma^\emix_{Nt})^2 
	+ \Ex[Y_{Nt} | \bY_0]^2 (\sigma^\emix_{iT})^2 + (\sigma^\emix_{iT})^2 (\sigma^\emix_{Nt})^2 \right\} \right)^3 \right),
	\label{eq:lyapunov.3} 
\end{align} 
then 
\begin{align}
	\frac{ \hY_{NT}(0) - \mu^\emix_0}{\sqrt{v^\emix_0}} \xrightarrow{d} \mathcal{N}(0,1),
\end{align}
where $\mu^\emix_0 = \langle \balpha^*, \bY'_0 \bbeta^* \rangle$ and 
\begin{align}
	v^\emix_0 = 
	 (\bH^u \bbeta^*)' \bSigma^\emix_T (\bH^u \bbeta^*)
	+ (\bH^v \balpha^*)' \bSigma^\emix_N (\bH^v \balpha^*)  
	+ \tr( \bY^\dagger_0 \bSigma^\emix_T (\bY'_0)^\dagger \bSigma^\emix_N). 
\end{align}    
\end{theorem} 
If $(\varepsilon_{iT}, (\sigma^\hz_{iT})^2)$ are bounded, then \eqref{eq:lyapunov} translates to $\sum_{i \le N_0} | \hbeta_i |^3 = o( \| \bhbeta \|_2^3)$, which rules out outlier coefficients; a similar interpretation can be derived for \eqref{eq:lyapunov.2}. 
Similarly, if 
$(\varepsilon_{iT}, \varepsilon_{Nt})$ and $(\sigma^2_{iT}, \sigma^2_{Nt})$ are bounded for all $(i,t)$,  
then \eqref{eq:lyapunov.3} loosely translates to 
$$
\sum_{i \le N_0} | \hbeta_i|^3 + \sum_{t \le T_0} | \halpha_t|^3 + \sum_{i \le N_0} \sum_{t \le T_0} | (\bY^\dagger_0)_{it} |^3 
= o\left( \| \bhbeta \|_2^3 + \| \bhalpha \|_2^3 + \| \bY_0^\dagger \|_F^3 \right), 
$$
which effectively bounds the magnitudes of the \HZ~and \VT~OLS coefficients and pseudoinverse matrix entries. 
We note that \eqref{eq:lyapunov}--\eqref{eq:lyapunov.3} are known as Lyapunov's condition, and we refer the interested reader to \cite{lehmann} for details. 

Next, we provide a bound on the trace term in $v^\mix_0$. 
Beginning with the upper bound, notice that 
\begin{align}
	\tr( \bY^\dagger_0 \bSigma^\mix_T (\bY'_0)^\dagger \bSigma^\mix_N) 
	\le \max_{i \le N_0} (\sigma^\mix_{iT})^2 \max_{t \le T_0} (\sigma^\mix_{Nt})^2
	\tr( \bY^\dagger_0 (\bY'_0)^\dagger). 
\end{align}
%
%
By the cyclic property of the trace operator, 
\begin{align}
	\tr( \bY^\dagger_0 (\bY'_0)^\dagger) 
	= \tr( \bV \bS^{-2} \bV') 
	= \tr(\bS^{-2} \bV' \bV) 
	= \tr(\bS^{-2}). 
\end{align} 
Putting everything together, we obtain 
\begin{align}
	\tr( \bY^\dagger_0 \bSigma^\mix_T (\bY'_0)^\dagger \bSigma^\mix_N) 
	\le \max_{i \le N_0} (\sigma^\mix_{iT})^2 \max_{t \le T_0} (\sigma^\mix_{Nt})^2
	\tr(\bS^{-2}). 
\end{align} 
The same arguments can be applied to derive the lower bound, which yields 
%
\begin{align}
	&\tr( \bY^\dagger_0 \bSigma^\mix_T (\bY'_0)^\dagger \bSigma^\mix_N) 
	\\ &\qquad \in 
	\tr(\bS^{-2})
	\left[ \min_{i \le N_0} (\sigma^\mix_{iT})^2 \min_{t \le T_0} (\sigma^\mix_{Nt})^2, 
	~\max_{i \le N_0} (\sigma^\mix_{iT})^2 \max_{t \le T_0} (\sigma^\mix_{Nt})^2
	\right]. \label{eq:trace_interval} 
\end{align} 
From \eqref{eq:trace_interval}, we see that the lower and upper bounds match under homoskedasticity. 

\subsection{Model-Based Inference: Confidence Intervals} 
Next, we discuss technical aspects of the variance estimators in Section~\ref{sec:ci}. 

\smallskip 
\noindent 
{\bf Homoskedastic errors.} 
We take note of the recent work of \cite{si} in the synthetic controls literature. 
\cite{si} propose a \VT~PCR estimator under the homoskedastic setting and provide a similar confidence interval to that of \eqref{eq:vt.var.1} via large-sample approximations. 
Under a closely related \VT~model, they propose $\bhbeta' \bhSigma_N^\homo \bhbeta$ in place of $\bhalpha' \bhSigma_N^\homo \bhalpha$. 
While the point estimate of \cite{si} also takes the form $\langle \by_T, \bhbeta \rangle$, their variance estimator only depends on $(\by_N, \bY_0)$; in comparison, ours depends on $(\by_N, \by_T, \bY_0)$. 
Therefore, the confidence interval as per \cite{si} is numerically identical for every post-treatment point estimate while ours can vary across the post-treatment periods, which may be favorable. 

\smallskip 
\noindent 
{\bf Heteroskedastic errors.} 
Consider the heteroskedastic setting. 

\smallskip 
\noindent {\em I: Jackknife.} 
%
%
%
In the following lemma, we quantify the bias in Lemma~\ref{lemma:inference.2}. 
\begin{lemma} [Detailed restatement of Lemma~\ref{lemma:inference.2}] \label{lemma:inference.2.formal} 
(i) \emph{[\HZ~model]} Let Assumptions~\ref{assump:hz.outcome}--\ref{assump:hz.errors} hold.  
If $(\bH^u_\perp \circ  \bH^u_\perp \circ \bI)$ is nonsingular, then 
\begin{align}
	\Ex[\bhSigma^\ejack_T | \by_N, \bY_0] = \bSigma^\ehz_T + \bDelta^\ehz
	\quad \text{and} \quad 
	\Ex[ \hv^{\ehz, \ejack}_0 | \by_N, \bY_0] = v^\ehz_0 +  \bhalpha' \bDelta^\ehz \bhalpha, 
\end{align} 
where $\Delta^\ehz_{\ell \ell} = \sum_{j \neq \ell} (\sigma^\ehz_{jT})^2 (H^u_{\ell j})^2 (1 - H^u_{\ell \ell})^{-2}$ for $\ell=1, \dots, N_0$. 
%
%
(ii) \emph{[\VT~model]} Let Assumptions~\ref{assump:vt.outcome}--\ref{assump:vt.errors} hold. 
If $(\bH^v_\perp \circ  \bH^v_\perp \circ \bI)$ is nonsingular, then
\begin{align} 
	\Ex[ \bhSigma_N^\ejack | \by_T, \bY_0] = \bSigma^\evt_N + \bGamma^\evt
	\quad \text{and} \quad 
	\Ex[ \hv^{\evt, \ejack}_0 | \by_T, \bY_0] = v^\evt_0 + \bhbeta' \bGamma^\evt \bhbeta, 
\end{align}  
where $\Gamma^\evt_{\ell \ell} = \sum_{j \neq \ell} (\sigma^\evt_{Nj})^2 (H^v_{\ell j})^2(1-H^v_{\ell \ell})^{-2}$ for $\ell=1, \dots, T_0$. 
(iii) \emph{[Mixed model]} Let Assumptions~\ref{assump:mixed.outcome}--\ref{assump:mixed.errors} hold. 
If $(\bH^u_\perp \circ  \bH^u_\perp \circ \bI)$ and $(\bH^v_\perp \circ  \bH^v_\perp \circ \bI)$ are nonsingular, then
\begin{align}
	&\Ex[\bhSigma_T^\ejack | \bY_0] = \bSigma^\emix_T + \bDelta^\emix,
	\quad 
	\Ex[ \bhSigma_N^\ejack | \bY_0] = \bSigma^\emix_N + \bGamma^\emix,
	\\ 
	&\Ex[\hv^{\emix, \ejack}(\bY_0) | \bY_0] 
	\\&\qquad = v^\emix_0 
	+ (\bH^u \bbeta^*)' \bDelta^\emix (\bH^u \bbeta^*) 
	+ (\bH^v \balpha^*)' \bGamma^\emix (\bH^v \balpha^*)
	+ \tr(\bY_0^\dagger \bDelta^\emix (\bY'_0)^\dagger \bGamma^\emix),  
\end{align}
where $\Delta^\emix_{\ell \ell}$ and $\Gamma^\emix_{\ell \ell}$ are defined analogously to $\Delta^\ehz_{\ell \ell}$ and $\Gamma^\evt_{\ell \ell}$, respectively, with $(\sigma^\emix_{jT})^2$ and $(\sigma^\emix_{Nj})^2$ in place of $(\sigma^\ehz_{jT})^2$ and $(\sigma^\evt_{Nj})^2$, respectively. 
\end{lemma} 
Without loss of generality, we consider the magnitudes of $(\Delta^\mix_{\ell \ell}, \Gamma^\mix_{\ell \ell})$. 
Towards bounding the former quantity, notice that $\bH^u$ is an orthogonal projector and is thus idempotent, i.e., $(\bH^u)^2 = \bH^u$, and symmetric. 
Therefore,  
\begin{align}
	H^u_{\ell \ell} = (H^u_{\ell \ell})^2 + \sum_{j \neq \ell} (H^u_{\ell j})^2
	\implies \sum_{j \neq \ell} (H^u_{\ell j})^2 = H^u_{\ell \ell} ( 1 - H^u_{\ell \ell}). \label{eq:diag.2} 
\end{align} 
This yields  
$\Delta^\mix_{\ell \ell} \in H^u_{\ell \ell}  (1 - H^u_{\ell \ell})^{-1} [ \min_{j \neq \ell} (\sigma^\mix_{jT})^2, \max_{j \neq \ell} (\sigma^\mix_{jT})^2]$.
Since $\sum_{j \neq \ell} (H_{\ell j}^u)^2 \ge 0$, \eqref{eq:diag.2} implies $H^u_{\ell \ell} \in [0,1]$. 
Thus, $\Delta^\mix_{\ell \ell} = 0$ if $H^u_{\ell \ell} = 0$ and diverges if $H^u_{\ell \ell} = 1$. 
If 
$H^u_{\ell \ell}$ takes the average value $N_0^{-1} \sum_{\ell} H^u_{\ell \ell} 
= R N_0^{-1}$, then it follows that $\Delta^\mix_{\ell \ell} \in R (N_0-R)^{-1} [ \min_{j \neq \ell} (\sigma^\mix_{jT})^2, \max_{j \neq \ell} (\sigma^\mix_{jT})^2]$.  
A similar result is derived for $\Gamma^\mix_{\ell \ell}$. 

\smallskip 
\noindent 
{\em II: HRK-estimator.} 
Recall Lemma~\ref{lemma:inference.3}.
Consider the \HZ~estimator and the invertibility of $(\bH^u \circ \bH^u)$. 
A sufficient condition is strict diagonal dominance \citep{varga1969matrix}: 
$(1 - H^u_{\ell \ell})^2 > \sum_{j \neq \ell} (H^u_{\ell j})^2$. 
Using \eqref{eq:diag.2}, we simplify this condition as
$(1 - H^u_{\ell \ell})^2 > H^u_{\ell \ell} - (H^u_{\ell \ell})^2$.
Thus, $\max_\ell H^u_{\ell \ell} < 1/2$ is a sufficient condition for invertibility.  
%
%
The same arguments apply for \VT~regression. 

\smallskip 
\noindent 
{\bf Mixed variance estimator.} 
Let us bound $\hv^\mix_0$. 
From \eqref{eq:trace_interval}, it follows that $\hv^\mix_0 \in [\hv^\mix_{0, \min}, \hv^\mix_{0, \max}]$, where 
\begin{align}
	& \hv^\mix_{0, \min} =   
	\min_{i \le N_0} \hsigma^2_{iT} \| \bhbeta \|_2^2 
	~+ \min_{t \le T_0} \hsigma^2_{Nt} \| \bhalpha \|_2^2 
	~- 
	\max_{i \le N_0} \hsigma^2_{iT} \max_{t \le T_0} \hsigma^2_{Nt} \tr(\bS^{-2}), 
	\\ 
	& \hv^\mix_{0, \max} =  
	\max_{i \le N_0} \hsigma^2_{iT} \| \bhbeta \|_2^2 
	~+ \max_{t \le T_0} \hsigma^2_{Nt} \| \bhalpha \|_2^2
	~- 
	\min_{i \le N_0} \hsigma^2_{iT} \min_{t \le T_0} \hsigma^2_{Nt} \tr(\bS^{-2}). 
\end{align} 
Here, $\hsigma^2_{iT}$ and $\hsigma^2_{Nt}$ are the $i$th and $t$th diagonal elements of $\bhSigma_T$ and $\bhSigma_N$, respectively.  
Observe that $\hv^\mix_{0, \min} = \hv^\mix_{0, \max}$ for homoskedastic errors. 



\section{Illustrations} \label{sec:sim_emp}
This section provides details on the simulations that were absent in the main article. 

\subsection{Implementation Details} 
%
%
For ridge, lasso, and elastic net regressions, we use the default \texttt{scikit-learn} hyperparameters ($\lambda_1, \lambda_2$). 
For PCR, we choose the number of principal components $k$ via the approach described in Section~\ref{sec:simulations.1}. 
This yields $k=$ for the Basque study, $k=3$ for the California study, and $k = 4$ for the West Germany study. 
We implement simplex regression using the code made available at \href{https://matheusfacure.github.io/python-causality-handbook/15-Synthetic-Control.html}{https://matheusfacure.github.io/python-causality-handbook/15-Synthetic-Control.html}. 
%

\subsection{Data-Inspired Simulation Studies} 
Formally, we define average length (AL) as 
\begin{align}
	\text{AL} = \frac{1}{m} \sum_{\ell=1}^m \frac{(2 \cdot 1.96) \sqrt{\hv^{(\ell)}_0}}{|\hY^{(\ell)}_{NT}(0) |},
\end{align} 
where $\hv^{(\ell)}_0$ and $\hY^{(\ell)}_{NT}(0)$ are the variance and point estimates for the $\ell$th repeat. 

	 \section{Proofs for Point Estimation} 
{\bf Helper lemmas.} To establish Theorems~\ref{thm:equivalent} and \ref{thm:not_equivalent}, we first state the following useful lemmas for the collection of regression formulations presented in Section~\ref{sec:estimation}. 
 We provide their proofs in Appendix~\ref{sec:proofs_estimation_lemmas}. 
 
\begin{lemma} [OLS] \label{lemma:ols} 
$\EHZ = \EVT$ for OLS with $(\bhalpha, \bhbeta)$ as the minimum $\ell_2$-norm solutions:  
\begin{align} \label{eq:ols} 
    \hY_{NT}^{\ehz}(0) &= \hY_{NT}^{\evt}(0) 
    = \langle \by_N, \bY_0^\dagger \by_T \rangle 
    = \sum_{\ell=1}^R (1/s_\ell) \langle \by_N, \bv_\ell \rangle \langle \bu_\ell, \by_T \rangle.
\end{align}
\end{lemma} 

\begin{lemma} [PCR] \label{lemma:pcr} 
$\EHZ = \EVT$ for PCR with the same choice of $k < R$:  
\begin{align} \label{eq:pcr} 
    \hY_{NT}^{\ehz}(0) &= \hY_{NT}^{\evt}(0) 
    = \langle \by_N, (\bY_0^{(k)})^\dagger \by_T \rangle 
    = \sum_{\ell=1}^k (1/s_\ell) \langle \by_N, \bv_\ell \rangle \langle \bu_\ell, \by_T \rangle.
\end{align}
\end{lemma} 

\begin{lemma} [Ridge] \label{lemma:ridge} 
$\EHZ = \EVT$ for ridge regression with the same choice of $\lambda_2 > 0$:  
\begin{align} \label{eq:ridge} 
    \hY_{NT}^{\ehz}(0) &= \hY_{NT}^{\evt}(0) 
    = \langle \by_N,  (\bY_0' \bY_0 + \lambda_2 \bI)^{-1} \bY_0' \by_T \rangle 
    = \sum_{\ell=1}^R \frac{s_\ell}{s_\ell^2 + \lambda_2} \langle \by_N, \bv_\ell \rangle \langle \bu_\ell, \by_T \rangle.
\end{align}
\end{lemma} 

\begin{lemma} [Lasso] \label{lemma:lasso} 
$\EHZ \neq \EVT$ for lasso regression. 
\end{lemma} 

\begin{lemma} [Elastic net] \label{lemma:enet} 
$\EHZ \neq \EVT$ for elastic net regression. 
\end{lemma} 

\begin{lemma} [Simplex regression] \label{lemma:cvx}
$\EHZ \neq \EVT$ for simplex regression. 
\end{lemma}

\subsection{Proof of Theorem~\ref{thm:equivalent}} 

\begin{proof} 
The proof is immediate from Lemmas~\ref{lemma:ols}--\ref{lemma:ridge}. 
\end{proof} 

\subsection{Proof of Theorem~\ref{thm:not_equivalent}} 

\begin{proof} 
The proof is immediate from Lemmas~\ref{lemma:lasso}--\ref{lemma:cvx}. 
\end{proof} 

\subsection{Proof of Corollary~\ref{cor:sdid}} 

\begin{proof}
We begin with the OLS. 
Recall that $\hY^\hz_{NT}(0)  = \langle \by_N, \bhalpha \rangle$ with $\bhalpha = \bY_0^\dagger \by_T$ and $\hY^\vt_{NT}(0) = \langle \by_T, \bhbeta \rangle$ with $\bhbeta = (\bY_0')^\dagger \by_N$.  
By Theorem~\ref{thm:equivalent}, we have 
\begin{align}
	\hY^\hz_{NT}(0)  = \hY^\vt_{NT}(0) = \langle \by_N, \bY_0^\dagger \by_T \rangle. 
\end{align} 
Returning to \eqref{eq:sdid}, we obtain
\begin{align}
	\hY^\sdid_{NT}(0) &= \hY^\asc_{NT}(0) =  \langle \by_T, \bhbeta \rangle + \langle \by_N, \bhalpha \rangle - \langle \bhalpha, \bY_0' \bhbeta \rangle  
	= 2 \langle \by_T, \bhbeta \rangle - \langle  \bhalpha, \bY_0' \bhbeta \rangle. \label{eq:sdid.0} 
\end{align} 
Recall $(\bY_0')^\dagger  = (\bY_0^\dagger)'$. 
Therefore,  
\begin{align}
 	\langle  \bhalpha, \bY_0' \bhbeta \rangle &= \by'_T(\bY_0')^\dagger \bY_0' (\bY_0')^\dagger \by_N 
	= \by'_T(\bY_0')^\dagger \by_N 
	= \by'_T \bhbeta. \label{eq:sdid.1}
\end{align} 
Plugging \eqref{eq:sdid.1} into \eqref{eq:sdid.0}, we conclude 
\begin{align}
	\hY^\sdid_{NT}(0) &= \hY^\asc_{NT}(0) = 2 \langle \by_T, \bhbeta \rangle - \langle \by_T, \bhbeta \rangle
	= \hY^\vt_{NT}(0) = \hY^\hz_{NT}(0). 
\end{align}  
Now, observe that the same arguments above hold when $\bY_0^{(k)}$ takes the place of $\bY_0$ for any $k < R$. 
Therefore, the same reduction can be derived for PCR. 
\end{proof} 

\subsection{Proof of Corollary~\ref{cor:intercepts}} 

\begin{proof}
Let $\bY^\hz_0 = [\bone, \bY_0]$ and $\bY^\vt_0 = [\bone, \bY'_0]$. 
The proof is immediate from Theorem~\ref{thm:equivalent} by noting that $(\bY^\hz_0)' \neq \bY^\vt_0$. 
\end{proof}

\subsection{Proof of Corollary~\ref{cor:equivalent.int}} 

\begin{proof}
Consider \HZ~ridge regression. 
To begin, the optimality conditions give 
\begin{align}
	\nabla_{(\alpha_0, \alpha_1, \balpha)} \left\{ \| \by_T - \bY_0 \balpha - \alpha_0 \bone\|_2^2 + \| \by_N - \balpha_1 \bone \|_2^2 + \lambda \| \balpha \|_2^2 \right\} &= 0. 
\end{align}
Solving for $\alpha_0$, we have $\halpha_0 = (1/N_0) \langle \by_T, \bone \rangle$, 
where we have used the fact that $\bY'_0 \bone = \bzero$. 
Solving for $\alpha_1$, we have $\halpha_1 = (1/T_0) \langle \by_N, \bone \rangle$. 
Finally, solving for $\balpha$, we have $\bhalpha = (\bY'_0 \bY_0 + \lambda \bI)^{-1} \bY'_0 \by_T$. 

Switching gears to \VT~ridge regression, similar arguments yield 
$\hbeta_0 = (1/T_0) \langle \by_N, \bone \rangle$,  
$\hbeta_1 = (1/N_0) \langle \by_T, \bone \rangle$,
and 
$\bhbeta = (\bY_0 \bY'_0 + \lambda \bI) \bY_0 \by_N$. 

We establish our desired result by invoking Theorem~\ref{thm:equivalent} to obtain  
$\langle \by_N, \bhalpha \rangle = \langle \by_T, \bhbeta \rangle$. 
Next, we observe that $\halpha_0 = \hbeta_1$ and $\hbeta_0 = \halpha_1$. 
This proves that \HZ~and \VT~ridge yield numerically identical point estimates. 
Setting $\lambda = 0$ and using the pseudoinverse, we have our result for OLS. 
The result for PCR is then established from the OLS result by substituting $\bY_0^{(k)}$ for $\bY_0$. 
This completes the proof. 
\end{proof}

\subsection{Proofs of Point Estimation Lemmas} \label{sec:proofs_estimation_lemmas} 

\subsubsection{Proof of Lemma~\ref{lemma:ols}: OLS} 
\begin{proof} 
Consider \HZ~regression. 
By the optimality conditions, 
\begin{align}
	\nabla_{\balpha} \| \by_T - \bY_0 \balpha \|_2^2 &= 0. 
\end{align}
Solving for $\balpha$, we derive the well-known ``normal equations'' 
\begin{align}
	\bY_0' \bY_0 \balpha = \bY_0' \by_T. 
\end{align} 
Using the pseudoinverse, we obtain 
\begin{align}
	\bhalpha &= (\bY'_0 \bY_0)^\dagger \bY'_0 \by_T 
	= \bY_0^\dagger \by_T  \label{eq:hz_ols} 
\end{align} 
Observe that this corresponds to the unique minimum $\ell_2$-norm solution that lies within the rowspace of $\bY_0$. 
Therefore, the \HZ~prediction is given by 
\begin{align}
	\hY^\hz_{NT}(0) &= \langle \by_N, \bhalpha \rangle 
	= \langle \by_N, \bY^\dagger_0 \by_T \rangle. \label{eq:hz_ols.1} 
\end{align}

Following the arguments above for \VT~regression, it follows that 
\begin{align}
	\bhbeta &= (\bY_0 \bY_0')^\dagger \bY_0 \by_N 
	= (\bY_0')^\dagger \by_N,
\end{align} 
which corresponds to the unique minimum $\ell_2$-norm solution that lies within the columnspace of $\bY_0$. 
Therefore, 
\begin{align}
	\hY^\vt_{NT}(0) &= \langle \by_T, \bhbeta \rangle 
	= \langle \by_T, (\bY'_0)^\dagger \by_N \rangle. \label{eq:vt_ols.1} 
\end{align} 
Given that $(\bY_0')^\dagger = (\bY_0^\dagger)'$, we conclude  
\begin{align}
	\hY^\hz_{NT}(0) &= \langle \by_N, \bY_0^\dagger \by_T \rangle 
	= \langle \by_T, (\bY_0')^\dagger \by_N \rangle 
	= \hY^\vt_{NT}(0). \label{eq:ols_equiv} 
\end{align} 
\end{proof}

\subsubsection{Proof of Lemma~\ref{lemma:pcr}: PCR} 

\begin{proof} 
Consider \HZ~regression with any $k < R$. 
Let $\bU_k \in \Rb^{N_0 \times k}$ and $\bV_k \in \Rb^{T_0 \times k}$ denote the matrices formed by the top $k$ left and right singular vectors, respectively, and $\bS_k \in \Rb^{k \times k}$ denote the matrix of top $k$ singular values. 
Observe that 
\begin{align}
	\left((\bY_0^{(k)})' \bY_0^{(k)} \right)^\dagger (\bY_0^{(k)})' 
	&= (\bV_k \bS_k^{-2} \bV'_k) \bV_k \bS_k \bU'_k 
	= \bV_k \bS_k^{-1} \bU'_k
	= (\bY_0^{(k)})^\dagger. 
\end{align} 
Therefore, $\bhalpha = (\bY_0^{(k)})^\dagger \by_T$, which corresponds to the unique minimum $\ell_2$-norm solution that lies within the rowspace of $\bY^{(k)}_0$. 
Following the proof of Lemma~\ref{lemma:ols}, we conclude that
\begin{align} 
	\hY^\hz_{NT}(0) &= \langle \by_N, \bhalpha \rangle = \langle \by_N, (\bY_0^{(k)})^\dagger \by_T \rangle.
\end{align} 
Similarly, for \VT~regression, we note that 
\begin{align}
	\left(\bY_0^{(k)} (\bY_0^{(k)})' \right)^\dagger \bY_0^{(k)}
	&= (\bU_k \bS_k^{-2} \bU'_k) \bU_k \bS_k \bV'_k 
	= \bU_k \bS_k^{-1} \bV'_k
	= ((\bY_0^{(k)})')^\dagger. 
\end{align} 
In turn, we have $\bhbeta = ((\bY_0^{(k)})')^\dagger \by_N$, which corresponds to the unique minimum $\ell_2$-norm solution that lies within the columnspace of $\bY^{(k)}_0$. 
Moreover, 
\begin{align} 
	\hY^\vt_{NT}(0) &= \langle \by_T, \bhbeta \rangle = \langle \by_T, ((\bY_0^{(k)})')^\dagger \by_N \rangle. 
\end{align} 
We finish by establishing  
\begin{align}
	\hY^\hz_{NT}(0) = \langle \by_N, (\bY_0^{(k)})^\dagger \by_T \rangle 
	= \langle \by_T, ((\bY_0^{(k)})')^\dagger \by_N \rangle = \hY^\vt_{NT}(0). 
\end{align} 
\end{proof}

\smallskip 
\noindent 
{\bf Helper lemma.} To establish Lemmas~\ref{lemma:ridge}--\ref{lemma:enet}, we first establish a general result in Lemma~\ref{lemma:reg_ls} for $\ell_p$-penalties, where $p = 2/K$ and $K$ is an integer $\ge 1$, based on the contributions of \cite{hoff}. 
More formally, consider
\begin{itemize}
	\item [(a)] \HZ~regression: for $K \ge 1$ and $\lambda > 0$, 
	\begin{align} 
		&\bhalpha = \argmin_{\balpha} \| \by_T- \bY_0 \balpha \|_2^2 ~+~ \lambda \| \balpha \|_p^p \label{eq:hz.lp.1} 
		\\ &\hY_{NT}^\hz(0) =\langle \by_N, \bhalpha \rangle. \label{eq:hz.lp.2} 
	\end{align} 
	
	\item [(b)] \VT~regression: for $K \ge 1$ and $\lambda > 0$,
	\begin{align} 
		&\bhbeta = \argmin_{\bbeta} \| \by_N - \bY_0' \bbeta \|_2^2 ~+~ \lambda \| \bbeta \|_p^p \label{eq:vt.lp.1} 
		\\ &\hY_{NT}^\vt(0) =\langle \by_T, \bhbeta \rangle. \label{eq:vt.lp.2} 
	\end{align} 
\end{itemize}
We remark that $K=1$ and $K=2$ yield ridge and lasso regression, respectively, while $K > 2$ yields non-convex penalties. 
We relegate the proof of Lemma~\ref{lemma:reg_ls} to Appendix~\ref{sec:proof_reg_ls}. 

\begin{lemma} \label{lemma:reg_ls} 
For any $K \ge 1$ and $\lambda > 0$, a \EHZ~and \EVT~regression solution is
\begin{align}
	\hY^\ehz_{NT}(0) &=\langle \by_N, \bhalpha_1 \circ \dots \circ \bhalpha_K \rangle
	\\ 
	\hY^\evt_{NT}(0) &=\langle \by_T, \bhbeta_1 \circ \dots \circ \bhbeta_K \rangle,
\end{align}
where for every $k \le K$, 
\begin{align}
	\bhalpha_k &= \left( \bD(\bhalpha_{\sim k}) \bY_0' \bY_0 \bD(\bhalpha_{\sim k}) + \frac{\lambda}{K} \bI \right)^{-1} \bD(\bhalpha_{\sim k}) \bY_0' \by_T, \label{eq:lp.alpha} 
	\\ 
	\bhbeta_k &= \left( \bD(\bhbeta_{\sim k}) \bY_0 \bY_0' \bD(\bhbeta_{\sim k})+ \frac{\lambda}{K} \bI \right)^{-1} \bD(\bhbeta_{\sim k}) \bY_0 \by_N, \label{eq:lp.beta} 
\end{align} 
$\bhalpha_{\sim k} = \bhalpha_1 \circ \dots \circ \bhalpha_{k-1} \circ \bhalpha_{k+1} \circ \dots \circ \bhalpha_K$,
$\bhbeta_{\sim k} = \bhbeta_1 \circ \dots \circ \bhbeta_{k-1} \circ \bhbeta_{k+1} \circ \dots \circ \bhbeta_K$,
and $\bD(\bhalpha_{\sim k})$ and $\bD(\bhbeta{\sim k})$ are diagonal matrices formed from $\bhalpha_{\sim k}$ and $\bhbeta_{\sim k}$, respectively.  
\end{lemma} 

\subsubsection{Proof of Lemma~\ref{lemma:ridge}: Ridge Regression} 

\begin{proof}
By Lemma~\ref{lemma:reg_ls} for $K=1$ and $\lambda = \lambda_2 > 0$, the \HZ~regression solution is 
\begin{align}
	\hY^\hz_{NT}(0) &=\langle \by_N, (\bY_0' \bY_0 + \lambda_2 \bI)^{-1} \bY_0' \by_T\rangle. 
\end{align} 
Similarly, the \VT~regression solution is given by 
\begin{align}
	\hY^\vt_{NT}(0) &=\langle \by_T, (\bY_0 \bY_0' + \lambda_2 \bI)^{-1} \bY_0 \by_N \rangle. 
\end{align} 
Since $(\bY_0' \bY_0 + \lambda \bI)^{-1} \bY_0' = \bY_0' (\bY_0 \bY_0' + \lambda \bI)^{-1}$,
it follows that $\hY^\hz_{NT}(0) = \hY^\vt_{NT}(0)$. 
\end{proof}

\subsubsection{Proof of Lemma~\ref{lemma:lasso}: Lasso Regression} 

\begin{proof}
By Lemma~\ref{lemma:reg_ls} for $K=2$ and $\lambda = \lambda_1 > 0$, a \HZ~regression solution is
\begin{align} \label{eq:hz.lasso}
	\hY^\hz_{NT}(0) &=\langle \by_N, \bhalpha_1 \circ \bhalpha_2 \rangle,
\end{align} 
where 
\begin{align}
	\bhalpha_{1+k} &= \left( \bD(\bhalpha_{2-k}) \bY_0' \bY_0 \bD(\bhalpha_{2-k}) + \frac{\lambda_1}{2} \bI \right)^{-1} \bD(\bhalpha_{2-k}) \bY_0' \by_T
\end{align} 
for $k \in \{0,1\}$. 
Similarly, a \VT~regression solution is given by 
\begin{align} \label{eq:vt.lasso}
	\hY^\vt_{NT}(0) &=\langle \by_T, \bhbeta_1 \circ \bhbeta_2 \rangle,
\end{align} 
where 
\begin{align}
	\bhbeta_{1+k} &= \left( \bD(\bhbeta_{2-k}) \bY_0 \bY_0' \bD(\bhbeta_{2-k}) + \frac{\lambda_1}{2} \bI \right)^{-1} \bD(\bhbeta_{2-k}) \bY_0 \by_N 
\end{align} 
for $k \in \{0,1\}$.  
Leveraging \eqref{eq:hz.lasso} and \eqref{eq:vt.lasso}, we find that the \HZ~regression solution can be linear in $y$ and at least quadratic in $q$.
On the other hand, the \VT~regression solution can be linear in $q$ and at least quadratic in $y$.
Since the lasso solution is unique under the assumption the entries of $\bY_0$ are drawn from a continuous distribution, this implies that \HZ~and \VT~regressions do not yield matching solutions in general. 
\end{proof}

\subsubsection{Proof of Lemma~\ref{lemma:enet}: Elastic Net Regression} 

\begin{proof}
Consider \HZ~regression. 
We rewrite \eqref{eq:hz.rr.1} in a lasso formulation:  
\begin{align} 
	\bhalpha^* &= \argmin_{\balpha^*} \| \by_T^* - \bY_0^* \balpha^* \|_2^2 ~+~ \lambda^* \| \balpha^* \|_1, \label{eq:hz.enet.lasso} 
\end{align} 
where 
\begin{align}
	&q^* = \begin{pmatrix}
		\by_T \\ \bzero
	\end{pmatrix},
	\quad 
	\bY_0^* = \frac{1}{\sqrt{1 + \lambda_2}}  
	\begin{pmatrix}
		\bY_0 \\ \sqrt{\lambda_2} \bI
	\end{pmatrix},
	\quad
	\lambda^* = \frac{\lambda_1}{\sqrt{1 + \lambda_2}},
	\quad
	\balpha^* = (\sqrt{1 + \lambda_2}) \balpha. 
\end{align} 
%
%
We apply Lemma~\ref{lemma:reg_ls} to \eqref{eq:hz.enet.lasso} with $K=2$ and $\lambda = \lambda^* > 0$ to obtain 
\begin{align} \label{eq:hz.enet} 
	\hY^\hz_{NT}(0) &= \frac{\langle \by_N, \bhalpha^*_1 \circ \bhalpha^*_2 \rangle}{ \sqrt{1 + \lambda_2}},
\end{align}
where
\begin{align}
	\bhalpha^*_{1+k} &= \left( \frac{1}{\sqrt{1 + \lambda_2}} \bD(\bhalpha^*_{2-k}) (\bY_0' \bY_0 + \lambda_2 \bI) \bD(\bhalpha^*_{2-k}) + \frac{\lambda_1}{2} \bI \right)^{-1} \bD(\bhalpha^*_{2-k}) \bY_0' \by_T
\end{align} 
for $k \in \{0,1\}$. 
%
Similarly, for \VT~regression, we proceed as above to obtain 
\begin{align} \label{eq:vt.enet} 
	\hY^\vt_{NT}(0) &= \frac{\langle \by_T, \bhbeta^*_1 \circ \bhbeta^*_2 \rangle}{\sqrt{1 + \lambda_2}}, 
\end{align}
where
\begin{align} 
	\bhbeta^*_{1+k} &= \left( \frac{1}{\sqrt{1 + \lambda_2}} \bD(\bhbeta^*_{2-k}) (\bY_0 \bY_0' + \lambda_2 \bI) \bD(\bhbeta^*_{2-k}) + \frac{\lambda_1}{2} \bI \right)^{-1} \bD(\bhbeta^*_{2-k}) \bY_0 \by_N
\end{align} 
for $k \in \{0,1\}$.  
Leveraging \eqref{eq:hz.enet} and \eqref{eq:vt.enet}, we find that the \HZ~regression solution can be linear in $y$ and at least quadratic in $q$.
On the other hand, the \VT~regression solution can be linear in $q$ and at least quadratic in $y$.
Since the elastic net regression solution is unique, provided $\lambda_2 > 0$, this implies that \HZ~and \VT~regressions do not yield matching solutions in general. 
\end{proof}


\subsubsection{Proof of Lemma~\ref{lemma:cvx}: Simplex Regression} 

\begin{proof}
Consider \HZ~regression. 
We write the Lagrangian of \eqref{eq:hz.cvx.1} as 
\begin{align}
	\bhalpha &= \argmin_{\balpha} ~\| \by_T - \bY_0 \balpha \|_2^2 ~+~ \lambda \| \balpha \|_2^2 ~-~ (\btheta^\hz)' \balpha + \nu^\hz (\bone' \balpha - 1),
\end{align} 
where $\btheta^\hz \in \Rb^{T_0}$ and $\nu^\hz \in \Rb$. 
By the Karush-Kuhn-Tucker (KKT) conditions, optimality is achieved if the following are satisfied:  
\begin{align} 
	&\bhalpha \succeq \bzero, 
	\quad
	 \bone' \bhalpha = 1, 
	\\  
	& \bhtheta^\hz \succeq \bzero, 
	\\
	& \htheta^\hz_i \halpha_i = 0 \quad \text{for } i=1, \dots, T_0, 
	\\ 
	& \bhalpha = (\bY_0' \bY_0 + \lambda \bI)^{-1} \left( \bY_0' \by_T + \frac{1}{2} \bhtheta^\hz - \frac{\hnu^\hz}{2} \bone \right).  
\end{align} 
Therefore, given primal and dual feasible variables $(\bhalpha, \bhtheta^\hz, \hnu^\hz)$, we can write the final \HZ~prediction as  
\begin{align}
	\hY^\hz_{NT}(0) &= \hY^{\hz, \ols}_{NT}(0) 
	+ (1/2) \by'_N (\bY_0' \bY_0 + \lambda \bI)^{-1} (\bhtheta^\hz - \hnu^\hz \bone ),  
\end{align}
where $\hY^{\hz, \ols}_{NT}(0) = \by'_N (\bY_0' \bY_0 + \lambda \bI)^{-1} \bY_0' \by_T$ converges to the prediction corresponding to the OLS solution with minimum $\ell_2$-norm as $\lambda \rightarrow 0^+$. 
Similarly, for \VT~regression, the KKT conditions are
\begin{align} 
	&\bhbeta \succeq \bzero, 
	\quad
	 \bone' \bhbeta = 1, 
	\\  
	& \bhtheta^\vt \succeq \bzero, 
	\\
	& \htheta^\vt_i \hbeta_i = 0 \quad \text{for } i=1, \dots, N_0, 
	\\ 
	& \bhbeta = (\bY_0 \bY_0' + \lambda \bI)^{-1} \left(\bY_0 \by_N + \frac{1}{2} \bhtheta^\vt - \frac{\hnu^\vt}{2} \bone\right).  
\end{align} 
For primal and dual feasible variables $(\bhbeta, \bhtheta^\vt, \hnu^\vt)$, this yields 
\begin{align}
	\hY^\vt_{NT}(0) &= \hY^{\vt, \ols}_{NT}(0) 
	+ (1/2) \by'_T (\bY_0 \bY_0' + \lambda \bI)^{-1} (\bhtheta^\vt - \hnu^\vt \bone ),  
\end{align} 
where $\hY^{\vt, \ols}_{NT}(0) = \by'_T (\bY_0 \bY_0' + \lambda \bI)^{-1} \bY_0 \by_N$ converges to the prediction corresponding to the OLS solution with minimum $\ell_2$-norm as $\lambda \rightarrow 0^+$. 
Notably, as per Theorem~\ref{thm:equivalent}, $\hY^{\hz, \ols}_{NT}(0) = \hY^{\vt, \ols}_{NT}(0) = \hY^{\ols}_{NT}(0)$ for any $\lambda \ge 0$.
As a result,
\begin{align}
	\hY^\hz_{NT}(0) &= \hY^{\ols}_{NT}(0) 
	+ (1/2) \by'_N (\bY_0' \bY_0 + \lambda \bI)^{-1} (\bhtheta^\hz - \hnu^\hz \bone ) \label{eq:hz.cvx.ols}
	\\ 
	\hY^\vt_{NT}(0) &= \hY^{\ols}_{NT}(0) 
	+ (1/2) \by'_T (\bY_0 \bY_0' + \lambda \bI)^{-1} (\bhtheta^\vt - \hnu^\vt \bone ). \label{eq:vt.cvx.ols} 
\end{align} 
As seen from \eqref{eq:hz.cvx.ols} and \eqref{eq:vt.cvx.ols}, the leading terms in the \HZ~and \VT~simplex regression predictions are identical.
%
The remaining terms, however, can differ from one another. 
As an example, consider $N = T$ with  
\begin{align} \label{eq:example.1} 
	\bY_0 = \bI, 
	\quad 
	\by_N = \bzero,
	\quad
	\by_T = (1+ \lambda) (\bhtheta^\vt - \hnu^\vt \bone). 
\end{align} 
%
By construction, observe that
\begin{align}
	\bhbeta = \frac{1}{2(1+ \lambda)} (\bhtheta^\vt - \hnu^\vt \bone).  \label{eq:example.2} 
\end{align} 
Recall from the KKT conditions for \VT~regression that $\bhbeta \succeq \bzero$ and $\bone' \bhbeta = 1$.
Therefore, at least one entry of $(\bhtheta^\vt - \hnu^\vt \bone)$ must be strictly positive. 
This yields 
\begin{align}
	(1+ \lambda)^{-1} \by'_T (\bhtheta^\vt - \hnu^\vt \bone) 
	&= (\bhtheta^\vt - \hnu^\vt \bone)' (\bhtheta^\vt - \hnu^\vt \bone) > 0. \label{eq:example.3}
\end{align}
Plugging \eqref{eq:example.1} and \eqref{eq:example.3} into \eqref{eq:hz.cvx.ols} and \eqref{eq:vt.cvx.ols}, we obtain 
\begin{align}
	\hY^\hz_{NT}(0) &= 0  
	\quad \text{and} \quad 
	\hY^\vt_{NT}(0) > 0, 
\end{align} 
which concludes our proof. 
\end{proof}

\subsubsection{Proof of Lemma~\ref{lemma:reg_ls}: $\ell_p$-penalties} \label{sec:proof_reg_ls}
\begin{proof}
%
%
We recall the Hadamard product parametrization (HPP): for any vector $\bz$ and integer $K \ge 1$, 
\begin{align}
	\| \bz \|_p^p &=  \min_{\bz_1\circ \dots \circ \bz_K = \bz} \frac{1}{K} \sum_{k=1}^K \| \bz_k \|_2^2,
\end{align} 
where $\circ$ denotes the Hadamard (componentwise) product. 
We rewrite our subclass of $\ell_p$-penalties, i.e., \eqref{eq:hz.lp.1} and \eqref{eq:vt.lp.1}, as sums of $\ell_2$-penalties via the HPP technique: 
\begin{align}
	(\bhalpha_1, \dots, \bhalpha_K) &= \argmin_{\balpha_1, \dots, \balpha_K} \| \by_T- \bY_0 (\balpha_1 \circ \dots \circ \balpha_K) \|_2^2 ~+~ \frac{\lambda}{K} \sum_{k=1}^K \| \balpha_k \|_2^2 \label{eq:hz.hoff} 
	\\ 
	(\bhbeta_1, \dots, \bhbeta_K) &= \argmin_{\bbeta_1, \dots, \bbeta_K} \| \by_N - \bY_0' (\bbeta_1 \circ \dots \circ \bbeta_K) \|_2^2 ~+~ \frac{\lambda}{K} \sum_{k=1}^K \| \bbeta_k \|_2^2, \label{eq:vt.hoff} 
\end{align} 
where $\bhalpha = \bhalpha_1 \circ \dots \circ \bhalpha_K$ and $\bhbeta = \bhbeta_1 \circ \dots \circ \bhbeta_K$. 
Below, we leverage the results of \cite{hoff}, which provides an alternating ridge regression algorithm to solve for \eqref{eq:hz.hoff}--\eqref{eq:vt.hoff}. 

Consider \HZ~regression. 
Let us solve for $\balpha_k$ for $k \in [K]$ by fixing $\balpha_{k'}$ for $k' \neq k$. 
By the optimality conditions, 
\begin{align} \label{eq:hoff_opt.1} 
	\nabla_{\balpha_k} \left\{ (\balpha_1 \circ \dots \circ \balpha_K)' \bY_0' \bY_0  (\balpha_1 \circ \dots \circ \balpha_K)  - 2 (\balpha_1 \circ \dots \circ \balpha_K)' \bY_0' \by_T+ \frac{\lambda}{K} \balpha_k' \balpha_k \right\}= 0. 
\end{align} 
In order to solve for \eqref{eq:hoff_opt.1}, observe that
\begin{align}
	& (\balpha_1 \circ \dots \circ \balpha_K)' \bY_0' \bY_0  (\balpha_1 \circ \dots \circ \balpha_K) = \balpha_k' (\bY_0' \bY_0 \circ \balpha_{\sim k} \balpha_{\sim k}') \balpha_k
	\\ 
	& (\balpha_1 \circ \dots \circ \balpha_K)' \bY_0' \by_T= \balpha_k' (\balpha_{\sim k} \circ \bY_0' \by_T),
\end{align} 
where $\balpha_{\sim k} = \balpha_1 \circ \dots \circ \balpha_{k-1} \circ \balpha_{k+1} \circ \dots \circ \balpha_K$. 
This allows us to rewrite \eqref{eq:hoff_opt.1} as
\begin{align}
	\nabla_{\balpha_k} \left\{ \balpha_k' \left(\bY_0' \bY_0 \circ \balpha_{\sim k} \balpha_{\sim k}' + \frac{\lambda}{K} \bI \right) \balpha_k - 2 \balpha_k' (\balpha_{\sim k} \circ \bY_0' \by_T)  \right\} = 0. 
\end{align} 
This is quadratic in $\balpha_k$ for fixed $\balpha_{\sim k}$. 
Thus, the unique minimizer at convergence is 
\begin{align} \label{eq:hpp.2} 
	\bhalpha_k &= \left(\bY_0' \bY_0 \circ \bhalpha_{\sim k} \bhalpha_{\sim k}' + \frac{\lambda}{K} \bI \right)^{-1} (\bhalpha_{\sim k} \circ \bY_0' \by_T),
\end{align} 
where $\bhalpha_{\sim k} = \bhalpha_1 \circ \dots \circ \bhalpha_{k-1} \circ \bhalpha_{k+1} \circ \dots \circ \bhalpha_K$. 
Leveraging properties of the Hadamard product noted in \cite{styan}, we rewrite 
\begin{align}
	\bY_0' \bY_0 \circ \bhalpha_{\sim k} \bhalpha_{\sim k}' &= \bD(\bhalpha_{\sim k})  \bY_0' \bY_0 \bD(\bhalpha_{\sim k})
	\\ 
	\bY_0' \by_T\circ \bhalpha_{\sim k} &= \bD(\bhalpha_{\sim k}) \bY_0' \by_T,
\end{align} 
where $\bD(\bhalpha_{\sim k})$ is the diagonal matrix formed from $\bhalpha_{\sim k}$. 
Leveraging these equalities, we simplify \eqref{eq:hpp.2} as 
\begin{align}
	\bhalpha_k &= 
	\Big( \bD(\bhalpha_{\sim k})  \bY_0' \bY_0 \bD(\bhalpha_{\sim k}) + \frac{\lambda}{K} \bI \Big)^{-1} \bD(\bhalpha_{\sim k}) \bY_0' \by_T. 
\end{align} 
%
%
We now turn to \VT~regression. 
Following the arguments above, for every $k \in [K]$,
\begin{align}
	\bhbeta_k &= 
	\Big( \bD(\bhbeta_{\sim k})  \bY_0 \bY_0' \bD(\bhbeta_{\sim k}) + \frac{\lambda}{K} \bI \Big)^{-1} \bD(\bhbeta_{\sim k}) \bY_0 \by_N, 
\end{align} 
where $\bhbeta_{\sim k} = \bhbeta_1 \circ \dots \circ \bhbeta_{k-1} \circ \bhbeta_{k+1} \circ \dots \circ \bhbeta_K$ 
and $\bD(\bhbeta_{\sim k})$ is the diagonal matrix formed from $\bhbeta_{\sim k}$. 
This completes the proof. 
\end{proof}

	\section{Proofs for Inference} \label{sec:proofs_inference} 

\subsection{Proof of Theorem~\ref{thm:inference}} \label{sec:proofs_inference} 
To establish Theorem~\ref{thm:inference}, we first state a few useful results.
\begin{lemma} [Theorem 2.7.1 of \cite{lehmann}] \label{thm:lehmann}  
Let $X_i$ for $i=1, \dots, n$ be independently distributed with means $\Ex[X_i] = \zeta_i$ and variances $\sigma^2_i$, and with finite third moments. 
Let $\bar{X} = (1/n) \sum_{i=1}^n X_i$.
Then 
\begin{align}
	\frac{\bar{X} - \Ex[\bar{X}]}{\Var(\bar{X})^{1/2}} \xrightarrow{d} \mathcal{N}(0,1), 
\end{align} 
provided
\begin{align}
	\left(  \sum_{i=1}^n \Ex \left[ | X_i - \zeta_i |^3 \right] \right)^2 = o\left( \Big( \sum_{i=1}^n \sigma_i^2 \Big)^3 \right). 
\end{align} 
\end{lemma} 

\begin{lemma} \label{lemma:trace}
Consider a random vector $\bx$ and random matrix $\bA$. 
Let $\Ex[\bx | \bA] = \bzero$ and $\Cov(\bx | \bA) = \bSigma$.  
Then $\Ex[ \bx' \bA \bx | \bA] = \tr(\bA \bSigma)$. 
\end{lemma} 

\begin{proof}
%
%
(i) [\HZ~model] Let Assumptions~\ref{assump:hz.outcome}--\ref{assump:hz.errors} hold. 
By \eqref{eq:lyapunov}, Lemma~\ref{thm:lehmann} yields   
\begin{align}
	\frac{ \hY_{NT}(0) - \Ex[\hY_{NT}(0) |\by_N, \bY_0] }{\Var(\hY_{NT}(0) | \by_N, \bY_0)^{1/2}} \xrightarrow{d} \mathcal{N}(0,1). 
\end{align} 
To evaluate $\Ex[\hY_{NT}(0) | \by_N, \bY_0]$, we first observe that 
\begin{align}
	\Ex[\hY_{NT}(0) | \by_N, \bY_0] 
	&= \Ex[ \langle \by_N, \bhalpha \rangle | \by_N, \bY_0] 
	\\&= \Ex[ \langle \by_N, \bY_0^\dagger \by_T \rangle | \by_N, \bY_0] 
	\\&= \by'_N \bY_0^\dagger \Ex[ \by_T | \by_N, \bY_0]
	\\&= \by'_N \bY_0^\dagger \bY_0 \balpha^*
	\\&= \by'_N \bH^v \balpha^*. \label{eq:hz.exp.0} 
\end{align}
%
%
Moving to the variance term, we note that 
\begin{align}
	\Var(\hY_{NT}(0) |\by_N, \bY_0) &= \by'_N \Cov(\bhalpha | \by_N, \bY_0) \by_N.
	\label{eq:inf.2.0} 
\end{align} 
Towards evaluating the above, we note that 
\begin{align}
	\Cov(\bhalpha | \by_N, \bY_0) 
	&= \Cov( \bY_0^\dagger \by_T | \by_N, \bY_0)
	\\&=  \bY_0^\dagger \Cov(\by_T | \by_N, \bY_0) (\bY'_0)^\dagger
	\\&= \bY_0^\dagger \Cov(\bvarepsilon_T | \by_N, \bY_0) (\bY'_0)^\dagger
	\\&= \bY_0^\dagger \bSigma^\hz_T (\bY'_0)^\dagger. \label{eq:inf.2.1} 
\end{align} 
Plugging \eqref{eq:inf.2.1} into \eqref{eq:inf.2.0}, we obtain 
\begin{align}
	\Var(\hY_{NT}(0) |\by_N, \bY_0) = \by_N' \bY_0^\dagger \bSigma^\hz_T (\bY'_0)^\dagger \by_N 
	= \bhbeta' \bSigma^\hz_T \bhbeta, 
\label{eq:inf.2}
\end{align} 
where we recall that $\bhbeta = (\bY_0')^\dagger \by_N$. 
Putting it all together, we conclude 
\begin{align}
	\frac{ \hY_{NT}(0) -  \langle \by_N, \bH^v \balpha^* \rangle }{(\bhbeta' \bSigma^\hz_T \bhbeta)^{1/2}} \xrightarrow{d} \mathcal{N}(0,1). 
\end{align} 
%

\smallskip 
(ii) [\VT~model] Let Assumptions~\ref{assump:vt.outcome}--\ref{assump:vt.errors} hold. 
Following the arguments above, we have 
\begin{align}
	\frac{ \hY_{NT}(0) - \langle \by_T, \bH^u \bbeta^* \rangle }{(\bhalpha' \bSigma^\vt_N \bhalpha)^{1/2}} \xrightarrow{d} \mathcal{N}(0,1). 
\end{align} 
%

\smallskip 
(iii) [Mixed model] Let Assumptions~\ref{assump:mixed.outcome}--\ref{assump:mixed.errors} hold. 
We will find it useful to write  
\begin{align}
	\hY_{NT}(0) &= \by'_N \bY_0^\dagger \by_T 
	= \sum_{i \le N_0} \sum_{t \le T_0} (\bY_0^\dagger)_{it} Y_{iT} Y_{Nt}. \label{eq:genvar.1} 
\end{align} 
By Assumption~\ref{assump:mixed.errors}, \eqref{eq:genvar.1} is a sum of independent random variables with 
\begin{align}
	\Ex[Y_{iT} Y_{Nt} | \bY_0] &= \Ex[Y_{iT} | \bY_0] \Ex[Y_{Nt} | \bY_0]
	\\ 
	\Var(Y_{iT} Y_{Nt} | \bY_0) &= \Ex[Y_{iT} | \bY_0]^2 \sigma^2_{Nt} 
	+ \Ex[Y_{Nt} | \bY_0]^2 \sigma^2_{iT} + \sigma^2_{iT} \sigma^2_{Nt}.  
\end{align} 
Lemma~\ref{thm:lehmann} then establishes that   
\begin{align}
	\frac{ \hY_{NT}(0) - \Ex[ \hY_{NT}(0) | \bY_0]}{ \Var(\hY_{NT}(0) | \bY_0)^{1/2}} \xrightarrow{d} \mathcal{N}(0,1). 
\end{align} 
Our aim is to evaluate $\Ex[\hY_{NT}(0) | \bY_0]$ and $\Var(\hY_{NT}(0) | \bY_0)$. 
Towards the former, we use Assumptions~\ref{assump:mixed.outcome}--\ref{assump:mixed.errors} with the law of total expectation to obtain  
\begin{align}
	\Ex[\hY_{NT}(0) | \bY_0] 
	&= \Ex \left[ \Ex[ \langle \by_N, \bY_0^\dagger \by_T \rangle | \bvarepsilon_N, \bY_0] | \bY_0 \right]
	\\ &= \Ex \left[ \Ex[ \by_N' \bY_0^\dagger (\bY_0 \balpha^* + \bvarepsilon_T) | \bvarepsilon_N,  \bY_0] | \bY_0 \right] 
	\\ &= \Ex \left[ (\bY'_0 \bbeta^* + \bvarepsilon_N)' \bY_0^\dagger \bY_0 \balpha^* | \bY_0 \right] 
	\\ &= \langle \bbeta^*, \bY_0 \balpha^* \rangle. 
	\label{eq:inf.4} 
\end{align}  
Note that we have used the fact that $\by_N$ is deterministic given $(\bvarepsilon_N, \bY_0)$. 
Similarly, by the law of total variance, 
\begin{align}
	\Var(\hY_{NT}(0) | \bY_0) &= \Ex[ \Var(\hY_{NT}(0) | \bvarepsilon_N, \bY_0) | \bY_0 ] + \Var( \Ex[\hY_{NT}(0) | \bvarepsilon_N, \bY_0] | \bY_0 ).  \label{eq:total_var} 
\end{align} 
Following the derivation of \eqref{eq:inf.2}, we have  
\begin{align}
	&\Ex[ \Var(\hY_{NT}(0) | \bvarepsilon_N, \bY_0) | \bY_0] 
	\\ &\qquad \qquad = \Ex[ \by'_N \bY_0^\dagger \bSigma^\mix_T (\bY'_0)^\dagger \by_N | \bY_0] 
	\\ &\qquad \qquad =  (\bY'_0 \bbeta^*)' \bA (\bY'_0 \bbeta^*) + \Ex[ \bvarepsilon'_N \bA \bvarepsilon_N | \bY_0] 
	+ 2 \Ex[ \bvarepsilon'_N  \bY'_0 \bbeta^* | \bY_0], \label{eq:inf.5} 
\end{align} 
where $\bA = \bY^\dagger_0 \bSigma^\mix_T (\bY'_0)^\dagger$. 
Notice that Assumption~\ref{assump:mixed.errors} gives $\Ex[\bvarepsilon'_N  \bY'_0 \bbeta^* | \bY_0] = 0$. 
Since $\bA$ is deterministic given $\bY_0$, Lemma~\ref{lemma:trace} yields   
\begin{align}
	\Ex[ \bvarepsilon'_N \bA \bvarepsilon_N | \bY_0] &= \tr( \bA \bSigma^\mix_N). \label{eq:tr.2} 
\end{align} 
Following the arguments that led to the derivation of \eqref{eq:hz.exp.0}, we have 
\begin{align}
	\Var( \Ex[\hY_{NT}(0) | \bvarepsilon_N, \bY_0] | \bY_0 ) &= \Var( \by'_N \bH^v \balpha^* | \bY_0) 
	= (\bH^v \balpha^*)' \bSigma^\mix_N (\bH^v \balpha^*). \label{eq:inf.6}
\end{align} 
Plugging \eqref{eq:inf.5}, \eqref{eq:tr.2}, and \eqref{eq:inf.6} into \eqref{eq:total_var}, we arrive at  
\begin{align} 
	&\Var(\hY_{NT}(0) | \bY_0) 
	\\
	&\qquad = (\bH^v \balpha^*)' \bSigma^\mix_N (\bH^v \balpha^*)  
	+ (\bH^u \bbeta^*)' \bSigma^\mix_T (\bH^u \bbeta^*)
	+ \tr( \bY^\dagger_0 \bSigma^\mix_T (\bY'_0)^\dagger \bSigma^\mix_N). 
\end{align} 
This completes the proof. 
\end{proof} 

\subsection{Proofs for Model-Based Confidence Intervals} 
We first state a useful lemma to prove Lemmas~\ref{lemma:inference.1}--\ref{lemma:inference.3}. 
\begin{lemma} \label{lemma:mixed.vest} 
\emph{[Mixed model]} Let Assumptions~\ref{assump:mixed.outcome}--\ref{assump:mixed.errors} hold. 
Then, 
\begin{align}
	\Ex[ \hv^\emix_0 | \bY_0] 
	&=  
	(\bH^u \bbeta^*)' \Ex[ \bhSigma_T | \bY_0] (\bH^u \bbeta^*) 
	+ \tr( \bY_0^\dagger \Ex[ \bhSigma_T | \bY_0] (\bY'_0)^\dagger \bSigma^\emix_N)
	\\ &\quad 
	+ (\bH^v \balpha^*)' \Ex[ \bhSigma_N | \bY_0] (\bH^v \balpha^* )
	+ \tr( \bY_0^\dagger \bSigma^\emix_T  (\bY'_0)^\dagger \Ex[\bhSigma_N | \bY_0])
	\\ &\quad - \tr( \bY_0^\dagger \Ex[ \bhSigma_T | \bY_0] (\bY'_0)^\dagger \Ex[\bhSigma_N | \bY_0]). 
\end{align} 
\end{lemma}


\subsubsection{Proof of Lemma~\ref{lemma:inference.1}} 

\begin{proof}
(i) [\HZ~model] Let Assumptions~\ref{assump:hz.outcome}--\ref{assump:hz.errors} hold. 
Taking note that $\bH^u_\perp \bY_0 = \bzero$, 
\begin{align} 
	\| \bH^u_\perp \by_T \|_2^2 
	&= \by'_T \bH^u_\perp \by_T 
	\\ &= (\bY_0 \balpha + \bvarepsilon_T)' \bH^u_\perp (\bY_0 \balpha + \bvarepsilon_T) 
	\\ &= \bvarepsilon'_T \bH^u_\perp \bvarepsilon_T. 
\end{align} 
Applying Lemma~\ref{lemma:trace} then gives 
\begin{align}
	\Ex[\bvarepsilon'_T \bH^u_\perp \bvarepsilon_T | \by_N, \bY_0]
	= \tr(\bH^u_\perp) (\sigma^\hz_T)^2
	= (N_0 - R) (\sigma^\hz_T)^2, \label{eq:hz.trace_var} 
\end{align} 
where the final equality follows because the trace of a projection matrix equals its rank. 
Taken altogether, we have $\Ex[\bhSigma_T^\homo | \by_N, \bY_0] = \bSigma^\hz_T$.  
Therefore, 
\begin{align}
	\Ex[ \hv_0^{\hz, \homo} | \by_N, \bY_0] 
	&= \bhbeta' \Ex[ \bhSigma^\homo_T | \by_N, \bY_0] \bhbeta 
	= v_0^\hz. 
\end{align} 

(ii) [\VT~model] Let Assumptions~\ref{assump:vt.outcome}--\ref{assump:vt.errors} hold. 
Following the arguments above, we conclude that $\Ex[\bhSigma_N^\homo | \by_T, \bY_0] = \bSigma^\vt_N$ and $\Ex[\hv_0^{\vt, \homo} | \by_T, \bY_0] = v_0^\vt$. 

\medskip 
(iii) [Mixed model] Let Assumptions~\ref{assump:mixed.outcome}--\ref{assump:mixed.errors} hold. 
Following the arguments that led to \eqref{eq:hz.trace_var}, we obtain 
$\Ex[\bhSigma^\homo_T | \bY_0] = \bSigma^\mix_T$ 
and
$\Ex[\bhSigma^\homo_N | \bY_0] = \bSigma^\mix_N$.  
Applying Lemma~\ref{lemma:mixed.vest} then gives $\Ex[ \hv^{\mix, \homo}_0 | \bY_0] = v^\mix_0$. 
The proof is complete. 
\end{proof} 

\subsection{Proof of Lemma~\ref{lemma:inference.2}} 
\label{sec:proof_jackknife}

\begin{proof}
Before we establish the biases of $(\bhSigma^\jack_T, \bhSigma^\jack_N)$, we first justify their forms. 
As noted in Section~\ref{sec:ci}, jackknife is a popular approach to estimate the covariances of $(\bhalpha, \bhbeta)$. 
Below, we follow the standard techniques to derive the jackknife estimate of these objects, which will then be used to derive $(\bhSigma^\jack_T, \bhSigma^\jack_N)$.
Without loss of generality, we begin with $\bhalpha$. 
Notably, while standard derivations consider $\bY_0$ with full column rank, we consider a general matrix $\bY_0$ that may be rank deficient. 
This difference is subtle so the following proof is by no means novel. 
We provide it simply for completeness. 

To describe the jackknife, we define $\bhalpha_{\sim i}$ as the minimum $\ell_2$-norm solution to \eqref{eq:hz.rr.1}, where $\lambda_1 = \lambda_2 = 0$, without the $i$th observation, i.e., 
\begin{align}
	\bhalpha_{\sim i} &= (\bY'_{0, \sim i} \bY_{0, \sim i})^\dagger \bY'_{0, \sim i} \by_{T, \sim i}, \label{eq:pseudo} 
\end{align} 
where $\bY_{0, \sim i}$ and $\by_{T, \sim i}$ correspond to $\bY_0$ and $\by_T$ without the $i$th observation. 
We define the pseudo-estimator as $\btalpha_i = T_0 \bhalpha - (T_0-1) \bhalpha_{\sim i}$.
With these quantities defined, we write the jackknife variance estimator as 
\begin{align}
	\bhV^\jack = \frac{1}{(T_0-1)^2} \sum_{i \le N_0} (\btalpha_i - \bhalpha) (\btalpha_i - \bhalpha)'. \label{eq:jack_var} 
\end{align} 
To evaluate this quantity, we will rewrite $\bhalpha_{\sim i}$ in a more convenient form. 
In particular, 
\begin{align}
	&\bY'_{0, \sim i} \bY_{0, \sim i} = \bY'_0 \bY'_0 - \by_i \by'_i 
	\\ 
	&\bY'_{0, \sim i} \by_{T, \sim i} = \bY'_0 \by_T - \by_i Y_{iT}, 
\end{align} 
where $\by_i = [Y_{it}: t \le T_0]$ is the $i$th row of $\bY_0$. 
We do not assume that $\bY'_0 \bY_0$ is nonsingular. 
As such, we use a generalized form of the Sherman-Morrison formula \citep{cline, meyer} to obtain 
\begin{align}
	(\bY'_{0, \sim i} \bY_{0, \sim i})^\dagger = 
	(\bY'_0 \bY_0)^\dagger + (1 - H^u_{ii})^{-1} (\bY'_0 \bY_0)^\dagger \by_i \by'_i (\bY'_0 \bY_0)^\dagger. \label{eq:sherman_morrison} 
\end{align} 
Recall  $\bhalpha = (\bY'_0 \bY_0)^\dagger \bY'_0 \by_T$ and note $Y_{iT} - \by'_i \bhalpha$ is the $i$th element of $\bhvarepsilon_T = \bH^u_\perp \by_T$. 
Using these facts, we plug \eqref{eq:sherman_morrison} into \eqref{eq:pseudo} to yield  
\begin{align}
	\bhalpha_{\sim i} 
	&= \left[ (\bY'_0 \bY_0)^\dagger + (1 - H^u_{ii})^{-1} (\bY'_0 \bY_0)^\dagger \by_i \by'_i (\bY'_0 \bY_0)^\dagger \right] ( \bY'_0 \by_T - \by_i Y_{iT}) 
	\\ &= \bhalpha 
	- (\bY'_0 \bY_0)^\dagger \by_i Y_{iT} 
	+ (1 - H^u_{ii})^{-1} (\bY'_0 \bY_0)^\dagger \by_i \by'_i \bhalpha 
	- H^u_{ii} (1 - H^u_{ii})^{-1} (\bY'_0 \bY_0)^\dagger \by_i Y_{iT} 
	\\ &= \bhalpha - (1-H^u_{ii})^{-1} (\bY'_0 \bY_0)^\dagger \by_i \hvarepsilon_{iT}. \label{eq:jack.1} 
\end{align}  
Inserting \eqref{eq:jack.1} into our pseudo-estimate, we have 
\begin{align}
	\btalpha_i &= T_0 \bhalpha - (T_0-1) \left( \bhalpha - (1-H^u_{ii})^{-1} (\bY'_0 \bY_0)^\dagger \by_i \hvarepsilon_{iT} \right) 
	\\ &= \bhalpha + (T_0-1)(1-H^u_{ii})^{-1} (\bY'_0 \bY_0)^\dagger \by_i \hvarepsilon_{iT}. \label{eq:jack.2} 
\end{align} 
Inserting \eqref{eq:jack.2} into \eqref{eq:jack_var}, we have 
\begin{align}
	\bhV^\jack &= (\bY'_0 \bY_0)^\dagger  \left( \sum_{i \le N_0} \frac{\hvarepsilon_{iT}^2}{(1-H^u_{ii})^2} \by_i \by'_i \right) (\bY'_0 \bY_0)^\dagger 
	\\ &=  (\bY'_0 \bY_0)^\dagger \bY'_0 \bOmega \bY_0 (\bY'_0 \bY_0)^\dagger, 
\end{align} 	
where $\bOmega$ is a diagonal matrix with $\Omega_{ii} = \hvarepsilon_{iT}^2 (1-H^u_{ii})^{-2}$. 
Equivalently, $\bOmega = \text{diag}( [ \bH^u_\perp \circ \bH^u_\perp \circ \bI ]^\dagger [\bhvarepsilon_T \circ \bhvarepsilon_T])$. 
It then follows that 
\begin{align}
	\by'_N \bhV^\jack \by_N &= 
	 \bhbeta' \bOmega \bhbeta. 
\end{align} 
To arrive at \eqref{eq:hz.var.2}, 
we define $\bhSigma^\jack_T = \bOmega$. 
This corresponds to the EHW estimator with the jackknife correction. 
We derive \eqref{eq:vt.var.2} for $\bhbeta$ by applying the same arguments above. 
Now, we will evaluate the biases of $(\bhSigma^\jack_T, \bhSigma^\jack_N)$. 

(i) [\HZ~model] Let Assumptions~\ref{assump:hz.outcome}--\ref{assump:hz.errors} hold. 
We define $(\sigma^\hz_{iT})^2 = \Var(\varepsilon_{iT} | \by_N, \bY_0)$ for $i = 1, \dots, N_0$. 
Observe that 
\begin{align}
	\Ex[ ( \bH^u_\perp \circ \bH^u_\perp \circ \bI )^\dagger (\bhvarepsilon_T \circ \bhvarepsilon_T) | \by_N, \bY_0] 
	&= ( \bH^u_\perp \circ \bH^u_\perp \circ \bI )^\dagger  \Ex[ \bhvarepsilon_T \circ \bhvarepsilon_T | \by_N, \bY_0].
	\label{eq:var.1} 
\end{align} 
To evaluate \eqref{eq:var.1}, we follow the derivations of \eqref{eq:hz.exp.0} and \eqref{eq:inf.2.1} to obtain 
\begin{align}  
	\Ex[\bhvarepsilon_T | \by_N, \bY_0] &= \bH^u_\perp \bY_0 \balpha^* = \bzero \label{eq:var.2}
	\\ 
	\Cov(\bhvarepsilon_T | \by_N, \bY_0) &= \bH^u_\perp \bSigma^\hz_T \bH^u_\perp.  \label{eq:var.3}
\end{align} 
Recall that $\Ex[X^2] = \Var(X) + \Ex[X]^2$ for any random variable $X$. 
Thus, combining \eqref{eq:var.2} with \eqref{eq:var.3} gives 
\begin{align}
	\Ex[ \bhvarepsilon_T \circ \bhvarepsilon_T | \by_N, \bY_0] &= ( \bH^u_\perp \bSigma^\hz_T \bH^u_\perp \circ \bI ) \bone. \label{eq:var.4} 
\end{align} 
Let $\bhgamma = \Ex[ \bhvarepsilon_T \circ \bhvarepsilon_T | \by_N, \bY_0]$. 
By \eqref{eq:var.4}, the $\ell$th entry of $\bhgamma$ can be written as  
\begin{align}
	\hgamma_\ell = \sum_{j \neq \ell} (H^u_{j\ell})^2  (\sigma^\hz_{jT})^2 +  (1 - H^u_{\ell\ell})^2 (\sigma^\hz_{\ell T})^2,  
\end{align} 
where $H^u_{j\ell}$ is the $(j,\ell)$th entry of $\bH^u$. 
In turn, this allows us to rewrite \eqref{eq:var.4} as 
\begin{align}
	\bhgamma &= ( \bH^u_\perp \circ \bH^u_\perp ) \bSigma^\hz_T \bone. \label{eq:var.5} 
\end{align} 
Next, let $\bhzeta = (\bH^u_\perp \circ \bH^u_\perp \circ \bI )^{-1} \bhgamma$. 
Notice that the $\ell$th entry of $\bhzeta$ is given by 
\begin{align}
	\hzeta_\ell &= (\sigma^\hz_{\ell T})^2 + \sum_{j \neq \ell} \frac{(H^u_{\ell j})^2}{(1-H^u_{\ell \ell})^2} (\sigma^\hz_{jT})^2. 
\end{align} 
Therefore, $\text{diag}( \bhzeta ) = \bSigma^\hz_T + \bDelta^\hz$, where $\Delta^\hz_{\ell \ell} = \sum_{j \neq \ell} (\sigma^\hz_{jT})^2 (H^u_{\ell j})^2 (1 - H^u_{\ell \ell})^{-2}$ for $\ell=1, \dots, N_0$.  
Notice if $\max_\ell H^u_{\ell \ell} < 1$, then $(\bH^u_\perp \circ \bH^u_\perp \circ \bI)$ is nonsingular, i.e., the pseudo-inverse is precisely the inverse. 
In this situation, plugging the above into \eqref{eq:var.1} gives 
\begin{align} 
	\Ex[\bhSigma^\jack_T | \by_N, \bY_0]  
	&= \text{diag}\left((\bH^u_\perp \circ \bH^u_\perp \circ \bI )^{-1} \Ex[ \bhvarepsilon_T \circ \bhvarepsilon_T | \by_N, \bY_0] \right) 
	\\ &= \text{diag}\left((\bH^u_\perp \circ \bH^u_\perp \circ \bI )^{-1} \bhgamma \right)
	\\ &= \text{diag}(\bhzeta) 
	\\ &= \bSigma^\hz_T + \bDelta^\hz. \label{eq:hz.var_jack.0} 
\end{align} 
From this, we conclude that 
\begin{align}
	\Ex[\hv_0^{\hz, \jack} | \by_N, \bY_0]
	&= \bhbeta' \Ex[\bhSigma_T^\jack | \by_N, \bY_0] \bhbeta 
	\\ &= \bhbeta' (\bSigma^\hz_T + \bDelta^\hz) \bhbeta
	\\ &= v_0^\hz + \bhbeta' \bDelta^\hz \bhbeta, 
\end{align} 
where we note that $\bhbeta' \bDelta^\hz \bhbeta \ge 0$. 

\smallskip 
(ii) [\VT~model] Let Assumptions~\ref{assump:vt.outcome}--\ref{assump:vt.errors} hold.
Following the arguments above, we conclude 
$\Ex[\bhSigma^\jack_N |  \by_T, \bY_0] = \bSigma^\vt_N + \bGamma^\vt$, where 
$\Gamma^\vt_{\ell \ell} = \sum_{j \neq \ell} (\sigma^\vt_{Nj})^2 (H^v_{\ell j})^2(1-H^v_{\ell \ell})^{-2}$ for $\ell=1, \dots, T_0$. 
Thus, 
$\Ex[\hv_0^{\vt, \jack} | \by_T, \bY_0] = v^\vt_0 + \bhalpha' \bGamma^\vt \bhalpha$, 
where we note that $\bhalpha' \bGamma^\vt \bhalpha \ge 0$. 

\smallskip 
(ii) [Mixed~model] Let Assumptions~\ref{assump:mixed.outcome}--\ref{assump:mixed.errors} hold. 
We define $(\sigma^\mix_{iT})^2 = \Var(\varepsilon_{iT} | \bY_0)$ for $i = 1, \dots, N_0$ and
$(\sigma^\mix_{Nt})^2 = \Var(\varepsilon_{Nt} | \bY_0)$ for $t = 1, \dots, T_0$. 
Following the arguments that led to \eqref{eq:hz.var_jack.0}, we obtain 
$\Ex[\bhSigma^\jack_T |  \bY_0] = \bSigma^\mix_T + \bDelta^\mix$, 
where $\Delta^\mix_{\ell \ell} = \sum_{j \neq \ell} (\sigma^\mix_{jT})^2 (H^u_{\ell j})^2 (1 - H^u_{\ell \ell})^{-2}$ for $\ell=1, \dots, N_0$. 
Similarly, we obtain 
$\Ex[\bhSigma^\jack_N |  \bY_0] = \bSigma^\mix_N + \bGamma^\mix$, where 
$\Gamma^\mix_{\ell \ell} = \sum_{j \neq \ell} (\sigma^\mix_{Nj})^2 (H^v_{\ell j})^2(1-H^v_{\ell \ell})^{-2}$ for $\ell=1, \dots, T_0$. 
Applying Lemma~\ref{lemma:mixed.vest} then gives 
\begin{align} 
	&\Ex[\hv^{\mix, \jack}_0 | \bY_0] 
	\\ &= v^\mix_0 
	+ (\bH^u \bbeta^*)' \bDelta^\mix (\bH^u \bbeta^*) 
	+ (\bH^v \balpha^*)' \bGamma^\mix (\bH^v \balpha^*)
	+  \tr( \bY_0^\dagger \bDelta^\mix (\bY'_0)^\dagger \bGamma^\mix ).
\end{align}  
The proof is complete. 
\end{proof}

\subsection{Proof of Lemma~\ref{lemma:inference.3}} 

\begin{proof} 
We adopt the strategy of \cite{variance} to prove our desired result. 

\smallskip 
(ii) [\HZ~model] Let Assumptions~\ref{assump:hz.outcome}--\ref{assump:hz.errors} hold.
As in the proof of Lemma~\ref{lemma:inference.2}, we define $\bhvarepsilon_T = \bH^u_\perp \by_T$. 
Observe 
\begin{align}
	\Ex[ ( \bH^u_\perp \circ \bH^u_\perp )^{-1} (\bhvarepsilon_T \circ \bhvarepsilon_T) | \by_N, \bY_0] 
	&= ( \bH^u_\perp \circ \bH^u_\perp )^{-1}  \Ex[ \bhvarepsilon_T \circ \bhvarepsilon_T | \by_N, \bY_0].  \label{eq:var.11} 
\end{align} 
To evaluate \eqref{eq:var.11}, we plug in \eqref{eq:var.5} to obtain 
\begin{align}
	\Ex[ ( \bH^u_\perp \circ \bH^u_\perp )^{-1} (\bhvarepsilon_T \circ \bhvarepsilon_T) |\by_N, \bY_0] 
	&= ( \bH^u_\perp \circ \bH^u_\perp )^{-1}  ( \bH^u_\perp \circ \bH^u_\perp ) \bSigma^\hz_T \bone 
	= \bSigma^\hz_T \bone. \label{eq:var.12} 
\end{align} 
Plugging \eqref{eq:var.12} into \eqref{eq:var.11} yields  
\begin{align}
	\Ex[\bhSigma^\hrk_T | \by_N, \bY_0] 
	&= \text{diag} \left( (\bH^u_\perp \circ \bH^u_\perp )^{-1}  \Ex[ \bhvarepsilon_T \circ \bhvarepsilon_T | \by_N, \bY_0] \right) 
	= \bSigma^\hz_T. \label{eq:hz.var_hrk.0} 
\end{align} 
It then follows that $\Ex[ \hv_0^{\hz, \hrk} | \by_N, \bY_0] = v_0^\hz$. 

\smallskip 
(ii) [\VT~model] Let Assumptions~\ref{assump:vt.outcome}--\ref{assump:vt.errors} hold.
Following the same arguments as above, we conclude
$\Ex[\bhSigma^\hrk_N | \by_T, \bY_0] = \bSigma^\vt_N$ 
and 
$\Ex[ \hv_0^{\vt, \hrk} | \by_T, \bY_0] = v^\vt_0$. 

\smallskip 
(ii) [Mixed model] Let Assumptions~\ref{assump:mixed.outcome}--\ref{assump:mixed.errors} hold.
Following the arguments that led to \eqref{eq:hz.var_hrk.0}, we obtain 
$\Ex[\bhSigma^\hrk_T | \bY_0] = \bSigma^\mix_T$
and 
$\Ex[\bhSigma^\hrk_N | \bY_0] = \bSigma^\mix_N$. 
Applying Lemma~\ref{lemma:mixed.vest} then gives $\Ex[ \hv^{\mix, \hrk}_0 | \bY_0] = v^\mix_0$. 
The proof is complete. 
\end{proof} 

\subsection{Proof of Lemma~\ref{lemma:mixed.vest}} 

\begin{proof}
By linearity of expectations, 
\begin{align}
	\Ex[\hv^\mix_0 | \bY_0] &= \Ex[ \hv_0^\hz | \bY_0]
	+ \Ex[ \hv_0^\vt  | \bY_0]
	- \Ex[ \tr(\bY_0^\dagger \bhSigma_T (\bY'_0)^\dagger \bhSigma_N) | \bY_0]. \label{eq:finalvar.0} 
\end{align} 
We evaluate each term in \eqref{eq:finalvar.0}. 

Beginning with the first term, note that the randomness in $\bhSigma_T$ stems from $\bvarepsilon_T$ and $\bhbeta$ is deterministic given $(\bvarepsilon_N, \bY_0)$. 
As such, Assumptions~\ref{assump:mixed.outcome}--\ref{assump:mixed.errors} with Lemma~\ref{lemma:trace} gives  
\begin{align}
	\Ex[ \hv_0^\hz | \bY_0] 
	&= \Ex[ \bhbeta' \bhSigma_T \bhbeta | \bY_0] 
	\\ &= \Ex\left[ \Ex[ \bhbeta' \bhSigma_T \bhbeta | \bvarepsilon_N, \bY_0] | \bY_0 \right] 
	\\ &= \Ex\left[ \by'_N \bY^\dagger_0 \Ex[ \bhSigma_T | \bY_0] (\bY'_0)^\dagger \by_N | \bY_0 \right] 
	\\ &= \Ex\left[ (\bY'_0 \bbeta^* + \bvarepsilon_N) \bY^\dagger_0 \Ex[ \bhSigma_T | \bY_0] (\bY'_0)^\dagger (\bY'_0 \bbeta^* + \bvarepsilon_N) | \bY_0 \right] 
	\\ &=  (\bH^u \bbeta^*)' \Ex[ \bhSigma_T | \bY_0] (\bH^u \bbeta^*) 
	+ \tr( \bY_0^\dagger \Ex[ \bhSigma_T | \bY_0] (\bY'_0)^\dagger \bSigma^\mix_N). \label{eq:hz_var.final} 
\end{align} 
%
By an analogous argument, we derive 
\begin{align}
	\Ex[ \hv_0^\vt | \bY_0] 
	=  (\bH^v \balpha^*)' \Ex[ \bhSigma_N | \bY_0] (\bH^v \balpha^* )
	+ \tr( \bY_0^\dagger \bSigma^\mix_T  (\bY'_0)^\dagger \Ex[\bhSigma_N | \bY_0]). 
	\label{eq:vt_var.final}
\end{align} 
Finally, we use the linearity of the trace operator with Assumption~\ref{assump:mixed.errors} to obtain 
\begin{align}
	\Ex[ \tr(\bY_0^\dagger \bhSigma_T (\bY'_0)^\dagger \bhSigma_N) | \bY_0]
	&= \Ex\left[ \Ex[ \tr(\bY_0^\dagger \bhSigma_T (\bY'_0)^\dagger \bhSigma_N) | \bvarepsilon_N, \bY_0] | \bY_0 \right] 
	\\ &= \Ex\left[ \tr(\bY_0^\dagger \Ex[\bhSigma_T | \bY_0] (\bY'_0)^\dagger \bhSigma_N) | \bY_0 \right]
	\\ &= \tr(\bY_0^\dagger \Ex[\bhSigma_T | \bY_0] (\bY'_0)^\dagger \Ex[\bhSigma_N | \bY_0 ]). 
	\label{eq:trace_var.final}
\end{align} 
%
Putting everything together completes the proof. 
\end{proof}

\end{appendix} 

\end{document}